\documentclass[a4paper, 11pt]{article}

\usepackage[margin=2cm]{geometry}

\usepackage{rotating}

\usepackage{amsmath,amsthm,amssymb,enumitem,bbm,xcolor,mathtools}
\usepackage{tikz}
\usepackage{etoolbox}
\usepackage{float}
\usepackage{subfig}

\usepackage[final, inline]{showlabels}

\theoremstyle{plain}
\newtheorem{theorem}{Theorem}[subsection]

\newtheorem{corollary}  [theorem]{Corollary}

\newtheorem{example}    [theorem]{Example}
\newtheorem{lemma}      [theorem]{Lemma}

\newtheorem{proposition}[theorem]{Proposition}

\newtheorem{observation}[theorem]{Observation}

\theoremstyle{definition}

\newtheorem*{claim*}{Claim}

\newcommand{\N}{\mathbb{N}}
\newcommand{\Z}{\mathbb{Z}}

\newcommand{\Define}[1]{\textbf{#1}}

\newcommand{\Lattice}[1]{ \mathcal{#1} }
\newcommand{\Opposite}[1]{ {#1}^{\mathrm{op}} }

\newcommand{\Powerset}{ P }

\newcommand{\meet}      {\wedge}
\newcommand{\join}      {\vee}
\newcommand{\bigmeet}   {\bigwedge}
\newcommand{\bigjoin}   {\bigvee}

\newcommand{\Reflexive}             [1]{ {#1}^\mathrm{R} }
\newcommand{\Transitive}            [1]{ {#1}^\mathrm{T} }
\newcommand{\ReflexiveTransitive}   [1]{ {#1}^\mathrm{RT} }
\newcommand{\Symmetric}             [1]{ {#1}_\mathrm{S} }

\newcommand{\Residual}              [1]{ {#1}_+ }
\newcommand{\Residuated}            [1]{ {#1}_- }

\newcommand{\Dual}                  [1]{ {#1}_* }
\newcommand{\Transpose}             [1]{ {#1}_\top }
\newcommand{\CoTranspose}             [1]{ {#1}_\perp }

\newcommand{\Graphs}    { \Lattice{G} }

\newcommand{\one}       { \mathbf{1} }
\newcommand{\zero}      { \mathbf{0} }

\newcommand{\id}        {\mathrm{id}}
\newcommand{\Fix}       {\mathrm{Fix}}
\newcommand{\Pyr}       {\mathrm{Pyr}}
\newcommand{\Par}       {\mathrm{Par}}

\newcommand{\Converges}{\to}
\newcommand{\StronglyConverges}{\twoheadrightarrow}

\newcommand{\MeetSequence}{ \mathrm{M} }
\newcommand{\JoinSequence}{ \mathrm{J} }

\setlist[enumerate]{label=(\alph*)}

\newlist{equivalent}{enumerate}{1}
\setlist[equivalent]{label=(\arabic*)}

\usepackage{moreenum}

\newlist{dependency}{enumerate}{1}
\setlist[dependency]{label=(\greek*)}

\title{Robert's theorem and graphs on complete lattices}
\author{Maximilien Gadouleau\thanks{Department of Computer Science, Durham University, UK. Email: \texttt{m.r.gadouleau@durham.ac.uk} }}
\date{\today}

\begin{document}

\maketitle


\begin{abstract}
Automata networks, and in particular Boolean networks, are used to model diverse networks of interacting entities. The interaction graph of an automata network is its most important parameter, as it represents the overall architecture of the network. A continuous amount of work has been devoted to infer dynamical properties of the automata network based on its interaction graph only. Robert's theorem is the seminal result in this area; it states that automata networks with an acyclic interaction graph converge to a unique fixed point. The feedback bound can be viewed as an extension of Robert's theorem; it gives an upper bound on the number of fixed points of an automata network based on the size of a minimum feedback vertex set of its interaction graph. Boolean networks can be viewed as self-mappings on the power set lattice of the set of entities. In this paper, we consider self-mappings on a general complete lattice. We make two conceptual contributions. Firstly, we can view a digraph as a residuated mapping on the power set lattice; as such, we define a graph on a complete lattice as a residuated mapping on that lattice. We extend and generalise some results on digraphs to our setting. Secondly, we introduce a generalised notion of dependency whereby any mapping $\phi$ can depend on any other mapping $\alpha$. In fact, we are able to give four kinds of dependency in this case. We can then vastly expand Robert's theorem to self-mappings on general complete lattices; we similarly generalise the feedback bound. We then obtain stronger results in the case where the lattice is a complete Boolean algebra. We finally show how our results can be applied to prove the convergence of automata networks.
\end{abstract}

\section{Introduction} \label{section:introduction}

\subsection{Automata networks and Robert's theorem} \label{subsection:background}

Automata networks are used to represent different networks of interacting entities, such as neural networks, gene regulatory networks, or social interactions (see \cite{GRR15} and references therein). In such networks, each entity $i \in V$ has a state $x_i$ taking its values in some alphabet. Usually, this alphabet $Q$ is finite and is the same for all entities, so that the overall configuration of the network is $x \in Q^V$; if $Q = \{0,1\}$, this is called a Boolean network. The value of $x_i$ then evolves over time according to a deterministic function $\phi_i$, which depends on the state of several entities in the network. As such, the whole network is formalised as a mapping $\phi : Q^V \to Q^V$ which represents the evolution of the configuration over time.

A main problem in the theory of automata networks is to infer properties on the dynamics of the configuration $x$ based only on partial knowledge of the mapping $\phi$. Amongst the numerous dynamical properties of interest, such as the number of periodic points of $\phi$ or its transient length, the number of fixed points of $\phi$ is the most thoroughly studied. In terms of partial knowledge of $\phi$, the interaction graph of an automata network represents the overall architecture of the network, where each arc represents the influence of an entity on another. Formally, it is a digraph $D = (V,E)$ where $(i,j)$ is an arc whenever $\phi_j$ depends essentially on $x_i$, i.e.
\begin{equation} \label{equation:interaction_graph}
    \exists a,b \in Q^V : a_k = b_k \forall k \ne i, \phi_j(a) \ne \phi_j(b).
\end{equation}
A continuous amount of work has been devoted to infer dynamical properties based on the interaction graph only (see \cite{Gad20} for a review). We shall focus on two main results on the number of fixed points. Firstly, Robert's theorem \cite{Rob80, Rob95} is the seminal result in the study of automata networks. It states that automata networks with an acyclic interaction graph converge to a unique fixed point. Secondly, the feedback bound \cite{Rii07, Ara08} can be viewed as an extension of Robert's theorem as it gives an upper bound on the number of fixed points of an automata network based on the size of a minimum feedback vertex set of its interaction graph.

\subsection{Robert's theorem, paradoxes and hypodoxes}

On top of their applications, automata networks have also been used in theoretical work, such as network coding \cite{Rii07, Rii07a, GR11, GRF16}, hat games \cite{GG15, Gad18a}, or logic. Let us briefly explain the connection to logic, as it was used to expand on Robert's theorem to the case of infinite Boolean networks; the details can be found in \cite{RRM13}. 

The basic idea is to identify any set of propositional sentences with a Boolean network, where $Q = \{0,1\}$ and $x_i$ is the truth value of sentence $i$ for all $i$. A valid truth assignment to these sentences then corresponds to a fixed point of the Boolean network. Thus, a Boolean with no fixed points is a paradox, while a Boolean network with multiple fixed points corresponds to a hypodox. The interaction graph of the Boolean network represents how the sentences refer to one another; a cycle in the interaction graph is then a case of self-reference. Thus, Robert's theorem implies that a paradox or a hypodox involving a finite number of sentences must use self-reference. 

For instance, consider the well-known paradox:
\begin{quote}
    The next sentence is true. \\
    The previous sentence is false.
\end{quote}
This can be viewed as a Boolean network with two variables $x_1, x_2 \in \{0,1\}$ describing the truth value of each sentence, and two Boolean functions $\phi_1(x) = x_2, \phi_2(x) = \neg x_1$ expressing the actual sentences. The interaction graph of this Boolean network is simply a cycle of length two. Since this Boolean network has no fixed points, this is a paradox. Similarly, consider this hypodox:
\begin{quote}
    The next sentence is true. \\
    The previous sentence is true.
\end{quote}
This can be viewed as a Boolean network with two variables $x_1, x_2 \in \{0,1\}$, and two Boolean functions $\phi_1(x) = x_2, \phi_2(x) = x_1$. Again, its interaction graph is a cycle of length two. It has two fixed points, namely $x = 11$ (both sentences are true) and $x = 00$ (both sentences are false).

However, the situation is different when considering infinite Boolean networks. Yablo \cite{Yab93} discovered his famous paradox without self-reference with a sentence for each natural number:
\begin{quote}
    Sentence $i$: ``For all $j > i$, Sentence $j$ is false.''
\end{quote}
In our terminology, this corresponds to the Boolean network with update function $\phi_i(x) = \bigmeet_{j > i} \neg x_j$ for every $i \in \N$. This Boolean network has no fixed points, and hence it is a paradox, even though its interaction graph is acyclic. Another difference with the finite case is the definition of dependency. Take this simple Boolean network as an example: $\omega: \{ 0, 1 \}^\N \to \{0,1\}^\N$ where $\omega_i(x) = 1$ if and only if $x$ has infinite support for all $i \in \N$. According to Equation \eqref{equation:interaction_graph}, the interaction graph of $\omega$ is empty: changing the value of a single variable does not affect the finiteness of the overall support. Rabern, Rabern and Macauley \cite{RRM13} use a workaround to the interaction graph issue, whereby a function is independent of a set of variables; this shall be reviewed in the sequel. They then prove that non-self-referential paradoxes must contain an infinite number of sentences that involve an infinite number of variables. They also fully classify the graphs that allow hypodoxes.

\subsection{Contributions and outline} \label{subsection:outline}

In this paper, we expand Robert's theorem and the feedback bound to even wider domains than infinite Boolean networks. We view a configuration as $x \in L$, and we consider self-mappings $\phi: L \to L$, where $\Lattice{L} = (L, \join, \meet, \zero, \one)$ is a complete lattice. Boolean networks then reduce to the case where $\Lattice{L}$ is the lattice of subsets $\Lattice{\Powerset}(V)$. Robert's theorem asserts that a finite automata network with an acyclic interaction graph converges to a single fixed point. Accordingly, we generalise the concepts of graph, interaction graph (via the notion of dependency), acyclic graph, and we obtain different convergence results. We give a summary of our contributions leading to the generalised Robert's theorem below; we simplify some technical terms which will be fully explained in the following sections of the paper.

\begin{itemize}
\item 
Firstly, we generalise the concept of graph to any complete lattice. A digraph $D$ naturally induces a mapping on $\Lattice{\Powerset}(V)$, whereby $f(x) = N_D(x)$ is the in-neighbourhood of $x$. This mapping preserves arbitrary unions, i.e. $f(\bigcup_{x \in S} x) = \bigcup_{x \in S} f(x)$; equivalently, $f$ is a residuated mapping \cite{BJ72}, and in fact all residuated self-mappings of $\Lattice{\Powerset}(V)$ are of this form. In our general setting, we then define a graph as any residuated mapping on $\Lattice{L}$. In Section \ref{section:graphs}, we investigate the properties of graphs on lattices and we generalise some well known results on digraphs to our setting, such as the topological sort of acyclic graphs or the finite topology theorem. This solidifies our intuition that this is an appropriate way to generalise graphs. 

\item
Secondly, we widely generalise the notion of dependency given in \cite{RRM13} into four kinds of dependency. One crucial aspect is that a mapping $\phi$ can depend on any other mapping $\alpha : L \to L$, even if $\alpha$ is not a graph. Those four kinds of dependency are related to each other but they yield slightly different results.

\item
Thirdly, we generalise the notion of finite acyclic graph in several ways. In \cite{RRM13}, it is proved that the graphs on which hypodoxes depend are exactly the digraphs without infinite walks. As such, we first introduce the notion of a fixed point-free graph in Section \ref{subsection:fixed_point-free}, which generalises such digraphs to all complete lattices. However, this notion is not strong enough to generalise Robert's theorem, and hence we also introduce asymptotically meet-nilpotent graphs in Section \ref{subsection:nilpotent_graphs}, which will be appropriate when the lattice is a complete Boolean algebra. In general, Robert's theorem can be viewed as a consequence of the Banach contraction principle, where the interaction graph of $\phi$ induces a metric for which the mapping $\phi$ is contractive. Therefore, in Section \ref{section:discriminating}, we introduce so-called covering/metric/complete mappings $\alpha$, which allow to define a metric of varying strength for which the mapping $\phi$ shall be contractive.

\item 
Fourthly, we then prove a generalisation of Robert's theorem in Theorem \ref{theorem:generalised_robert}. In fact, there are four versions, one for each kind of dependency, and each version gives three levels, depending on the properties of the mapping $\alpha$ on which $\phi$ depends (e.g. metric or complete). One version can be informally summarised as: if $\phi$ depends on the complete mapping $\alpha$, then $\phi$ converges to a unique fixed point $e$. An interesting feature of Theorem \ref{theorem:generalised_robert} is that it gives an actual formula for the unique fixed point $e$. In Theorem \ref{theorem:feedback_bound}, we then generalise the feedback bound to mappings over complete lattices; this is based on metric mappings.
\end{itemize}

We then focus on the case where the lattice $\Lattice{L}$ is a complete Boolean algebra. We first introduce the transpose of a graph, which generalises the notion of reversing the direction of every arc in a digraph. The generalised Robert's theorem simplifies greatly when $\phi$ depends on a graph over a complete Boolean algebra. In that case, all four kinds of dependency are equivalent, and a graph is covering/metric/complete if and only if its transpose is asymptotically meet-nilpotent. As such, we obtain a stronger (and more succinct) Robert's theorem for complete Boolean algebras in Theorem \ref{theorem:robert_CBA}. Fixed point-free graphs play a crucial role here: based on them we strengthen the classification of graphs of hypodoxes in \cite{RRM13} and generalise it to all complete Boolean algebras, and the feedback bound for complete Boolean algebras depends on fixed point-free graphs instead of metric mappings.


Our results are generalisations of the original Robert's theorem. Even though most of our work is devoted to infinite lattices, where nontrivial problems of convergence arise, our results can also be used to infer the convergence of finite automata networks. In fact, we show a converse of Robert's theorem in the finite case: any mapping $\phi$ that converges to a single fixed point depends on a nilpotent graph. We also highlight some concrete examples of mappings whose convergence can be proved using our generalised Robert's theorem but which fail to satisfy the hypothesis of the original Robert's theorem.


The rest of the paper is organised as follows. In Section \ref{section:preliminaries}, we review some useful concepts about the Banach contraction principle, Robert's theorem, and complete lattices. In Section \ref{section:graphs}, we introduce graphs as residuated mappings on complete lattices and generalise some graph-theoretic results to this case. We then focus on fixed point-free and asymptotically meet-nilpotent graphs. Section \ref{section:robert} then deals with Robert's theorem for complete lattices. We first introduce the four kinds of dependency in Section \ref{section:dependency} and covering/metric/complete mappings in Section \ref{section:discriminating}. We then give Robert's theorem for complete lattices in Section \ref{section:robert_complete_lattices} and discuss its tightness in Section \ref{subsection:limitations_robert}. We finally give the generalised feedback bound in Section \ref{Section:feedback_bounds}. Section \ref{section:CBA} is then dedicated to Robert's theorem and the feedback bound for complete Boolean algebras. We provide some applications of our results to automata networks in Section \ref{subsection:applications_robert} and conclude the paper in Section \ref{section:conclusion}.

\section{Preliminaries} \label{section:preliminaries}

\subsection{Contraction principles} \label{subsection:contraction_principles}

We denote the set of natural numbers as $\N = \{0, 1, 2, \dots \}$ and for all $j \in \N$, we denote $\N_j = \{ i \in \N: i \ge j \}$. Let $X$ be a set and $\phi: X \to X$. We consider the following sets of points of $X$ with respect to $\phi$.
\begin{enumerate}
    \item \label{item:fixed_points}
    A point $x \in X$ is a \Define{fixed point} of $\phi$ if $\phi(x) = x$. The set of fixed points of $\phi$ is denoted by $\Fix( \phi )$.


    \item \label{item:parabolic_points}
    A point $x \in X$ is a \Define{parabolic point} of $\phi$ if there exists a sequence $(x_i : i \in \N)$ with $x_0 = x$ and $\phi( x_i ) = x_{i-1}$ for all $i \ge 1$. The set of parabolic points of $\phi$ is denoted by $\Par( \phi )$.

    \item \label{item:pyramidal_points}
    A point $x \in X$ is a \Define{pyramidal point} of $\phi$ if there exists a sequence $(x_i : i \in \N)$ with $\phi^i( x_i ) = x$ for all $i \ge 1$. The set of pyramidal points of $\phi$ is denoted by $\Pyr( \phi )$. We then have
    \[
        \Pyr(\phi) = \bigcap_{i \in \N} \phi^i( X ).
    \]
\end{enumerate}
We then have
\[
    \Fix( \phi ) \subseteq \Par( \phi ) \subseteq \Pyr( \phi ).
\]
Those special points are illustrated in Figure \ref{figure:points}.

\begin{figure}[!htp]
    \begin{tikzpicture}[xscale=3, yscale=1]
    
    \begin{scope}[yshift = 0cm]
        \node (fix) at (1,0) {\textcolor{blue}{Fixed point}};
        \node[draw, circle, blue] (f) at (0,0) {$x$};
        \draw[-latex] (f) to [loop left] (f);
    \end{scope}
    
    \begin{scope}[yshift = -2cm]
        \node (par) at (1,0) {\textcolor{violet}{Parabolic point}};
        \node[draw, circle, violet] (p0) at (0,0) {$x$};
        \node[draw, circle] (p1) at (-1,0) {$x_1$};
        \node[draw, circle] (p2) at (-2,0) {$x_2$};
        \node[draw, circle] (p3) at (-3,0) {$x_3$};
        \node (p4) at (-4,0) {$\dots$};
    
        \draw[-latex] (p4) -- (p3);
        \draw[-latex] (p3) -- (p2);
        \draw[-latex] (p2) -- (p1);
        \draw[-latex] (p1) -- (p0);
    \end{scope}

    \begin{scope}[yshift = -4cm]
        \node (par) at (1,0) {\textcolor{red}{Pyramidal point}};
        \node[draw, circle, red] (p00) at (0,0) {$x$};
        \node[draw, circle] (p11) at (-1,0) {$x_1$};
        \node[draw, circle] (p21) at (-1,-1) {};
        \node[draw, circle] (p22) at (-2,-1) {$x_2$};
        \node[draw, circle] (p31) at (-1,-2) {};
        \node[draw, circle] (p32) at (-2,-2) {};
        \node[draw, circle] (p33) at (-3,-2) {$x_3$};
        \node (p4) at (-1,-3) {$\vdots$};
        \node[draw, circle] (pii) at (-4,-4) {$x_i$};
        \node (pii1) at (-3,-4) {$\dots$};
        \node[draw, circle] (pi2) at (-2,-4) {};
        \node[draw, circle] (pi1) at (-1,-4) {};
        \node (p5) at (-1, -5) {$\vdots$};
        
        \draw[-latex] (pii) -- (pii1);
        \draw[-latex] (pii1) -- (pi2);
        \draw[-latex] (pi2) -- (pi1);
        \draw[-latex] (pi1) -- (p00);
    
        \draw[-latex] (p33) -- (p32);
        \draw[-latex] (p32) -- (p31);    
        \draw[-latex] (p31) -- (p00);
        \draw[-latex] (p22) -- (p21);
        \draw[-latex] (p21) -- (p00);
        \draw[-latex] (p11) -- (p00);
    \end{scope}
    
    \end{tikzpicture}
    \caption{Special points of $X$ with respect to $\phi$.}
        \label{figure:points}
\end{figure}
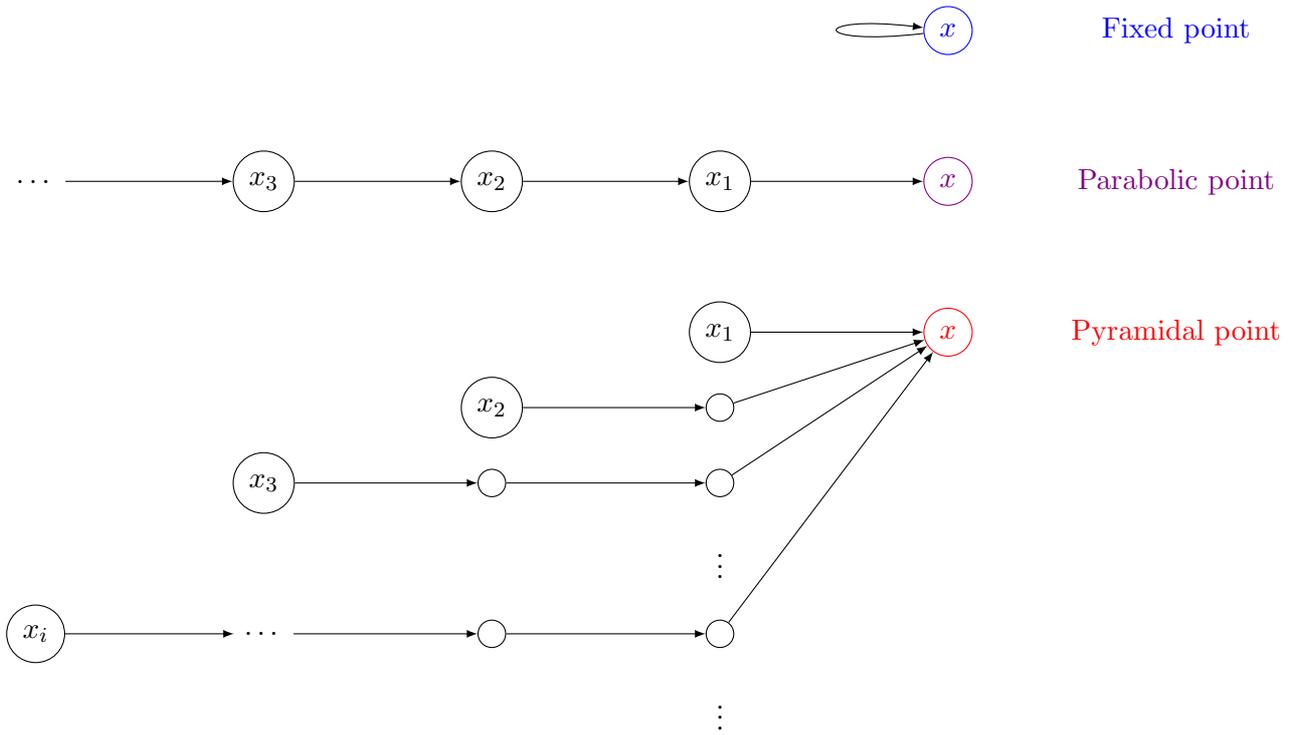

It is readily seen that $\phi( \Fix( \phi ) ) = \Fix( \phi )$. The same holds for the set of parabolic points, as seen below.

\begin{lemma} \label{lemma:phi(Par)}
For any $\phi : X \to X$, we have $\phi( \Par( \phi ) ) = \Par( \phi )$.
\end{lemma}

\begin{proof}
We first prove that $\phi( \Par( \phi ) ) \subseteq \Par( \phi )$. Let $y \in \phi( \Par( \phi ) )$, i.e. $y = \phi(x)$ with $x_0 = x$ and $\phi( x_i ) = x_{i-1}$ for all $i \ge 1$. Define $y_0 = y$ and $y_i = x_{i-1}$ for all $i \ge 1$, then $\phi( y_i ) = y_{i-1}$ for all $i \ge 1$; this shows that $y \in \Par( \phi )$. 

Conversely, we now prove that $\Par( \phi ) \subseteq \phi( \Par( \phi ) )$. If $x \in \Par( \phi )$ with $x_0 = x$ and $\phi( x_i ) = x_{i-1}$ for all $i \ge 1$, then let $z = x_1$. Define $z_0 = z$ and $z_i = x_{i+1}$ for all $i \ge 1$, then $\phi( z_i ) = z_{i-1}$ for all $i \ge 1$, hence $z \in \Par( \phi )$ and $x = \phi( z ) \in \phi( \Par( \phi ) )$.
\end{proof}

We have $\phi( \Pyr( \phi ) ) \subseteq \Pyr( \phi )$ for all $\phi$, but we do not necessarily have the reverse containment, as we shall see in Example \ref{example:stairways}.

\medskip

We can now express three contraction prinicples. Let $( X, d )$ be a metric space and $\phi : X \to X$. All limits shall be taken with respect to $d$. For any $c < 1$, $\phi$ is $c$-contractive if $d( \phi(x), \phi(y) ) \le c d( x, y )$, and $\phi$ is \Define{contractive} if it is $c$-contractive for some $c < 1$. Also, $\phi$ \Define{converges} to $e \in X$ if for all $x \in X$, $\lim_{n \to \infty} \phi^n(x) = e$, which we denote by $\phi \Converges e$. Clearly, if a contractive mapping $\phi$ converges to $e$, then $\Fix(\phi) = \{ e \}$. The Banach contraction principle \cite[Chapter 1, Theorem 1.1]{GD03} asserts that a contractive mapping on a complete metric space converges to a unique fixed point -- to contrast with our next result, we shall refer to it as the complete Banach contraction principle.

\begin{theorem}[Complete Banach contraction principle] \label{theorem:complete_Banach}
Let $(X,d)$ be a complete metric space, and let $\phi: X \to X$ be contractive. Then there exists $e \in X$ such that $\Fix( \phi ) = \{ e \}$, and $\phi \Converges e$.
\end{theorem}

The complete Banach contraction principle does not assert anything about the parabolic points of $\phi$. Indeed, if $X = \mathbb{R}$ endowed with the usual metric $d(x,y) = |x - y|$, then the mapping $\phi: x \mapsto \frac{x}{2}$ is contractive, but all real numbers are parabolic points (with $x_i = 2^i x$ for all $i \in \N$).

We say $\phi$ \Define{strongly converges} to $e \in X$ if $\phi$ converges to $e$ and $\Pyr( \phi ) = \Fix( \phi ) = \{ e \}$, and we denote it by $\phi \StronglyConverges e$. A metric space $(X,d)$ is \Define{bounded} if there exists $M$ such that $d(x,y) \le M$ for all $x, y \in X$. The bounded Banach contraction principle \cite{Jac00} shows that if $\phi$ has a pyramidal point, then it strongly converges to it.

\begin{theorem}[Bounded Banach contraction principle] \label{theorem:bounded_Banach}
Let $(X,d)$ be a bounded metric space, and let $\phi: X \to X$ be contractive. Then there exists $e \in X$ such that $\Pyr( \phi ) = \Fix( \phi ) \subseteq \{ e \}$, and if $\Fix( \phi ) = \{ e \}$, then $\phi \StronglyConverges e$.
\end{theorem}


Combining the previous two contraction principles, we obtain the following corollary.

\begin{corollary}[Complete bounded Banach contraction principle] \label{corollary:complete_bounded_Banach}
Let $(X,d)$ be a complete bounded metric space, and let $\phi: X \to X$ be contractive. Then there exists $e \in X$ such that $\Pyr( \phi ) = \Fix( \phi ) = \{ e \}$, and $\phi \StronglyConverges e$.
\end{corollary}

\subsection{Robert's theorem for finite automata networks} \label{subsection:robert}

Let $X$ be a Cartesian product of sets indexed by a set $V$, i.e. $X = \prod_{v \in V} X_v$. We shall introduce shorthand notations for any subset $I$ of $V$, e.g. $X_I = \prod_{i \in I} X_i$, $x_I = (x_i : i \in I)$, etc. Then an \Define{automata network} is a mapping $\phi: X \to X$; a \Define{finite automata network} is an automata network where $V$ is finite; a \Define{Boolean network} is an automata network where $|X_v| = 2$ for all $v \in V$. We can decompose any automata network $\phi: X \to X$ as a family $\phi = (\phi_v : X \to X_v)$, where $\phi(x)_v = \phi_v(x)$ for all $x \in X$.

A \Define{digraph} is a pair $D = (V,E)$ where $E \subseteq V^2$. We shall identify the digraph $D$ with its in-neighbourhood function: for any set $I \subseteq V$ of vertices, its in-neighbourhood in $D$ is 
\[
    D(I) = \{ u : \exists v \in I, (u,v) \in E \}. 
\]
Let $D = (V,E)$ be a digraph. Adapting \cite[Definition 24]{RR13}, we say that $\phi: X \to X$ \Define{depends} on $D$ if for all $x, y \in X$ and $S \subseteq V$,
\[
    x_{ D( S ) } = y_{ D( S ) } \implies \phi_S(x) = \phi_S(y).
\]
It is worth noting that when we write ``$\phi$ depends on $D$'', we effectively mean that the local function $\phi_i$ of $\phi$ only depends on the variables $x_j$ in the in-neighbourhood of $i$. In fact, \cite{RR13} say that $\phi$ is ``independent'' of all the other variables. For this notion of dependency, and its generalisations introduced in Section \ref{section:dependency}, one should always have this interpretation of ``$\phi$ depending on nothing more than $D$'' in mind.

The mapping $\phi_j : X \to X_j$ \Define{depends essentially} on $x_i$ if changing the value of $x_i$ can change the value of $\phi_j( x )$; more formally, if and only if there exist $a,b \in X$ such that $a_{V \setminus i} = b_{V \setminus i}$ and $\phi_j( a ) \ne \phi_j( b )$. The \Define{interaction graph} of $\phi$ is the digraph $G^\phi = (V, E)$ with $(i,j) \in E$ if and only if $\phi_j$ depends essentially on $i$. If $V$ is finite, then $\phi$ depends on $D$ if and only if $G^\phi \subseteq D$.

Robert's theorem states that if $\phi$ is a finite automata network with an acyclic interaction graph, then $\phi$ converges to a unique fixed point (in at most $|V|$ steps). Below we give two proofs: the first one is based on a ``crystallisation'' phenomenon, which proves convergence in a finite number of steps; while the second is based on the complete bounded Banach contraction principle, which we will be able to generalise.

\begin{theorem}[Robert's theorem for finite automata networks \cite{Rob80}] \label{theorem:robert}
Let $\phi$ be a finite automata network. If the interaction graph of $\phi$ is acyclic, then $\phi$ strongly converges to a unique fixed point.
\end{theorem}

\begin{proof}
Let $|V| = n$, and order the vertices of $G^\phi$ so that $(i,j)$ is an arc only if $i < j$. Let $S_0 = \emptyset$ and $S_i = \{1, \dots, i\}$ for $1 \le i \le n$. We then have $G^\phi(S_i) \subseteq S_{i-1}$, thus for any $x,y \in X$,
\[
     x_{S_{i-1}} = y_{S_{i-1}} \implies x_{ G^\phi( S_i ) } = y_{ G^\phi( S_i ) } \implies \phi_{S_i}(x) = \phi_{S_i}(y).
\]

First proof: Since $x_{S_0} = y_{S_0}$ for all $x,y$, we obtain by simple induction $\phi_{S_i}^i(x) = \phi_{S_i}^i(y)$ and hence $\phi^n( x ) = \phi^n( y )$.

Second proof: Alternatively, let $\Delta(x,y) = \{ 0 \le i \le n : x_{S_i} \ne y_{S_i} \}$ and $d( x, y ) = \sum_{i \in \Delta(x,y)} 2^{-i}$. Then $( L, d )$ forms a complete bounded metric space. Since $d( \phi(x), \phi(y) ) \le \frac{1}{2} d( x, y )$, $\phi$ is contractive and hence by the complete bounded Banach contraction principle it strongly converges to a unique fixed point.
\end{proof}

The feedback bound is a result closely related to Robert's theorem. Let $D$ be a finite digraph. A \Define{feedback vertex set} of $D$ is a set $I \subseteq V$ such that $D[V \setminus I]$ is acyclic. Anticipating notation introduced in the next section in more generality, define $\Residual{D}(S) = \{ v \in V: D(v) \subseteq S \}$ for all $S \subseteq V$. Then $I$ is a feedback vertex set of $D$ if and only if ${ \Residual{ D } }^n(I) = V$, where $n = |V|$, see \cite{Gad13}.

\begin{theorem}[Feedback bound for finite automata networks \cite{Rii07, Ara08}] \label{theorem:feedback}
Let $\phi$ be a finite automata network and let $I$ be a feedback vertex set of its interaction graph. Then $|\Fix( \phi )| \le |X_I|$.
\end{theorem}

\begin{proof}
Let $G = G^\phi$ be the interaction graph of $\phi$ and suppose $a, b \in \Fix(\phi)$ satisfy $a_I = b_I$. Then it is easily shown by induction that $a_{ { \Residual{ G } }^i(I) } = b_{ { \Residual{ G } }^i(I) }$ for all $0 \le i \le n$, thus $a = b$. Therefore, the mapping $\Fix(\phi) \to X_I$ which maps $a$ to $a_I$ is injective and $|\Fix( \phi )| \le |X_I|$.
\end{proof}

\subsection{Complete lattices} \label{section:lattices}

For a review of lattices, the reader is directed to the authoritative monograph by Gr\"azer \cite{Gra13}. A \Define{complete lattice} is a lattice $\Lattice{L} = (L, \join, \meet, \zero, \one)$ where every subset $S \subseteq L$ has a well-defined meet and join: $\bigjoin S = \bigjoin_{x \in S} x \in L$ and $\bigmeet S = \bigmeet_{x \in S} x \in L$. By convention, $\zero = \bigjoin \emptyset$ while $\one = \bigmeet \emptyset$.

The \Define{opposite} of a complete lattice $\Lattice{L} = ( L, \join, \meet, \zero, \one )$ is obtained by reversing the order: $\Opposite{ \Lattice{L} } = ( L, \meet, \join, \one, \zero )$. When referring to properties for the opposite lattice, we will either add the prefix ``co-'' or use the prefixes ``meet-'' and ``join-'' in order to clarify between $\Lattice{L}$ and $\Opposite{ \Lattice{L} }$.

We consider the following special kinds of complete lattices; the containment relations amongst them are represented in Figure \ref{figure:lattices}.
\begin{enumerate}
    \item \label{item:distributive}
    A complete lattice is \Define{distributive} if for all $x,y,z \in L$, $x \meet (y \join z) = (x \meet y) \join (x \meet z)$.

    \item \label{item:frame}
    A complete lattice $\Lattice{L}$ is a \Define{frame} \cite{Ran17}, also known as a complete Heyting algebra, if for all $x \in L$ and $Y \subseteq L$, $x \meet \bigjoin Y = \bigjoin (x \meet Y)$.

    \item \label{item:locale}
    The opposite of a frame is a \Define{locale}, also known as a complete co-Heyting algebra. 

    \item \label{item:Boolean_algebra}
    A lattice is complemented if any $x \in L$ has a complement $\neg x \in L$ such that $x \join \neg x = \one$ and $x \meet \neg x = \zero$. A \Define{complete Boolean algebra} is a complemented, complete distributive lattice. The complement is necessarily unique in a Boolean algebra. A thorough introduction to Boolean algebras is given in \cite{GH09}.

    \item \label{item:power_set_lattice}
    For any set $V$, its power set is denoted as $P(V)$ and the \Define{power set lattice} of $V$ is $\Lattice{\Powerset}(V) = ( L = \Powerset(V), \join = \cup, \meet = \cap, \zero = \emptyset, \one = V)$. The power set lattice is a complete Boolean algebra, where $\neg x = V \setminus x$ for all $x \subseteq V$.

    \item \label{item:chain}
    A \Define{complete chain} is a complete lattice where all elements are comparable.

    \item \label{item:trivial_lattice}
    Let us call a complete lattice with $|L| \le 2$ \Define{trivial}.
\end{enumerate}
    
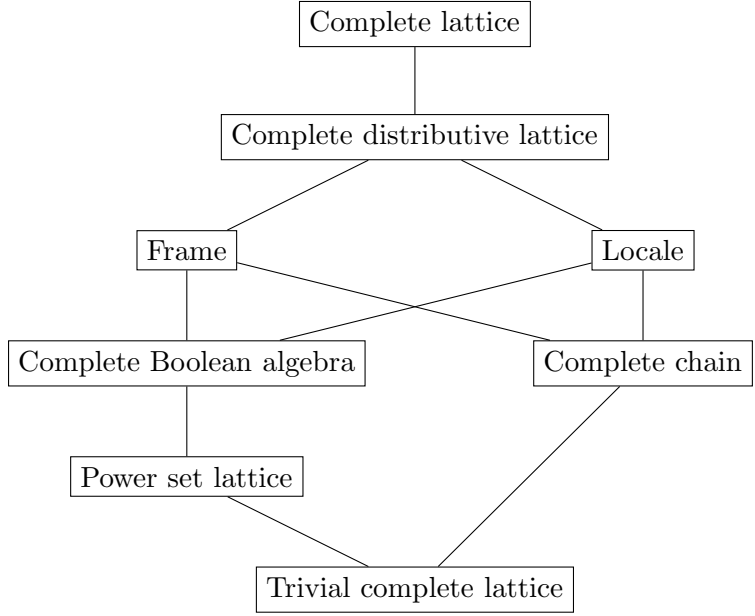
\begin{figure}
    \centering

\begin{tikzpicture}[xscale=3, yscale=1.5, every node/.style=draw]

    \node (lattice) at (1,5) {Complete lattice};
    \node (distributive) at (1,4) {Complete distributive lattice};
    \node (frame) at (0,3) {Frame};
    \node (locale) at (2,3) {Locale};
    \node (boolean) at (0,2) {Complete Boolean algebra};
    \node (chain) at (2,2) {Complete chain};
    \node (powerset) at (0,1) {Power set lattice};
    \node (trivial) at (1,0) {Trivial complete lattice};

    \draw (lattice) -- (distributive);
    \draw (distributive) -- (frame);
    \draw (distributive) -- (locale);
    \draw (frame) -- (boolean);
    \draw (frame) -- (chain);
    \draw (locale) -- (boolean);
    \draw (locale) -- (chain);
    \draw (boolean) -- (powerset);
    \draw (powerset) -- (trivial);
    \draw (chain) -- (trivial);
\end{tikzpicture}
   
    \caption{Classes of complete lattices used in this paper.}
    \label{figure:lattices}
\end{figure}

Say a complete lattice is ascending chain condition (\Define{ACC}) if all ascending chains are finite: if $(a_i : i \in \N)$ satisfy $a_{i+1} \ge a_i$ for all $i$, then there exists $n$ such that $a_m = a_n$ for all $m \ge n$. The descending chain condition (\Define{DCC}) is defined similarly: if $(b_i : i \in \N)$ satisfy $b_{i+1} \le b_i$ for all $i$, then there exists $n$ such that $b_m = b_n$ for all $m \ge n$.

For any $S \subseteq L$, we denote $S^\downarrow = \{ y \in L : \exists s \in S : y \le s \}$ and $S^\uparrow = \{ y \in L : \exists s \in S : y \ge s \}$. We denote $x^\uparrow = \{x\}^\uparrow$ and $x^\downarrow = \{x\}^\downarrow$ for all $x \in L$; those are the principal filters and ideals of the lattice, respectively.

\section{Graphs on lattices} \label{section:graphs}

In this section, we generalise the concept of graphs to all lattices. We justify our generalisation by proving multiple generalisations of classical results on digraphs in our setting.

\subsection{Graphs} \label{subsection:graphs}

The seminal observation is that in-neighbourhood functions of digraphs are exactly the union-preserving mappings:

\begin{observation} \label{observation:digraphs}
Let $f : \Powerset(V) \to \Powerset(V)$. Then $f$ is the in-neighbourhood function of some digraph on $V$ if and only if $f(\bigcup_{x \in S} x) = \bigcup_{x \in S} f(x)$ for all $S \subseteq \Powerset(V)$.
\end{observation}

\begin{proof}
It is clear that for any $D$ and any $S \subseteq \Powerset(V)$, $D( \bigcup_{x \in S} x ) = \bigcup_{x \in S} D(x)$. 

Conversely, suppose $f(\bigcup_{x \in S} x) = \bigcup_{x \in S} f(x)$ for all $S \subseteq \Powerset(V)$. First, note that $f( \emptyset )$ is the empty union, i.e. $f( \emptyset ) = \emptyset$. Let $D = (V,E)$ with $E = \{ (u,v) : u,v \in V, u \in f(v) \}$. We prove that $D(x) = f(x)$ for any $x \in \Powerset(V)$. For $x = \emptyset$, $D( \emptyset ) = \emptyset = f( \emptyset )$. For $x = \{v\}$, we have $D(v) = \{ u : u \in f(v) \} = f(v)$. For any other $x \in \Powerset(V)$, we have $D(x) = \bigcup_{v \in x} D(v) = \bigcup_{v \in x} f(v) = f( \bigcup_{v \in x} v) = f(x)$.
\end{proof}

Let $\Lattice{L} = (L, \join, \meet, \zero, \one)$ be a complete lattice. A \Define{graph} on $\Lattice{L}$ is a mapping $f: L \to L$ that preserves arbitrary joins, i.e. for all $S \subseteq L$,
\[
    f \left( \bigjoin_{x \in S} x \right) = \bigjoin_{x \in S} f(x).
\]
We shall reserve the term ``digraph'' for a graph on $\Lattice{L} = \Lattice{\Powerset}(V)$, and make remarks about digraphs in the subsequent observations throughout the paper. 

We shall use the following shorthand notation: $f(S) = \{ f(x) : x \in S \}$. As such, a graph is any mapping $f$ for which
\[
    f \left( \bigjoin S \right) = \bigjoin f(S)
\]
for all $S \subseteq L$. We note that, by taking $S = \emptyset$, we obtain $f(\zero) = \zero$ for any graph. Here are a few examples of graphs:
\begin{enumerate}
    \item \label{item:empty_graph}
    $f(x) = \zero$ for all $x \in L$ is the empty graph;
    
    \item \label{item:complete_graph}
    $f(x) = \one$ for all $x \ne \zero$ is the complete graph;
    
    \item \label{item:identity}
    $f = \id$ is the identity graph; 
    
    \item \label{item:conjunction}
    For any $a \in L$, let $g_a: L \to L$ be defined by $g_a(x) = x \meet a$ for all $x \in L$. Then $\Lattice{L}$ is a frame if and only if $g_a$ is a graph for all $a \in L$. If so, then for any graph $f$ on $\Lattice{L}$, the subgraph of $f$ \Define{induced} by $a$ is the graph $f[a]$ on $a^\downarrow$ defined by $f[a](x) = g_a f(x) = f(x) \meet a$.
\end{enumerate}

The family of graphs on $\Lattice{L}$, denoted by $\Graphs$, is closed under taking arbitrary joins and hence forms a complete lattice \cite[Chapter I, Section 3, Lemma 14]{Gra13}. Indeed, if $F \subseteq \Graphs$, then $h = \bigjoin F \in \Graphs$, since
\[
    h \left( \bigjoin_{x \in S} x \right) = \bigjoin_{f \in F} f\left( \bigjoin_{x \in S} x \right) = \bigjoin_{f \in F} \bigjoin_{x \in S} f(x) = \bigjoin_{x \in S} \bigjoin_{f \in F} f(x) = \bigjoin_{x \in S} h(x).
\]
The family of graphs $\Graphs$ is also closed under composition. We shall omit the composition symbol, and write $fg$ for the mapping where $fg(x) = f(g(x))$ for all $x \in L$. 

A graph on the opposite lattice $\Opposite{ \Lattice{L} }$ is a \Define{co-graph}, i.e. a mapping $g : L \to L$ such that $g \left( \bigmeet S \right) = \bigmeet g(S)$ for all $S \subseteq L$.

A mapping $f : L \to L$ is \Define{monotone} if $x \le y \implies f(x) \le f(y)$ for all $x,y \in L$, or equivalently, 
\[
    f \left( \bigjoin S \right) \ge \bigjoin f( S )
\]
for all $S \subseteq L$. Clearly, graphs are monotone. 

For any complete lattice $\Lattice{L}$ with at least two elements, it is easy to construct a monotone mapping that is not a graph: $f(x) = \one$ for all $x \in L$ is not a graph since $f(\zero) \ne \zero$. Instead, we characterise below the complete lattices where all monotone mappings fixing $\zero$ are graphs.

\begin{proposition} \label{proposition:classification_L}
All monotone mappings on $\Lattice{L}$ fixing $\zero$ are graphs if and only if $\Lattice{L}$ is an ACC chain.
\end{proposition}

\begin{proof}
Let $\Lattice{L}$ be an ACC chain, and let $f$ be monotone and fix $\zero$. For any nonempty $S \subseteq L$ with $s = \bigjoin S$, we have $s \in S$ and $f(t) \le f(s)$ for all $t \in S$, hence $f(\bigjoin S) = f(s) = \bigjoin f(S)$. Conversely, if $\Lattice{L}$ is not a chain, let $a$ and $b$ be non-comparable, and let 
\[
    f(x) = \begin{cases}
    \zero &\text{if } x \not\ge b, \\
    x & \text{if } x \ge b;
    \end{cases}
\]
then $f$ is monotone but $f(a) \join f(b) = b$ while $f(a \join b) = a \join b > b$, hence it is not a graph. Finally, let $\Lattice{L}$ be a complete chain with an infinite ascending chain $S$. Let $s = \bigjoin S$ and 
\[
    f(x) = \begin{cases} 
    \zero &\text{if } x < s \\ 
    \one & \text{if } x \ge s; 
    \end{cases}
\]
then $f$ is monotone but $f(\bigjoin S) = \one$ while $\bigjoin f(S) = \zero$, hence it is not a graph.
\end{proof}

\subsection{Graphs as residuated mappings} \label{subsection:residuated}

Let $\psi: L \to L$. For any $S \subseteq L$, the pre-image of $S$ under $\psi$ is denoted as usual: $\psi^{-1}( S ) = \{ x \in L : \psi(x) \in S \}$. We introduce two mappings $\Residual{\psi}, \Residuated{\psi} : L \to L$ such that for all $a,b \in L$,
\begin{alignat*}{3}
    \Residual{\psi}(a) &:= \bigjoin \{ x: \psi(x) \le a \} &&= \bigjoin \psi^{-1}( a^\downarrow ),\\
    \Residuated{\psi}(b)  &:= \bigmeet\{ y: \psi(y) \ge b \} &&= \bigmeet \psi^{-1}( b^\uparrow ).
\end{alignat*}

\begin{observation} \label{observation:residual_digraphs}
For any digraph $D$, $\Residual{D}(x) = \{ v \in V: D(v) \subseteq x \}$.
\end{observation}

Let us make some remarks in the general case, whose proofs are obvious and hence omitted. There is a natural partial order on $L \to L$ mappings, whereby $\psi \ge \theta$ if and only if $\psi(x) \ge \theta(x)$ for all $x \in L$, for any $\psi, \theta : L \to L$.

\begin{lemma} \label{lemma:general_remarks_residual}
Let $\psi, \theta: L \to L$. Then the following hold:
\begin{enumerate}
    \item \label{item:residual_monotone}
    $\Residual{\psi}$ and $\Residuated{\psi}$ are monotone;
    
    \item \label{item:residuated_compose}
    $\Residuated{\psi} \psi \le \id \le \Residual{\psi} \psi$;

    \item \label{ite:residual_residuated}
    $\Residual{ \Residuated{ \psi } } \le \psi \le \Residuated{ \Residual{ \psi } }$;
    
    \item \label{item:residual_antitone}
    if $\psi \ge \theta$, then $\Residual{\psi} \le \Residual{\theta}$ and $\Residuated{\psi} \le \Residuated{\theta}$.
\end{enumerate}
\end{lemma}

A mapping $f : L \to L$ is \Define{residuated} if the pre-image of every principal ideal is a principal ideal. It is easy to see that a residuated mapping is monotone. Many characterisations and properties of residuated mappings are given in \cite[Chapter 2]{BJ72} and \cite[Chapter 1]{Bly05}, some with proofs, others without proofs, and others still only implicit. For the sake of completeness and clarity, we combine the key properties that we need in the theorem below, for which we give a full proof.

\begin{theorem}[Properties of residuated mappings on complete lattices] \label{theorem:residuated}
Let $\Lattice{L}$ be a complete lattice and let $f : L \to L$. Then the following are equivalent:
\begin{equivalent}
    \item \label{item:residuated}
    $f$ is residuated;
    
    \item \label{item:graph}
    $f$ is a graph;

    \item \label{item:gf}
    $f \Residual{f} \le \id$.
\end{equivalent}

If $f$ satisfies any property above, then the following hold:
\begin{enumerate}
    \item \label{item:g_co-graph}
    $\Residual{f}$ is a co-graph;

    \item \label{item:f=g+}
    $f = \Residuated{ \Residual{ f } }$;

    \item \label{item:hf}
    for all $x,y \in L$, $f(x) \le y$ if and only if $x \le \Residual{f}(y)$;
    
    \item \label{item:vNregular}
    $f \Residual{f} f = f$ and $\Residual{f} f \Residual{f} = \Residual{f}$.
\end{enumerate}
\end{theorem}

\begin{proof}
$\ref{item:graph} \implies \ref{item:gf}$. If $f$ is a graph, then for any $a \in L$, we have
\[
    f \Residual{f}(a) = f \left( \bigjoin \{ x: f(x) \le a \} \right) = \bigjoin f( \{ x: f(x) \le a \} ) \le a,
\]
thus $f \Residual{f} \le \id$.

$\ref{item:gf} \implies \ref{item:graph}$. Suppose $f \Residual{f} \le \id$ and let $S \subseteq L$. Let $y = \bigjoin f( S )$, then for all $s \in S$,
\[
   \Residual{f} (y) \ge \Residual{f} f(s) \ge s,
\]
thus $\Residual{f}(y) \ge \bigjoin S$ and $y \ge f \Residual{f}(y) \ge f(\bigjoin S)$.

$\ref{item:gf} \iff \ref{item:residuated}$. If $f$ is not monotone, then it satisfies neither Property \ref{item:residuated} (as residuated mappings are monotone) nor Property \ref{item:gf} (as it is equivalent to being a graph). Let $f$ be a monotone mapping, so that $f^{-1}( x^\downarrow )$ is an ideal for all $x \in L$. For all $a \in L$, we have the following equivalences:
\[
    f \Residual{f}(a) \le a \iff f\left( \bigjoin f^{-1}( a^\downarrow ) \right) \le a \iff \bigjoin f^{-1}( a^\downarrow ) \in f^{-1}( a^\downarrow ) \iff f^{-1}( a^\downarrow ) \text{ is a principal ideal}.
\]

We now prove that, if $f$ is a graph, then $f$ satisfies Properties \ref{item:g_co-graph} to \ref{item:vNregular}.

\begin{enumerate}
    \item 
Let $S \subseteq L$. On the one hand, 
$\Residual{f}(\bigmeet S) \le \bigmeet \Residual{f}(S)$ by monotonicity. On the other hand, we have
\[
    f\left( \bigmeet \Residual{f}(S) \right) \le \bigmeet f \Residual{f}(S)  \le \bigmeet S,
\]
thus
\[
    \bigmeet \Residual{f}(S) \le \Residual{f}  f \left( \bigmeet \Residual{f}(S) \right)  \le \Residual{f} \left( \bigmeet S \right).
\]

\item 
Applying Property \ref{item:gf} to $\Residual{f}$, a graph on the opposite lattice $\Opposite{ \Lattice{L} }$, yields $\Residuated{ \Residual{f} } \Residual{ f } \le \id$ and $\Residual{ f } \Residuated{ \Residual{f} } \ge \id$. Combining with Property \ref{item:gf} for $f$, we obtain
\[
    f \le f \Residual{ f } \Residuated{ \Residual{ f } } \le \Residuated{ \Residual{ f } } \le \Residuated{ \Residual{ f } } \Residual{ f } f \le f.
\]

\item 
Let $x, y \in L$. Thanks to the monotonicity of $f$ and $\Residual{f}$ and to Property \ref{item:gf}, we have the following implications:
\[
    f(x) \le y \implies \Residual{ f } f(x) \le \Residual{ f }(y) 
    \implies x \le \Residual{ f }(y) 
    \implies f(x) \le f \Residual{ f }(y)
    \implies f(x) \le y.
\]

\item 
Follows immediately from Property \ref{item:gf} and Lemma \ref{lemma:general_remarks_residual}\ref{item:residuated_compose}.
\end{enumerate}

\end{proof}

\subsection{Reflexive and transitive graphs} \label{subsection:RT}

We now introduce three graphs on the lattice of graphs $\Graphs$, namely the reflexive, transitive, and reflexive transitive closures of graphs. 
\begin{enumerate}

    \item \label{item:reflexive}
    A graph $f$ is \Define{reflexive} if $f \ge \id$. 
    We denote the \Define{reflexive closure} of $f$ as $\Reflexive{f} = \id \join f$. The following are equivalent for a graph $g$: $g$ is reflexive; $g = \Reflexive{g}$; $g = \Reflexive{f}$ for some graph $f$.
    
    \item \label{item:transitive}
    A graph $f$ is \Define{transitive} if $f \ge f^2$. The \Define{transitive closure} of $f$ is $\Transitive{f} = \bigjoin_{i \in \N_1} f^i$. The following are equivalent for a graph $g$: $g$ is transitive; $g = \Transitive{g}$; $g = \Transitive{f}$ for some graph $f$.

    \item \label{item:reflexive_transitive}
    A graph $f$ is \Define{reflexive transitive} if $f^2 = f \ge \id$. The \Define{reflexive transitive closure} of $f$ is then $\ReflexiveTransitive{f} = \Reflexive{ ( \Transitive{f} ) } = \Transitive{ ( \Reflexive{f} ) } = \bigjoin_{i \in \N} f^i$. Therefore, the following are equivalent for a graph $g$: $g$ is reflexive transitive; $g = \ReflexiveTransitive{g}$; $g = \ReflexiveTransitive{f}$ for some graph $f$.
\end{enumerate}

Reflexive and transitive digraphs are easily characterised; see below (a \Define{loop} on the vertex $v$ is the edge $(v,v)$). The proof is obvious and hence omitted.

\begin{observation} \label{observation:reflexive_digraphs}
A digraph is reflexive if and only if there is a loop on every vertex. A digraph is transitive if and only for all $u,v,w \in V$, $(u,v), (v,w) \in E$ implies $(u,w) \in E$.
\end{observation}

Our next aim is to relate the sets of fixed points of $f$ with those of its different closures. Firstly, we say a graph $g$ is \Define{non-degenerate} if its transitive closure is reflexive, i.e. $\Transitive{g}(x) \ge x$ for all $x$.

\begin{observation} \label{observation:non-degenerate_digraphs}
A digraph is non-degenerate if and only if every vertex belongs to a cycle.
\end{observation}

\begin{proof}
Suppose $D$ is non-degenerate. We have $v \in \Transitive{D}(v)$ if and only if there exists $i$ such that $v \in D^i( v )$, and hence there exists a cycle $(v = v_0, v_1, \ldots, v_{i-1})$.

Conversely, suppose every vertex of $D$ belongs to a cycle, that is for all $v \in V$ there exists $i_v$ such that $v \in D^{i_v}( V )$. Let $X \subseteq V$, then $X \subseteq \bigcup_{v \in X} D^{i_v}(v) \subseteq \Transitive{ D }(X)$.
\end{proof}

\begin{proposition} \label{proposition:fixed_points}
Let $f$ be a graph. We have
\begin{equation} \label{equation:Fix}
    \Fix( \Transitive{f} ) = \Fix( f ) \subseteq \Fix( \Reflexive{ f } ) = \Fix( \ReflexiveTransitive{ f } ),    
\end{equation}
where $\Fix( f ) = \Fix( \Reflexive{ f } )$ if and only if $f$ is non-degenerate.
\end{proposition}

\begin{proof}
We first prove Equation \eqref{equation:Fix}. 
\begin{itemize}
    \item $\Fix(f) \subseteq \Fix(\Reflexive{f})$. If $x = f(x)$, then $x = f(x) \join x$.

    \item $\Fix(f) \subseteq \Fix(\ReflexiveTransitive{f})$. If $x = f(x)$, then $x = f^i(x)$ for all $i\ge 0$, and hence $x = \ReflexiveTransitive{f}(x)$.

    \item $\Fix( \Transitive{f} ) = \Fix(f)$. The first item shows that $\Fix(\Transitive{f}) \subseteq \Fix(\ReflexiveTransitive{f})$.  Since $\Transitive{f} = f \ReflexiveTransitive{f}$, we have for all $x \in L$,
    \[
        \Transitive{ f }(x) = x \iff \left[ f \ReflexiveTransitive{f}(x) = x \text{ and } \ReflexiveTransitive{f}(x) = x \right] \iff f(x) = x.
    \]

    \item $\Fix( \ReflexiveTransitive{f} ) = \Fix(\Reflexive{f})$. Apply the previous item to $\Reflexive{f}$.
\end{itemize}

We now prove that $\Fix( f ) = \Fix( \Reflexive{ f } )$ if and only if $f$ is non-degenerate. If $f$ is non-degenerate, then $\Transitive{f} = \ReflexiveTransitive{f}$, so equality follows Equation \eqref{equation:Fix}. Otherwise, let $z \in L$ such that $\Transitive{f}(z) < \ReflexiveTransitive{f}(z)$, and let $y = \ReflexiveTransitive{f}(z)$. We then have $f(y) = \Transitive{f}(z) < y$, while $\ReflexiveTransitive{f}(y) = y$ by idempotence of $\ReflexiveTransitive{f}$.
\end{proof}

Let $X$ be a finite set. Then topologies on $X$ are in one-to-one correspondence with partial preorders, i.e. reflexive and transitive relations on $X$ \cite[Theorem 3.9.1]{Cam99}. This fact, which we shall refer to as the finite topology theorem, gives an equivalence between finite topologies and finite reflexive transitive digraphs. In \cite{Gad21}, that equivalence is reviewed in terms of disjunctive Boolean networks. It is shown that finite topologies are exactly the sets of fixed points of finite reflexive transitive digraphs, which in turn are the same as the sets of fixed points of non-degenerate digraphs (\cite{Gad21} uses the term ``non-trivial'' but this is not the most appropriate choice of terminology). We now generalise the finite topology theorem to any complete lattice by considering so-called bi-topologies. We give two equivalent characterisations of such subsets of the lattice, one for reflexive transitive graphs, the other for non-degenerate graphs.

A \Define{bi-topology} on $\Lattice{L}$ is a subset $T$ of $L$ which is closed under arbitrary joins and meets, i.e. such that for any subset $S \subseteq T$, $\bigmeet S, \bigjoin S \in T$.

\begin{theorem}[The bi-topology theorem] \label{theorem:bi-topology}
Let $\Lattice{L}$ be a complete lattice. Then the following are equivalent for $T \subseteq L$.
\begin{equivalent}
    \item \label{item:bi-topology}
    $T$ is a bi-topology on $\Lattice{L}$;
    
    \item \label{item:fix_of_RT}
    $T$ is the set of fixed points of some reflexive transitive graph on $\Lattice{L}$;
    
    \item \label{item:fix_of_non-degenerate}
    $T$ is the set of fixed points of some non-degenerate graph on $\Lattice{L}$.
\end{equivalent}
\end{theorem}

\begin{proof}
$\ref{item:bi-topology} \implies \ref{item:fix_of_RT}$. Suppose $T$ is a bi-topology. Define $h : L \to L$ as follows. For all $x \in L$, let $T_x = T \cap x^\uparrow$; then $h(x) = \bigmeet T_x$. Clearly, $h$ is monotone and $h \ge \id$.

The fixed point set of $h$ is $T$. Indeed, $h(x) \in T$ for all $x$, hence $\Fix(h) \subseteq T$; conversely, $h(t) = t$ if $t \in T$. 

In fact, we have just shown that $h$ is idempotent. Therefore, all that is left to prove is that $h$ is a graph. Let $S \subseteq L$. On the one hand, $h( \bigjoin S ) \ge \bigjoin h( S )$ by monotonicity. On the other hand, since $\bigjoin h( S ) \ge \bigjoin S$ and $\bigjoin h( S ) \in T$, we have $\bigjoin h( S ) \in T_{\bigjoin S}$, thus
\[
    \bigjoin h( S ) \ge \bigmeet T_{\bigjoin S} = h( \bigjoin S ).
\]

$\ref{item:fix_of_RT} \implies \ref{item:bi-topology}.$ Conversely, let $h$ be a reflexive transitive graph. Because $h$ is a graph, its image set (the same as its set of fixed points, as $h$ is idempotent) is closed under arbitrary joins: $\bigjoin h( S ) = h( \bigjoin S )$. We now prove that $\Fix(h)$ is closed under arbitrary meets. Let $S \subseteq \Fix(h)$ and $z = \bigmeet S = \bigmeet h(S)$. On the one hand $h( z ) \ge z$, and on the other hand $h( z ) \le h( x )$ for all $x \in S$ hence $h(z) \le z$. Thus $h(z) = z$ and $z \in \Fix(h)$.

$\ref{item:fix_of_RT} \iff \ref{item:fix_of_non-degenerate}$. Follows from Proposition \ref{proposition:fixed_points}.
\end{proof}

A \Define{co-topology} on $\Lattice{L}$ is a subset $T$ of $L$ which is closed under arbitrary meets and finite joins. Equivalently, $T$ is a co-topology if $\zero,\one \in T$; for any subset $S \subseteq T$, $\bigmeet S \in T$; and for all $s,t \in T$, $s \join t \in T$. A \Define{pre-graph} is a mapping $g : L \to L$ that preserves finite joins; reflexive, transitive, and non-degenerate pre-graphs are defined naturally. The proof of Theorem \ref{theorem:bi-topology} can be easily adapted to pre-graphs to yield the corresponding result for co-topologies.

\begin{corollary}[The co-topology theorem] \label{theorem:co-topology}
Let $\Lattice{L}$ be a complete lattice. Then the following are equivalent for $T \subseteq L$:
\begin{equivalent}
    \item \label{item:co-topology}
    $T$ is a co-topology on $\Lattice{L}$;
    
    \item \label{item:fix_RT_pre-graph}
    $T$ is the set of fixed points of some reflexive transitive pre-graph on $\Lattice{L}$;
    
    \item \label{item:fix_non-degenerate_pre-graph}
    $T$ is the set of fixed points of some non-degenerate pre-graph on $\Lattice{L}$.
\end{equivalent}
\end{corollary}

\subsection{Fixed point-free graphs} \label{subsection:fixed_point-free}

We say a graph $f$ is \Define{fixed point-free} if $\zero$ is the only fixed point of $f$.

Let us first classify the fixed point-free digraphs. If a digraph $D$ has a cycle $X$, then $X \subseteq D(X)$ and hence it is not fixed-point free. Moreover, if $D$ is finite and acyclic, then $X \not\subseteq D(X)$ for all $X$. As such, a finite digraph is fixed-point free if and only if it is acyclic. However, this is not true for infinite digraphs. A (forward) \Define{infinite walk} in a digraph is a sequence $W = (w_i : i \in \N)$ of vertices such that $(w_i, w_{i+1})$ is an arc for all $i \in \N$.

\begin{observation} \label{observation:strongly_acyclic_digraphs}
A digraph is fixed point-free if and only if it has no infinite walks.
\end{observation}

\begin{proof}
If $D$ has an infinite walk $W = (w_i : i \in \N)$, then $D( W ) \supseteq W$, and hence $D$ is not fixed point-free. 

Conversely, if $D$ is not fixed point-free, let $X \ne \emptyset$ satisfy $X = D(X)$. Then for all $u \in X$, there exists $v \in X$ such that $(u,v) \in E$. As such, for any $u \in X$, there exists a sequence $V = (v_i : i \in \N)$ with $v_0 = u$, $v_i \in X$ for all $i$, and $(v_i, v_{i+1}) \in E$ for all $i$.
\end{proof}

We now give several alternate definitions of fixed point-free graphs.

\begin{proposition}
Let $f$ be a graph. Then the following are equivalent:
\begin{equivalent}
    \item \label{item:fixed_point-free}
    $f$ is fixed point-free, i.e. $\Fix( f ) = \{ \zero \}$;

    \item \label{item:increased_point-free}
    for all $x \in L$, $f(x) \ge x$ if only if $x = \zero$;

    \item \label{item:f<fR}
    $f(x) < \Reflexive{f}(x)$ for all $x \ne \zero$;

    \item \label{item:increased_set-free}
    for all $X \subseteq L$, $f(X) \supseteq X$ only if $X \subseteq \{ \zero \}$;

    
    \item \label{item:parabolic_point-free}
    $\Par( f ) = \{ \zero \}$.
\end{equivalent}
\end{proposition}

\begin{proof}
$\ref{item:fixed_point-free} \implies \ref{item:increased_point-free}$. For the sake of contradiction, suppose $f(u) \ge u$ for some $u \ne \zero$. Let $v = \ReflexiveTransitive{f}(u)$, then we obtain
\[
    f(v) = f \left( \bigjoin_{i \in \N} f^i( u ) \right) = \bigjoin_{i \in \N} f^{i+1}(u) = \bigjoin_{i \in \N} f^i(u) = v,
\]
which is the desired contradiction.

$\ref{item:increased_point-free} \implies \ref{item:fixed_point-free}$. Trivial.

$\ref{item:increased_point-free} \iff \ref{item:f<fR}$. For all $x$, $f(x) \not\ge x \iff f(x) < f(x) \join x = \Reflexive{f}( x )$.

$\ref{item:increased_point-free} \implies \ref{item:increased_set-free}$. For the sake of contradiction, suppose $X \subseteq f( X )$ with $q = \bigjoin X > 0$. We have
\[
    f( q ) = f \left( \bigjoin X \right) = \bigjoin f( X ) \ge \bigjoin X = q,
\]
which is the desired contradiction.

$\ref{item:increased_set-free} \implies \ref{item:parabolic_point-free}$. The proof follows Lemma \ref{lemma:phi(Par)}.

$\ref{item:parabolic_point-free} \implies \ref{item:fixed_point-free}$. Trivial.
\end{proof}

\begin{corollary} \label{corollary:f_le_g_fixed_point-free}
If $f \le g$ are graphs and $g$ is fixed point-free, then so is $f$.
\end{corollary}


A \Define{topological sort} for $f$ is a linear order $\preceq$ on $L$ such that $y \le \Transitive{f}(x)$ only if $x \prec y$ for all $x,y \ne \zero$. We remark that this definition differs from that given for digraphs in \cite{RRM13}, whereby the topological sort was defined for vertices only; in our definition, the topological sort for digraphs is defined for all subsets of vertices instead.

\begin{theorem}[Topological sort of fixed point-free graphs] \label{theorem:topological_sort}
A graph $f$ is fixed point-free if and only if it has a topological sort.
\end{theorem}

\begin{proof}
Let $f$ be fixed point-free; note that $\Transitive{f}$ is also fixed point-free. We only need to show that the relation $y \le \Transitive{f}(x)$ is transitive, irreflexive (unless $x=y=\zero$), and antisymmetric on $L$. The relation is clearly transitive and irreflexive. We now prove antisymmetry. Suppose $\Transitive{f}(x) \ge y$ and $\Transitive{f}(y) \ge x$ for $y \ne x$. Firstly, $x \ne \zero$ for otherwise, we have $y \le \Transitive{f}(\zero) = \zero = x$. Then
\[
    x \le \Transitive{f}(y) \le \Transitive{f} ( \Transitive{f}(x) ) \le \Transitive{f}(x),
\]
which contradicts the fact that $\Transitive{f}$ is fixed point-free.

Conversely, if $f(x) = x$ with $x \ne \zero$, then $x \le \Transitive{f}(x)$ and $x \not\prec x$ for any linear order $\preceq$ on $L$. 
\end{proof}

\begin{example} \label{example:backwards_ray}
The \Define{ray} is the digraph on $V = \N$ with arcs $\{ (i, i+1) : i \in \N \}$.  It is acyclic, but not fixed point-free, for $f( V ) = V$. And indeed, one cannot place a topological sort on the vertices that extends to all subsets: for instance, if $A, B \subseteq V$ are two infinite sets then $\Transitive{D}( A ) = V \supseteq B$ and $\Transitive{D}( B ) = V \supseteq A$.

The \Define{backwards ray} is the digraph on $V = \N$ with arcs $\{ (i, i-1) : i \ge 1 \}$. It is acyclic, and fixed point-free. For any non-empty subsets $x, y$ we have $y \subseteq \Transitive{D}( x )$ if and only if $\min x < \min y$. Therefore, any linear order of $\Powerset( V )$ such that $\min y > \min x \implies x \prec y$ is a topological order.

Those digraphs are illustrated in Figure \ref{figure:rays}.
\end{example}

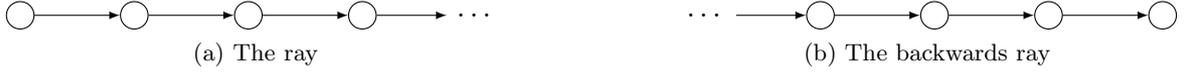
\begin{figure}
\centering
\subfloat[The ray]{
    \begin{tikzpicture}[xscale=1.5]
        \node[draw, circle] (0) at (0,0) {}; 
        \node[draw, circle] (1) at (1,0) {}; 
        \node[draw, circle] (2) at (2,0) {}; 
        \node[draw, circle] (3) at (3,0) {}; 
        \node (4) at (4,0) {$\dots$};
    
        \draw[-latex] (0) -- (1);
        \draw[-latex] (1) -- (2);
        \draw[-latex] (2) -- (3);
        \draw[-latex] (3) -- (4);
    \end{tikzpicture}
}
\hspace{2cm}\subfloat[The backwards ray]{
    \begin{tikzpicture}[xscale=1.5]
        \node[draw, circle] (0) at (0,0) {}; 
        \node[draw, circle] (1) at (-1,0) {}; 
        \node[draw, circle] (2) at (-2,0) {}; 
        \node[draw, circle] (3) at (-3,0) {}; 
        \node (4) at (-4,0) {$\dots$};
    
        \draw[-latex] (1) -- (0);
        \draw[-latex] (2) -- (1);
        \draw[-latex] (3) -- (2);
        \draw[-latex] (4) -- (3);
    \end{tikzpicture}
}
    \caption{The ray and backwards ray. The ray is not fixed point-free while the backwards ray is.}
    \label{figure:rays}
\end{figure}

\subsection{Meet-nilpotent graphs} \label{subsection:nilpotent_graphs}

We now introduce different notions of nilpotence. First, meet-nilpotence, which as we shall see is well suited for graphs:
\begin{itemize}
    \item let $K \in \N$, then we say $\psi : L \to L$ is \Define{$K$-meet-nilpotent} if $\psi^K( \one ) = \zero$;

    \item we simply say $\psi$ is \Define{meet-nilpotent} if it is $K$-meet-nilpotent for some $K$;

    \item we further say that $\psi$ is \Define{asymptotically meet-nilpotent} if $\bigmeet_{i \in \N} \psi^i( \one ) = \zero$.
\end{itemize}
Second, join-nilpotence, which is well suited for co-graphs instead:
\begin{itemize}
    \item let $K \in \N$, then we say $\psi : L \to L$ is \Define{$K$-join-nilpotent} if $\psi^K( \zero ) = \one$;

    \item we simply say $\psi$ is \Define{join-nilpotent} if it is $K$-join-nilpotent for some $K$;

    \item we further say that $\psi$ is \Define{asymptotically join-nilpotent} if $\bigjoin_{i \in \N} \psi^i( \zero ) = \one$.
\end{itemize}

Firstly, we focus on meet-nilpotent graphs. Let us characterise the meet-nilpotent digraphs. A \Define{finite walk} of length $l$ is a sequence of vertices $(w_0, w_1, \dots, w_l)$ such that $(w_i, w_{i+1})$ is an arc for all $0 \le i \le l-1$.

\begin{observation} \label{observation:nilpotent_digraphs}
A digraph $D$ is $K$-meet-nilpotent if and only if $D$ has no walks of length $K$.
\end{observation}

\begin{proof}
For any $i \in \N$, $D^i( V )$ is the set of vertices $v$ such that there exists a walk of length $i$ starting from $v$. Therefore, there exists a walk of length $K$ if and only if $D^K( V )$ is not empty.
\end{proof}

We now prove that nilpotence behaves well within the theory of residuated mappings.

\begin{lemma} \label{lemma:nilpotent}
Let $K \in \N$, $\psi, \theta : L \to L$ and $f$ be a graph on $\Lattice{L}$.
\begin{enumerate}
    \item \label{item:nilpotent_minus}
    If $\psi^K( \one ) = \zero$, then ${ \Residual{\psi} }^K ( \zero ) = \one$.
    
    \item \label{item:nilpotent_plus}
    If $\theta^K( \zero ) = \one$, then ${ \Residuated{\theta} }^K ( \one ) = \zero$.
    
    \item \label{item:nilpotent_graph}
    $f^K( \one ) = \zero$ if and only if ${ \Residual{f} }^K( \zero ) = \one$.
\end{enumerate}
\end{lemma}

\begin{proof}
\begin{enumerate}
    \item \label{item:nilpotent_minus_proof}
    By induction on $0 \le i \le K$, it is easily shown that ${ \Residual{\psi} }^i ( \zero ) \ge \psi^{K-i}( \one )$. In particular, for $i = K$ we obtain ${ \Residual{\psi} }^K ( \zero ) = \one$.
    
    \item \label{item:nilpotent_plus_proof}
    The proof is similar as the first item.
    
    \item \label{item:nilpotent_graph_proof}
    Easily follows from the first two items.
\end{enumerate}
\end{proof}

Secondly, we focus on asymptotically meet-nilpotent graphs. Those are fixed point-free: if $f(x) = x$ some $x \ne \zero$, then $\bigmeet f^i( \one ) \ge \bigmeet f^i(x) = x > \zero$.

Let us classify the asymptotically meet-nilpotent digraphs. A collection of walks $(W_i:  i \in \N)$ in $D$ with a common starting vertex where $W_i$ has length $i$ for all $i$ is called an \Define{infinite pyramid}. 

\begin{observation} \label{observation:asymptotically_nilpotent_digraphs}
A digraph $D$ is asymptotically meet-nilpotent if and only if it has no infinite pyramid.
\end{observation}

\begin{proof}
Suppose there exists a family $(W_i : i \in \N)$ of walks starting at $v$ with length $i$ for all $i \in \N$, and let $w_i$ be the end point of $W_i$. Then $v \in D^i( w_i ) \subseteq D^i( V )$ for all $i$, thus $\bigcap_i D^i( V ) \supseteq \{ v \}$ and $D$ is not asymptotically meet-nilpotent.

Conversely, if $D$ has no infinite pyramid, then for every vertex $v$, then $v \notin D^{ d }( V )$ for any $d$ greater than the maximum length of a walk starting at $v$. Therefore, $\bigcap_i D^i( V ) = \emptyset$, i.e. $D$ is asymptotically meet-nilpotent.
\end{proof}

\begin{example} \label{example:stairways}
We introduce two digraphs which shall be useful in the sequel of this paper.

The \Define{stairway to heaven} has vertex set $V = \{ a \} \cup \{ v_i^j : j, i \in \N, j \le i  \}$ and arc set $E = \{ (v_i^0, a) : i \in \N \} \cup \{ ( v_i^j, v_i^{j-1} ) : 1 \le j \le i \}$ (all arcs point towards $a$). Then the stairway to heaven is asymptotically meet-nilpotent, as there is no infinite pyramid.

The \Define{stairway to hell} has vertex set $V = \{ a \} \cup \{ v_i^j : j, i \in \N, j \le i  \}$ and arc set $E = \{ (a, v_i^0) : i \in \N \} \cup \{ ( v_i^{j-1}, v_i^j ) : 0 \le j \le i-1 \}$ (all arcs point away from $a$). Then the stairway to hell is not asymptotically meet-nilpotent, as it has the infinite pyramid $(W_i : i \in \N)$ where $W_0 = (a)$ and $W_{i+1} = (a, v_i^0, \dots, v_i^i)$ for all $i \ge 1$.

Those digraphs are illustrated in Figure \ref{figure:stairways}.
\end{example}

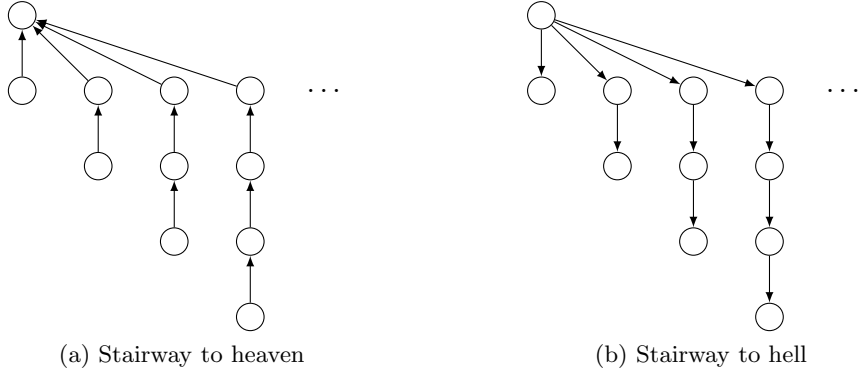
\begin{figure}
    \centering
\subfloat[Stairway to heaven]{
\begin{tikzpicture}
    
\begin{scope}[yshift = 0cm]
    \node[draw, circle] (p00) at (0,0) {};
    \node[draw, circle] (p11) at (0,-1) {};
    \node[draw, circle] (p21) at (1,-1) {};
    \node[draw, circle] (p22) at (1,-2) {};
    \node[draw, circle] (p31) at (2,-1) {};
    \node[draw, circle] (p32) at (2,-2) {};
    \node[draw, circle] (p33) at (2,-3) {};
    \node[draw, circle] (p41) at (3,-1) {};
    \node[draw, circle] (p42) at (3,-2) {};
    \node[draw, circle] (p43) at (3,-3) {};
    \node[draw, circle] (p44) at (3,-4) {};
    
    \node (p5) at (4,-1) {$\dots$};

    \draw[-latex] (p44) -- (p43);
    \draw[-latex] (p43) -- (p42);
    \draw[-latex] (p42) -- (p41);
    \draw[-latex] (p41) -- (p00);
    \draw[-latex] (p33) -- (p32);
    \draw[-latex] (p32) -- (p31);    
    \draw[-latex] (p31) -- (p00);
    \draw[-latex] (p22) -- (p21);
    \draw[-latex] (p21) -- (p00);
    \draw[-latex] (p11) -- (p00);
\end{scope}
\end{tikzpicture}
    
    

}%
\hspace{2cm}\subfloat[Stairway to hell]{
\begin{tikzpicture}
    
\begin{scope}[yshift = 0cm]
    \node[draw, circle] (p00) at (0,0) {};
    \node[draw, circle] (p11) at (0,-1) {};
    \node[draw, circle] (p21) at (1,-1) {};
    \node[draw, circle] (p22) at (1,-2) {};
    \node[draw, circle] (p31) at (2,-1) {};
    \node[draw, circle] (p32) at (2,-2) {};
    \node[draw, circle] (p33) at (2,-3) {};
    \node[draw, circle] (p41) at (3,-1) {};
    \node[draw, circle] (p42) at (3,-2) {};
    \node[draw, circle] (p43) at (3,-3) {};
    \node[draw, circle] (p44) at (3,-4) {};
    
    \node (p5) at (4,-1) {$\dots$};

    \draw[-latex] (p43) -- (p44);
    \draw[-latex] (p42) -- (p43);
    \draw[-latex] (p41) -- (p42);
    \draw[-latex] (p00) -- (p41);
    \draw[-latex] (p32) -- (p33);
    \draw[-latex] (p31) -- (p32);    
    \draw[-latex] (p00) -- (p31);
    \draw[-latex] (p21) -- (p22);
    \draw[-latex] (p00) -- (p21);
    \draw[-latex] (p00) -- (p11);
\end{scope}
\end{tikzpicture}
    
    

}
    \caption{The stairway to heaven is is asymptotically meet-nilpotent; the stairway to hell isn't. Both are fixed point-free.}
    \label{figure:stairways}
\end{figure}

We now characterise the lattices for which all fixed point-free graphs are meet-nilpotent.

\begin{proposition} \label{proposition:strongly_acyclic_nilpotent}
$\Lattice{L}$ is a DCC lattice if and only if all fixed point-free graphs on $\Lattice{L}$ are meet-nilpotent.
\end{proposition}

\begin{proof}
Let $\Lattice{L}$ be DCC. If $f$ is fixed point-free, then it is easily proved (by induction on $i$) that either $f^i(\one) = \zero$ or $f^{i+1}(\one) < f^i(\one)$. Thus the sequence $( f^i(\one) : i \in \N, f^i(\one) > 0 )$ forms a descending chain, and hence is finite, say of length $K$. Therefore, there exists $K$ such that $f^K( \one ) = \zero$.

If $L$ has an infinite descending chain $D = ( d_i : i \in \N )$ with $d_0 = \one$, then let 
\[
    f(x) = \bigjoin D \setminus x^\uparrow.
\]
We claim that $f$ is a graph, fixed point-free, but not meet-nilpotent.
\begin{itemize}
    \item Graph. Let $S \subseteq L$ and $q = \bigjoin S^\uparrow$, then we have $D \setminus q^\uparrow = \bigcup_{s \in S} D \setminus s^\uparrow$. We obtain
    \[
        f \left( \bigjoin S \right) = \bigjoin D \setminus q^\uparrow = \bigjoin_{s \in S} \bigjoin D \setminus s^\uparrow = \bigjoin f(S).
    \]
    
    \item Fixed point-free. Suppose $f(y) = y$ for $y \ne \zero$. Then $y \in D$, say $y = d_i$ and $f(y) = d_{i+1} < y$.
    
    \item Not meet-nilpotent. $f( d_i ) = d_{i+1}$ for all $i \in \N$, thus $f^K( \one ) = d_K > \zero$ for all $K \in \N$.
\end{itemize}
\end{proof}

Using a similar proof, we can characterise complete lattices where all fixed point-free graphs are asymptotically meet-nilpotent.

\begin{corollary} \label{corollary:strongly_acyclic_asymptotically_nilpotent}
Every descending chain $D = (d_i : i \in \N)$ in $\Lattice{L}$ satisfies $\bigmeet D = \zero$ if and only if all fixed point-free graphs on $\Lattice{L}$ are asymptotically meet-nilpotent.
\end{corollary}

We now provide different characterisations and properties of asymptotically meet-nilpotent graphs, which are summarised in Figure \ref{figure:asymptotically_meet-nilpotent}.

\begin{figure}
    \resizebox{\textwidth}{!}{
    \begin{tikzpicture}
        \node (c) at (0,0) [draw,thick,minimum width=2cm,minimum height=3cm,  text width = 3.2cm, align=center] {$f$ asymptotically meet-nilpotent\\~\\ $C(f) = \{ \zero \}$\\ $B(f) = \{ \zero \}$};
    
        \node (l) at (6,0) [draw,thick,minimum width=2cm,minimum height=3cm,  text width = 3.2cm, align=center] {$\Pyr(f) = \{ \zero \}$};
    
        \node (par) at (12,0) [draw,thick,minimum width=2cm,minimum height=3cm,  text width = 3.2cm, align=center] {$f$ fixed point-free\\~\\ $\Par(f) = \{ \zero \}$ \\ $\Fix(f) = \{ \zero \}$};

        \draw[-latex] (c) to node[above] {\tiny L.\ref{lemma:asymptotically_meet-nilpotent_Pyr}} (l);
        \draw[-latex] (l) -- (par);
    
        \draw[-latex, color=red, dotted] (l) to [bend right=15] node[above] {\tiny Ex.\ref{example:counterexample_pyramidal}} (c);
        \draw[-latex, color=green] (l) to [bend left=15] node[below] {\tiny O.\ref{observation:C(D)}} (c);
    
        \draw[-latex, color=red, dotted] (par) to [bend right=15] node[above] {\tiny Ex.\ref{example:stairways}} (l);
    \end{tikzpicture}
    }
    \caption{Implications for different kinds of fixed point-free graphs. The \textbf{\textcolor{red}{red dotted arrows}} represent counterexamples, the \textbf{\textcolor{green}{green arrows}} represent implications that hold for $\Lattice{L}$ being a power set lattice, while the \textbf{black arrows} represent unconditional implications.}
    \label{figure:asymptotically_meet-nilpotent}
\end{figure}
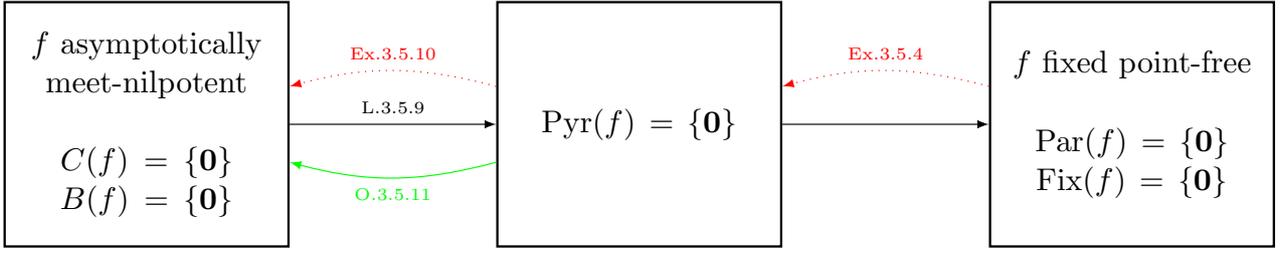

We first explain and prove the contents of the first box on Figure \ref{figure:asymptotically_meet-nilpotent}. For all $\phi: L \to L$, we denote
\begin{align*}
    B( \phi ) &= \bigcap_{i \in \N} \phi^i( L )^\downarrow,\\
    C( \phi ) &= \left( \bigmeet_{i \in \N} \phi^i( \one )  \right)^\downarrow.
\end{align*}

\begin{proposition} \label{proposition:asymptotically_meet-nilpotent_graph}
Let $f$ be a graph, then the following are equivalent.
\begin{equivalent}
    \item \label{item:asymptotically_meet-nilpotent_graph_A}
    $f$ is asymptotically meet-nilpotent, i.e. $\bigmeet_{i \in \N} f^i( \one ) = \zero$;

    \item \label{item:asymptotically_meet-nilpotent_graph_B}
    $B(f) = \{ \zero \}$.

    \item \label{item:asymptotically_meet-nilpotent_graph_C}
    $C(f) =  \{ \zero \}$.
\end{equivalent}
\end{proposition}

\begin{lemma} \label{lemma:B=C_graphs}
If $f$ is monotone, then $B( f ) = C( f )$.
\end{lemma} 

\begin{proof}
For any $\phi : L \to L$, let $D(\phi) = \{ x \in L : \exists (x_i : i \in \N), x = \bigmeet_{i \in \N} \phi^i( x_i )\}$. We prove that $B( \phi ) = D( \phi )$.
We have
\[
    B( \phi ) = \bigcap_{i \in \N} \phi^i( L )^\downarrow = \{ x \in L : \exists (x_i : i \in \N), x \le \phi^i( x_i ) \ \forall i \in \N \}.
\]
Let $x \in B( \phi )$ with $x \le \phi^i( x_i )$ for all $i \in \N$. Define $x'_0 = x$ and $x'_i = x_i$ for all $i \ge 1$, then $x = \bigmeet_{i \in \N} \phi^i( x'_i )$, hence $x \in D( \phi )$. Conversely, let $y \in D( \phi )$ with $y = \bigmeet_{i \in \N} \phi^i( y_i )$, then $y \le \phi^i( y_i )$ for all $i \in \N$ and $y \in D( \phi )$.

We now prove the lemma. Let $b \in B( f ) = D( f )$, then there exists $(b_i : i \in \N)$ such that
\[
    b = \bigmeet_{i \in \N} f^i( b_i ) \le \bigmeet_{i \in \N} f^i( \one ),
\]
and hence $b \in C( f )$. Conversely, let $c \in C( f )$, then $c = c \meet \bigmeet_{i \in \N} f^i( \one )$. Therefore $c = \bigmeet_{i \in \N} f^i( c_i )$ for $c_0 = 0$ and $c_i = \one$ for all $i \ge 1$, thus $c \in B( f )$.
\end{proof}

\begin{proof}[Proof of Proposition \ref{proposition:asymptotically_meet-nilpotent_graph}]
$\ref{item:asymptotically_meet-nilpotent_graph_A} \iff \ref{item:asymptotically_meet-nilpotent_graph_C}$. For all $x \in L$, $x = \zero$ if and only if $x^\downarrow = \{ \zero \}$. Applying this to $x = \bigmeet_{i \in \N} f^i( 
\one )$ yields the desired equivalence.

$\ref{item:asymptotically_meet-nilpotent_graph_B} \iff \ref{item:asymptotically_meet-nilpotent_graph_C}$. Follows from Lemma \ref{lemma:B=C_graphs}.
\end{proof}

We now prove the relationships between the first two boxes on Figure \ref{figure:asymptotically_meet-nilpotent}. We begin with the black and red arrows.

\begin{lemma} \label{lemma:asymptotically_meet-nilpotent_Pyr}
Let $f$ be graph. If $f$ is asymptotically meet-nilpotent, then $\Pyr( f ) = \{ \zero \}$. 
\end{lemma}

\begin{proof}
Let $f$ be asymptotically meet-nilpotent and let $x \in \Pyr( f )$, with $f^i( x_i ) = x$ for all $i \in \N$. We then have
\[
    \zero = \bigmeet_{i \in \N} f^i( \one ) \ge \bigmeet_{i \in \N} f^i( x_i ) = x.
\]
\end{proof}

\begin{example} \label{example:counterexample_pyramidal}
On the other hand, there exists $\Lattice{L}$ and $f$ such that $\Pyr( f ) = \{ \zero \}$ and $f$ is not asymptotically meet-nilpotent. Let $L = [0,1]$ equipped with the natural order, $u \in (0,1)$ and
\[
    f(x) = \begin{cases}
        \frac{x}{2} &\text{if } x \le u \\
        \frac{ x + u }{ 2 } & \text{otherwise}.
    \end{cases}
\]
It is easy to check the following properties: $f$ is a graph; $\Pyr( f ) = \{ 0 \}$; $\bigmeet_{i \in \N} f^i(1) = u$.
\end{example}

The result on power set lattices is transparent.

\begin{observation} \label{observation:C(D)}
For any digraph $D$, we have
\[
    C(D) = \Pyr(D)^\downarrow.
\]
Therefore, $D$ is asymptotically meet-nilpotent if and only if $\Pyr(D) = \{ \emptyset \}$.
\end{observation}

\section{Robert's theorem for mappings over complete lattices} \label{section:robert}

\subsection{Four kinds of dependency} \label{section:dependency}

Robert's theorem for finite automata networks relies on the concept of interaction graph, i.e. the smallest graph on which a network depends. Unfortunately, the interaction graph does not necessarily exist for general lattices, not even for $\Lattice{L} = \Lattice{\Powerset}(V)$ with infinite $V$, as we shall see in Example \ref{example:omega}. 
Instead, we generalise the concept, reviewed in Section \ref{subsection:robert}, of a mapping depending on a digraph to all lattices in two ways, thus yielding four kinds of dependency.

Let $\Lattice{L}$ be a complete lattice and $\phi, \alpha, \beta, \gamma, \delta : L \to L$. 
\begin{dependency}
    \item \label{item:pre-meet-depends}
    We say $\phi$ \Define{pre-meet-depends} on $\alpha$ if for all $x,y, s \in L$,
    \[
        x \meet \alpha(s) = y \meet \alpha(s) \implies \phi(x) \meet s = \phi(y) \meet s.
    \]

    \item \label{item:pre-join-depends}
    We say $\phi$ \Define{pre-join-depends} on $\beta$ if for all $x,y, s \in L$,
    \[
        x \join \beta(s) = y \join \beta(s) \implies \phi(x) \join s = \phi(y) \join s.
    \]

    \item \label{item:post-meet-depends}
    We say $\phi$ \Define{post-meet-depends} on $\gamma$ if for all $x,y, s \in L$,
    \[
        x \meet s = y \meet s \implies \phi(x) \meet \gamma(s) = \phi(y) \meet \gamma(s).
    \]

    \item \label{item:post-join-depends}
    We say $\phi$ \Define{post-join-depends} on $\delta$ if for all $x,y, s \in L$,
    \[
        x \join s = y \join s \implies \phi(x) \join \delta(s) = \phi(y) \join \delta(s).
    \]

\end{dependency}

It is clear that any mapping $\phi$ pre-meet-depends on the mapping $\alpha^*$ with $\alpha^*(x) = \one$ for all $x \in L$. Therefore, it is interesting to understand the minimal mappings on which $\phi$ pre-meet-depends. The lemma below gathers the corresponding observations for all four kinds of dependency.

\begin{lemma}
Let $\phi, \alpha, \bar{\alpha}, \beta, \bar{\beta}, \gamma, \bar{\gamma}, \delta, \bar{\delta} : L \to L$. Then the following hold.
\begin{dependency}
    \item \label{item:alpha_bar}
    If $\phi$ pre-meet-depends on $\alpha$ and $\bar{\alpha} \ge \alpha$, then $\phi$ pre-meet-depends on $\bar{\alpha}$.

    \item \label{item:beta_bar}
    If $\phi$ pre-join-depends on $\beta$ and $\bar{\beta} \le \beta$, then $\phi$ pre-join-depends on $\bar{\beta}$.

    \item \label{item:gamma_bar}
    If $\phi$ post-meet-depends on $\gamma$ and $\bar{\gamma} \le \gamma$, then $\phi$ post-meet-depends on $\bar{\gamma}$.
    
    \item \label{item:delta_bar}
    If $\phi$ post-join-depends on $\delta$ and $\bar{\delta} \ge \delta$, then $\phi$ post-join-depends on $\bar{\delta}$.
\end{dependency}
\end{lemma}

We illustrate in the next example why it makes sense to generalise dependency to all mappings, instead of only graphs, even for power set lattices.

\begin{example} \label{example:omega}
Let $\Lattice{L} = \Lattice{\Powerset}(V)$ for some infinite $V$. Define $\omega : L \to L$ by
\[
    \omega(x) = \begin{cases}
    \zero & \text{if } x \text{ is finite} \\
    \one & \text{otherwise}. 
    \end{cases}
\]
We can readily characterise on which mappings $\omega$ depends; the proof is a simple exercise.
\begin{dependency}
    \item \label{item:omega_alpha}
    $\omega$ pre-meet-depends on $\alpha$ if and only if $\alpha(s)$ is co-finite for all $s \ne \zero$.

    \item \label{item:omega_beta}
    $\omega$ pre-join-depends on $\beta$ if and only if $\beta(t)$ is finite for all $t \ne \one$.

    \item \label{item:omega_gamma}
    $\omega$ post-meet-depends on $\gamma$ if and only if $\gamma(s) = \zero$ if $s$ is not co-finite.

    \item \label{item:omega_delta}
    $\omega$ post-join-depends on $\delta$ if and only if $\delta(t) = \one$ if $t$ is infinite (i.e. $\delta \ge \omega$).
\end{dependency}

As a consequence, we now give an example of mapping $\hat{\alpha}$ such that $\omega$ pre-meet-depends on $\hat{\alpha}$ but $\omega$ does not pre-meet-depend on any monotone mapping (let alone graph) $f \le \hat{\alpha}$. Namely, let $\hat{\alpha}$ such that $\hat{\alpha}(s) = s$ for all $s$ co-finite and $\hat{\alpha}(s) = \one$ otherwise. Let $f \le \hat{\alpha}$ be monotone, then for any finite $s$, 
\[
    f( s ) \le \bigmeet \{ f( x ) : s \le x, x \text{ co-finite} \} \le \bigmeet \{ \hat{\alpha}( x ) : s \le x, x \text{ co-finite} \} =  \bigmeet \{ x : s \le x, x \text{ co-finite} \} = s.
\]
Thus $f(s)$ is finite, and hence $\omega$ does not pre-meet-depend on $f$.
\end{example}

As seen from Example \ref{example:omega}, even when $\Lattice{L}$ is a power set lattice, we cannot assume without loss that $\phi$ pre-meet-depends on a monotone mapping $\alpha$. However, whenever $\Lattice{L}$ is a frame, we can assume without loss of generality that $\phi$ post-meet-depends on a monotone mapping $\gamma$.

\begin{lemma} \label{lemma:frame_monotone_gamma}
Let $\Lattice{L}$ be a frame and $\gamma : L \to L$. Then there exists a monotone mapping $\Gamma : L \to L$ such that $\Gamma \ge \gamma$ and for all $\phi: L \to L$, $\phi$ post-meet-depends on $\gamma$ if and only if $\phi$ post-meet-depends on $\Gamma$.
\end{lemma}

\begin{proof}
Define $\Gamma : L \to L$ by $\Gamma(a) = \bigjoin_{b \le a} \gamma(b)$. We note that $\Gamma \ge \gamma$ and that $\Gamma$ is monotone. Suppose $\phi$ post-meet-depends on $\gamma$. For all $x, y, s \in L$, we have
\begin{align*}
    x \meet s = y \meet s 
    &\iff x \meet t = y \meet t \text{ for all } t \le s\\
    &\implies \phi(x) \meet \gamma(t) = \phi(y) \meet \gamma(t) \text{ for all } t \le s\\
    &\implies \bigjoin_{t \le s} \left( \phi(x) \meet \gamma(t) \right) = \bigjoin_{t \le s} \left( \phi(y) \meet \gamma(t) \right)\\
    &\iff \phi(x) \meet \Gamma(s) = \phi(y) \meet \Gamma(s).
\end{align*}
Thus, $\phi$ post-meet-depends on $\Gamma$. Conversely, suppose $\phi$ post-meet-depends on $\Gamma$. Since $\gamma \le \Gamma$, $\phi$ post-meet-depends on $\gamma$ as well.
\end{proof}

The four kinds of dependency are related to one another, thanks to residuation.

\begin{lemma} \label{lemma:pre-post-dependence}
Let $\phi, \alpha, \beta, \gamma, \delta : L \to L$.
\begin{dependency}
    \item \label{item:pre-post-meet-minus}
    If $\phi$ post-meet-depends on $\Residual{ \alpha }$, then $\phi$ pre-meet-depends on $\alpha$.

    \item \label{item:pre-post-join-minus}
    If $\phi$ post-join-depends on $\Residuated{ \beta }$, then $\phi$ pre-join-depends on $\beta$.

    \item \label{item:pre-post-meet-plus}
    If $\phi$ pre-meet-depends on $\Residuated{ \gamma }$, then $\phi$ post-meet-depends on $\gamma$.
    
    \item \label{item:pre-post-join-plus}
    If $\phi$ pre-join-depends on $\Residual{ \delta }$, then $\phi$ post-join-depends on $\delta$.
\end{dependency}
\end{lemma}

\begin{proof}
\begin{dependency}
    \item \label{item:pre-post-minus_proof}
    If $\phi$ post-meet-depends on $\Residual{ \alpha }$, then for all $x,y,s \in L$
    \[
        x \meet \alpha(s) = y \meet \alpha(s) \implies \phi(x) \meet \Residual{ \alpha }( \alpha(s) ) = \phi(y) \meet \Residual{ \alpha }( \alpha(s) ) \implies \phi(x) \meet s = \phi(y) \meet s,
    \]
    whence $\phi$ pre-meet-depends on $\alpha$.
    
    \item \label{item:pre-post-join-minus_proof}
    Same proof as above for $\Opposite{ \Lattice{L} }$.
    
    \item \label{item:pre-post-plus_proof}
    If $\phi$ pre-meet-depends on $\Residuated{ \gamma }$, then for all $x,y,s \in L$
    \[
        x \meet s = y \meet s \implies x \meet \Residuated{ \gamma }( \gamma(s) ) = y \meet \Residuated{ \gamma }( \gamma(s) ) \implies \phi(x) \meet \gamma(s) = \phi(y) \meet \gamma(s),
    \]
    whence $\phi$ post-meet-depends on $\gamma$.
    
    \item \label{item:pre-post-join-plus_proof}
    Same proof as above for $\Opposite{ \Lattice{L} }$.

\end{dependency}
\end{proof}

\begin{corollary} \label{corollary:graph_pre-post}
Let $\phi : L \to L$ and let $f$ be a graph on $\Lattice{L}$.
\begin{enumerate}
    \item \label{item:pre-post-meet-graph}
    $\phi$  pre-meet-depends on $f$ if and only if $\phi$ post-meet-depends on $\Residual{f}$.

    \item \label{item:pre-post-join-graph}
    $\phi$ post-join-depends on $f$ if and only if $\phi$ pre-join-depends on $\Residual{f}$.
\end{enumerate}
\end{corollary}

Let $f$ be a graph. We remark that $f$ pre-join-depends on $\Residual{f}$:
\begin{align*}
    x \join \Residual{ f }(s) = y \join \Residual{ f }(s) &\implies f( x \join \Residual{ f }(s) ) = f( y \join \Residual{ f }(s) ) \\
    &\implies f( x ) \join f \Residual{ f }(s) = f( y ) \join f \Residual{ f }(s) \\
    &\implies f( x ) \join s = f( y ) \join s.
\end{align*}
In fact, the dependency relations between $f$ and $\Residual{f}$ are represented in Figure \ref{figure:dependencies1}, where the head of each arc depends on the tail of the arc. For instance, the fact that $f$ pre-join-depends on $\Residual{f}$ is the blue arc at the top. The proofs of the different dependencies are straightforward and hence omitted.

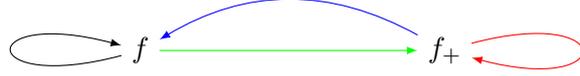
\begin{figure}
    \centering
    
    \begin{tikzpicture}[scale=4]
        \node (f) at (0,1) {$f$};
        \node (fm) at (1,1) {$\Residual{f}$};
        
        \draw[color=green,-latex] (f) -- (fm);
        
        \draw[color=red,-latex] (fm) to [loop right] (fm);
        \draw[color=black,-latex] (f) to [loop left] (f);
        \draw[color=blue,-latex] (fm) to [bend right] (f);
    \end{tikzpicture}
    
    \caption{Dependencies between $f$ and $\Residual{f}$. Each colour indicates a different kind of dependency: \textbf{\textcolor{green}{green for pre-meet}}, \textbf{\textcolor{blue}{blue for pre-join}}, \textbf{\textcolor{red}{red for post-meet}}, and \textbf{\textcolor{black}{black for post-join}}.}
    \label{figure:dependencies1}
\end{figure}

\subsection{Meet-metric and meet-complete mappings} \label{section:discriminating}

The aim of this subsection is to define mappings for which we shall be able to apply the (bounded or complete) Banach contraction principle. Firstly, the meet-metric property provides us with a contractive mapping on a bounded metric space (not necessarily complete) and hence guarantees that there is at most one pyramidal point; accordingly, this will also be the basis for the feedback bound. Secondly, the meet-complete property provides us with a bounded complete metric space and hence guarantees strong convergence towards a fixed point. Lastly, the weakest property is meet-covering, which is perhaps the simplest to state from a lattice point of view, and is actually equivalent to meet-complete for frames.

Consider a sequence $A = (a_i : i \in \N)$ of elements of $L$. For all $x \in L$ we denote $A_x = ( x \meet a_i : i \in \N )$, so that $A = A_\one$. We remark that $x \ge \bigjoin A_x$; we shall use this fact repeatedly in the sequel. For any $s,t \in L$, let $\Delta_A(s,t) = \{ i \in \N : s \meet a_i \ne t \meet a_i \}$ and $d_A(s,t) = \sum_{i \in \Delta_A(s,t)} 2^{-i}$. Then $d_A$ is a bounded pseudometric, i.e. it satisfies all the axioms of a metric apart from the identity of indiscernibles. We then say that $A$ is
\begin{enumerate}
    \item \Define{meet-covering} if $d_A( x, \one ) \ne 0$ for all $x \ne \one$;

    \item \Define{meet-metric} if $d_A$ is a bounded metric;

    \item \Define{meet-complete} if $(L,d_A)$ is a bounded complete metric space.
\end{enumerate}  
\Define{Join-covering}, \Define{join-metric} and \Define{join-complete} sequences are defined similarly, by placing ourselves in the opposite lattice $\Opposite{ \Lattice{L} }$.

We first give purely lattice-theoretic characterisations of these three properties.

\begin{proposition} \label{proposition:meet-covering}
Let $A$ be a sequence of elements of $L$. Then $A$ is meet-covering if and only if $\bigjoin A = \one$.
\end{proposition}

\begin{proof}
Suppose $\bigjoin A  = \one$. If $d_A( x, \one ) = 0$, then $A_x = A_\one = A$ and hence
\[
    x \ge \bigjoin A_x = \bigjoin A = \one.
\]
Conversely, suppose $d_A( x, \one ) \ne 0$ for all $x \ne \one$. Suppose $b = \bigjoin A < \one$, then $A_b = A = A_\one$ and hence $b = \one$.
\end{proof}

\begin{proposition} \label{proposition:meet-metric}
Let $A$ be a sequence of elements of $L$. The following are equivalent:
\begin{equivalent}
    \item \label{item:meet-metric_definition}
     $x = \bigjoin A_x$ for all $x \in L$;

    \item \label{item:meet-metric_injectivitiy}
    $A_x = A_y$ implies $x = y$ for all $x,y \in L$;

    \item \label{item:meet-metric_metric}
    $A$ is meet-metric, i.e. $d_A$ is a bounded metric over $L$.
\end{equivalent}
\end{proposition}

\begin{proof}
$\ref{item:meet-metric_definition} \iff \ref{item:meet-metric_injectivitiy}$. For any sequence $A$ and any $x \in L$, let $z = \bigjoin A_x$. On the one hand, $x \ge z$, hence $z \meet a_i \le x \meet a_i$, and on the other hand
\[
    z \meet a_i = a_i \meet \bigjoin_j (x \meet a_j) \ge a_i \meet (x \meet a_i) = x \meet a_i,
\]
thus $A_z = A_x$. Therefore, if $A$ does not satisfy Property \ref{item:meet-metric_definition}, we have $A_z = A_x$ while $z \ne x$. Conversely, if $A$ satisfies Property \ref{item:meet-metric_definition}, then for all $x,y \in L$ with $A_x = A_y$, we have $x = \bigjoin A_x = \bigjoin A_y = y$.

$\ref{item:meet-metric_injectivitiy} \iff \ref{item:meet-metric_metric}$. The identity of indiscernibles for $d_A$ is equivalent to Property \ref{item:meet-metric_injectivitiy}.
\end{proof}

Say a sequence $Z = (z_i : i \in \N)$ is \Define{strongly Cauchy} (for $A$) if $i \ge j$ implies $z_i \meet a_j = z_j \meet a_j$. Define $A_Z = ( z_i \meet a_i : i \in \N )$. A simple exercise verifies that a strongly Cauchy sequence for $A$ is indeed a Cauchy sequence for $d_A$.

\begin{proposition} \label{proposition:metric_Cauchy}
A sequence $A$ is meet-complete if and only if it is meet-metric and for all strongly Cauchy $Z$, $A_Z = A_z$ for $z = \bigjoin A_Z$.
\end{proposition} 

\begin{proof}
Let $A$ be meet-complete. Then $A$ is meet-metric by definition. Now, let $Z = (z_i : i \in \N)$ be strongly Cauchy and suppose $A_b \ne A_Z$ for all $b \in L$. Then there exists $j \in \N$ such that $b \meet a_j \ne z_j \meet a_j = z_i \meet a_j$ for all $i \ge j$. Therefore, $d_A( b, z_i ) \ge 2^{-j}$ for all $i \ge j$ and $Z$ does not converge towards $b$, which contradicts the fact that $( L, d_A )$ is a complete metric space. Thus, there exists $b \in L$ with $A_b = A_Z$. But then $b = \bigjoin A_b$, which finally gives $b = \bigjoin A_b = \bigjoin A_Z = z$.

Conversely, suppose $A$ is meet-metric and for all strongly Cauchy $Z$, $A_Z = A_z$. We now prove that $(L, d_A)$ forms a complete metric space. Let $Y = ( y_i : i \in \N )$ be a Cauchy sequence. For all $k \in \N$, let $N( k ) = \min\{ M : d_A( y_i, y_j ) < 2^{-k} \ \forall i,j \ge M \}$ and define the sequence $Z = ( z_k = y_{ N(k) } : k \in \N )$. Then for all $i \ge N( k )$, we have $y_i \meet a_k = y_{ N(k) } \meet a_k = z_k \meet a_k$. In particular for $i = N( l )$ for some $l \ge k$, we obtain $z_l \meet a_k = z_k \meet a_k$, that is $Z$ is a strongly Cauchy sequence with $A_Z = A_z$.   
Now, let $\epsilon > 0$ and $k$ such that $2^{-k} < \epsilon$, then for all $i \ge N( k )$, $d_A( z, y_i ) < \epsilon$, that is $Y$ converges towards $z$.
\end{proof}

We now give sufficient conditions for a sequence to be meet-complete.

\begin{proposition} \label{proposition:discriminating}
If one of the following holds, then $A$ is meet-complete:
\begin{enumerate}
    \item \label{item:discriminating_un}
    $a_n = \one$ for some $n \in \N$;
    
    \item \label{item:discriminating_T}
    $\Lattice{L}$ is distributive and $\bigjoin T = \one$ for some finite $T \subseteq A$;
    
    \item \label{item:discriminating_A}
    $\Lattice{L}$ is a frame and $\bigjoin A = \one$.
\end{enumerate}
\end{proposition}

\begin{proof}
\begin{enumerate}
    \item \label{item:discriminating_un_proof}
    We already proved that $\bigjoin A_x \le x$ for all $x$ in the proof of Proposition \ref{proposition:meet-metric}. If $a_n = \one$, then $\bigjoin A_x = \bigjoin_i x \meet a_i \ge x \meet a_n = x$. Let $Z$ be strongly Cauchy, then for all $i \ge n$, $z_i = z_i \meet a_n = z_n \meet a_n = z_n$, so the sequence is eventually stationary. For all $j \le n$, we have $z_n \meet a_j = z_j \meet a_j$ so that $A_Z = A_{z_n}$ and $z_n = \bigjoin A_Z$.

    \item \label{item:discriminating_T_proof}
    Without loss, suppose $T = \{1, \dots, \tau\}$. For any $x \in L$, we have
    \[
        \bigjoin (x \meet T) =  x \meet \bigjoin T = x.  
    \]
    Let $Z$ be strongly Cauchy. If $t \le \tau \le s$, we have
    \[
        z_s \meet a_t = z_t \meet a_t = z_\tau \meet a_t,
    \]
    thus $z_s \meet \bigjoin T = z_\tau \meet \bigjoin T$, whence $z_s = z_\tau$ for all $s \ge \tau$. For all $t \le \tau$, we have $z_\tau \meet z_t = z_t \meet a_t$ so that $A_Z = A_{z_\tau}$ and $z_\tau = \bigjoin A_Z$.

    \item \label{item:discriminating_A_proof}
    For any $x \in L$, we have
    \[
        \bigjoin A_x = \bigjoin (x \meet A) = x \meet \bigjoin A = x.  
    \]
    Let $Z$ be strongly Cauchy, and let $z = \bigjoin A_Z$, then
    \[
        z \meet a_i 
        = a_i \meet \bigjoin (z_j \meet a_j) 
        \ge a_i \meet z_i \meet a_i 
        = z_i \meet a_i
    \]
    and on the other hand
    \[
        z \meet a_i 
        = a_i \meet \bigjoin_j (z_j \meet a_j) 
        = \bigjoin_{k < i} (z_k \meet a_k \meet a_i)  \join  \bigjoin_{j \ge i} (z_j \meet a_i \meet a_j)
        = \bigjoin_{k < i} (z_i \meet a_k \meet a_i)  \join  \bigjoin_{j \ge i} (z_i \meet a_i \meet a_j) 
        \le z_i \meet a_i.
    \]
    Thus $A_z = A_Z$.
\end{enumerate}
\end{proof}

\begin{corollary} \label{corollary:frame_meet-metric}
If $\Lattice{L}$ is a frame and $A$ is meet-covering, then $A$ is meet-complete.
\end{corollary}

We can now define four kinds of -complete mappings, each corresponding to a kind of dependency. Let $\Lattice{L}$ be a complete lattice and $\alpha, \beta, \gamma, \delta : L \to L$.
\begin{dependency}
    \item \label{item:pre-meet-complete}
    Say $\alpha : L \to L$ is \Define{pre-meet-complete} if there exists a meet-complete sequence $A = (a_i : i \in \N)$ with $a_0 = \zero$ and $\alpha(a_i) \le a_{i-1}$ for all $i \ge 1$. 
    
    \item \label{item:pre-join-complete}
    Say $\beta : L \to L$ is \Define{pre-join-complete} if there exists a join-complete sequence $B = (b_i : i \in \N)$ with $b_0 = \one$ and $\beta(b_i) \ge b_{i-1}$ for all $i \ge 1$. 

    \item \label{item:post-meet-complete}
    Say $\gamma : L \to L$ is \Define{post-meet-complete} if there exists a meet-complete sequence $C = (c_i : i \in \N)$ with $c_0 = \zero$ and $\gamma(c_{i-1}) \ge c_i$ for all $i \ge 1$. 

    \item \label{item:post-join-complete}
    Say $\delta : L \to L$ is \Define{post-join-complete} if there exists a join-complete sequence $D = (d_i : i \in \N)$ with $d_0 = \one$ and $\delta(d_{i-1}) \le d_i$ for all $i \ge 1$. 
\end{dependency}
We then say $\alpha$ is $A$-pre-meet-complete, and similarly for $\beta$, $\gamma$, and $\delta$. The corresponding -covering and -metric mappings are defined similarly.

For any $\psi : L \to L$, we define the following two sequences:
\begin{align*}
    \MeetSequence( \psi )   &= ( \psi^i( \zero ) : i \in \N ) \\
    \JoinSequence( \psi ) &= ( \psi^i( \one ) : i \in \N ).
\end{align*}
We note that $\psi$ is asymptotically join-nilpotent ($\bigjoin \psi^i( \zero ) = \one$) if and only if $\MeetSequence( \psi )$ is meet-covering. Similarly, $\psi$ is asymptotically meet-nilpotent ($\bigmeet \psi^i( \one ) = \zero$) if and only if $\JoinSequence( \psi )$ is join-covering. Those two sequences are also closely related to -covering/metric/complete mappings, as seen in Proposition \ref{proposition:metric_nilpotent} below.


    
    


\begin{proposition} \label{proposition:metric_nilpotent}
Let $\Lattice{L}$ be a complete lattice and $\alpha, \beta, \gamma, \delta : L \to L$. Then the following hold.
\begin{dependency}
    \item \label{item:pre-meet-covering}
    If $\alpha$ is pre-meet-covering/metric, then $\MeetSequence( \Residual{\alpha} )$ is meet-covering/metric. Conversely, if $\alpha$ is a graph and $\MeetSequence( \Residual{ \alpha } )$ is meet-covering/metric/complete, then $\alpha$ is pre-meet-covering/metric/complete.
    
    \item \label{item:pre-join-covering}
    If $\beta$ is pre-join-covering/metric, then $\JoinSequence( \Residuated{\beta} )$ is join-covering/metric. Conversely, if $\beta$ is a co-graph and $\JoinSequence( \Residuated{ \beta } )$ is join-covering/metric/complete, then $\beta$ is pre-join-covering/metric/complete.
    
    \item \label{item:post-meet-covering}
    If $\gamma$ is monotone and post-meet-covering/metric, then $\MeetSequence( \gamma )$ is meet-covering/metric. Conversely, if $\MeetSequence( \gamma )$ is meet-covering/metric/complete, then $\gamma$ is post-meet-covering/metric/complete.

    \item \label{item:post-join-covering}
    If $\delta$ is monotone and post-join-covering/metric, then $\JoinSequence( \delta )$ is join-covering/metric. Conversely, if $\JoinSequence( \delta )$ is join-covering/metric/complete, then $\delta$ is post-join-covering/metric/complete.
\end{dependency}
\end{proposition}

The proof uses the following lemma.

\begin{lemma} \label{lemma:dominating_sequences}
Let $A$ and $B$ are two sequences of elements in $L$ and suppose $B \ge A$. If $A$ is meet-covering, then so is $B$. Similarly, if $A$ is meet-metric, then so is $B$.
\end{lemma}

\begin{proof}
We prove that $d_B \ge d_A$. Let $x,y \in L$, then for all $i \in \N$,
\[
    x \meet b_i = y \meet b_i \implies x \meet b_i \meet a_i = y \meet b_i \meet a_i \implies x \meet a_i = y \meet a_i.
\]
Thus $\Delta_A(x, y) \subseteq \Delta_B(x,y)$ and $d_B(x,y) \ge d_A(x,y)$.

Therefore, if $A$ is meet-covering, then $d_B(x, \one) \ge d_A( x, \one ) > 0$ for all $x \ne \one$ and $B$ is also meet-covering. The proof is similar for meet-metric.
\end{proof}

\begin{proof}[Proof of Proposition \ref{proposition:metric_nilpotent}]
\begin{dependency}
    \item 
    Let $\MeetSequence( \Residual{ \alpha } ) = (m_i = { \Residual{\alpha} }^i( \zero ) : i \in \N)$. If $\alpha$ is pre-meet-covering/metric, then let $A = (a_i : i \in \N)$ be a meet-covering/metric sequence with $a_0 = \zero$ and $\alpha(a_i) \le a_{i-1}$ for all $i \ge 1$. It is easily shown by induction that $a_i \le m_i$ for all $i \in \N$, thus $\MeetSequence( \Residual{ \alpha } )$ is pre-meet-covering/metric by Lemma \ref{lemma:dominating_sequences}.  Conversely, if $\alpha$ is a graph and $\MeetSequence( \Residual{ \alpha } )$ is meet-covering/metric/complete, then $m_0 = \zero$, and $\alpha( m_i ) = \alpha( \Residual{\alpha}( m_{i-1} ) ) \le m_{i-1}$ for all $i \ge 1$, and hence $\alpha$ is pre-meet-covering/metric/complete.

    \item
    Same proof as above for $\Opposite{ \Lattice{L} }$.
    
    \item
    Let $\MeetSequence( \gamma ) = (m_i = \gamma^i( \zero ) : i \in \N)$. If $\gamma$ is post-meet-covering/metric, then let $C = (c_i : i \in \N)$ be a meet-covering/metric sequence with $c_0 = \zero$ and $\gamma( c_{i-1} ) \ge c_i$ for all $i \ge 1$. Then if $\gamma$ is monotone, it is easily shown by induction that $c_i \le m_i$ for all $i \in \N$, thus $\MeetSequence( \gamma )$ is meet-covering/metric. Conversely, if $\MeetSequence( \gamma )$ is meet-covering/metric/complete, then $m_0 = \zero$, and $\gamma( m_{i-1} ) = m_i$ for all $i \ge 1$, and hence $\gamma$ is post-meet-covering/metric/complete.
    
    \item
    Same proof as above for $\Opposite{ \Lattice{L} }$.
\end{dependency}
\end{proof}

    


    



\subsection{Robert's theorem for mappings on complete lattices} \label{section:robert_complete_lattices}

In this subsection, we give our generalised Robert's theorem. We shall study its tightness in the following subsection. Robert's theorem comes in four kinds, one for each kind of dependency. The overall picture for pre-meet-dependency is displayed in Figure \ref{figure:robert_alpha}, while the overall picture for post-join-dependency is displayed in Figure \ref{figure:robert_delta}.

\begin{theorem}[Robert's theorem for mappings over complete lattices] \label{theorem:generalised_robert}
Let $\Lattice{L}$ be a complete lattice and let $\phi, \alpha, \beta, \gamma, \delta : L \to L$.
\begin{dependency}
    \item \label{item:robert_pre-meet}
    Suppose $\alpha$ is $A$-pre-meet-metric and $\phi$ pre-meet-depends on $\alpha$. The following hold.
    \begin{enumerate}
        \item \label{item:robert_pre-meet_formula_e}
        There exists $e \in L$ such that $e = \bigjoin_i \phi^i( x ) \meet a_i$ for all $x \in L$.
    
        \item \label{item:robert_pre-meet_injective}
        If $\phi$ has a pyramidal point, then $\phi$ strongly converges to $e$.
        
        \item \label{item:robert_pre-meet_metric}
        If $\alpha$ is $A$-pre-meet-complete, then $\phi$ strongly converges to $e$.
    
        \item \label{item:robert_pre-meet_nilpotent}
        If $\alpha$ is $K$-meet-nilpotent, then $e = \phi^K(x)$ for all $x \in L$.
    \end{enumerate}

    \item \label{item:robert_pre-join}
    Suppose $\beta$ is $B$-pre-join-metric and $\phi$ pre-join-depends on $\beta$. The following hold.
    \begin{enumerate}
        \item \label{item:robert_pre-join_formula_e}
        There exists $e \in L$ such that $e = \bigmeet_i \phi^i( x ) \join b_i$ for all $x \in L$.
    
        \item \label{item:robert_pre-join_injective}
        If $\phi$ has a pyramidal point, then $\phi$ strongly converges to $e$.
        
        \item \label{item:robert_pre-join_metric}
        If $\beta$ is $B$-pre-join-complete, then $\phi$ strongly converges to $e$.
    
        \item \label{item:robert_pre-join_nilpotent}
        If $\beta$ is $K$-join-nilpotent, then $e = \phi^K(x)$ for all $x \in L$.
    \end{enumerate}

    \item \label{item:robert_post-meet}
    Suppose $\gamma$ is $C$-post-meet-metric and $\phi$ post-meet-depends on $\gamma$. The following hold.
    \begin{enumerate}
        \item \label{item:robert_post-meet_formula_e}
        There exists $e \in L$ such that $e = \bigjoin_i \phi^i( x ) \meet c_i$ for all $x \in L$. 
    
        \item \label{item:robert_post-meet_injective}
        If $\phi$ has a pyramidal point, then $\phi$ strongly converges to $e$.
        
        \item \label{item:robert_post-meet_metric}
        If $\gamma$ is $C$-post-meet-complete, then $\phi$ strongly converges to $e$.
    
        \item \label{item:robert_post-meet_nilpotent}
        If $\gamma$ is $K$-join-nilpotent, then $e = \phi^K(x)$ for all $x \in L$. 
    \end{enumerate}
    
    \item \label{item:robert_post-join}
    Suppose $\delta$ is $D$-post-join-metric and $\phi$ post-join-depends on $\delta$. The following hold.
    \begin{enumerate}
        \item \label{item:robert_post-join_formula_e}
        There exists $e \in L$ such that $e = \bigmeet_i \phi^i( x ) \join d_i$ for all $x \in L$.
    
        \item \label{item:robert_post-join_injective}
        If $\phi$ has a pyramidal point, then $\phi$ strongly converges to $e$.
        
        \item \label{item:robert_post-join_metric}
        If $\delta$ is $D$-post-join-complete, then $\phi$ strongly converges to $e$.
    
        \item \label{item:robert_post-join_nilpotent}
        If $\delta$ is $K$-meet-nilpotent, then $e = \phi^K(x)$ for all $x \in L$.
    
    \end{enumerate}
\end{dependency}
\end{theorem}

\begin{proof}
We only give the proofs for \ref{item:robert_pre-meet}; the other proofs are similar and hence omitted.
\begin{enumerate}
    \item 
    Let $x, y \in L$. By induction on $i$, we prove that $\phi^i(x) \meet a_i = \phi^i(y) \meet a_i$. Therefore $\bigjoin_i \phi^i( x ) \meet a_i = \bigjoin_i \phi^i( y ) \meet a_i$, which is a point on the lattice.
    
    \item 
    We prove that for any $x,y \in L$, $\Delta_A( \phi(x), \phi(y) ) \subseteq \Delta_A( x, y ) + 1$. First, since $a_0 = \zero$, $\phi(x) \meet a_0 = \phi(y) \meet a_0$ and $0 \notin \Delta_A( \phi(x), \phi(y) )$. Second, for $i \ge 1$ we have
    \[
        x \meet a_{i-1} = y \meet a_{i-1} \implies x \meet \alpha( a_i ) = y \meet \alpha( a_i ) \implies \phi(x) \meet a_i = \phi(y) \meet a_i.
    \]
    Therefore, if $j \in \Delta_A( \phi(x), \phi(y) )$, then $j-1 \in \Delta_A( x, y )$. 
    
    This proves that the mapping $\phi$ is $\frac{1}{2}$-contractive for $d_A$. The bounded Banach contraction principle then shows that if $\phi$ has a pyramidal point, then it strongly converges to its unique fixed point. We prove that the only possible fixed point is $e$. Let $u \in \Fix( \phi )$, then 
    \[
        u = \bigjoin A_u = \bigjoin_i u \meet a_i = \bigjoin_i \phi^i( u ) \meet a_i = e.
    \]

    \item 
    The complete Banach contraction principle guarantees the existence of a fixed point in this case.

    \item
    It is easily shown by induction on $0 \le i \le K$, that for all $x, y \in L$, $\phi^i(x) \meet \alpha^{K-i}( \one ) = \phi^i(y) \meet \alpha^{K-i}( \one )$. In particular, for $i = K$ we obtain $\phi^K( x ) = \phi^K( y )$.
\end{enumerate}
\end{proof}

\begin{figure}
\resizebox{\textwidth}{!}{
\begin{tikzpicture}[xscale=5, yscale=4]
    \node (apn) at (0,0) [draw,thick,minimum width=2cm,minimum height=2cm, text width = 3.2cm, align=center] {$\Residual{\alpha}$\\ $K$-join-nilpotent};
    \node (apm) at (1,0) [draw,thick,minimum width=2cm,minimum height=2cm, text width = 3.2cm, align=center] {$\MeetSequence( \Residual{ \alpha } )$\\ meet-complete};
    \node (api) at (2,0) [draw,thick,minimum width=2cm,minimum height=2cm, text width = 3.2cm, align=center] {$\MeetSequence( \Residual{ \alpha } )$\\ meet-metric};
    \node (apc) at (3,0) [draw,thick,minimum width=2cm,minimum height=2cm, text width = 3.2cm, align=center] {$\Residual{\alpha}$\\ asymptotically join-nilpotent};

    \draw[-latex, color=black] (apn) -- (apm);
    \draw[-latex, color=black] (apm) -- (api);
    \draw[-latex, color=black] (api) -- (apc);

    \draw[-latex, color=purple] (api) to [bend left=15] node[below] {\tiny P.\ref{proposition:discriminating}\ref{item:discriminating_A}} (apm);
    \draw[-latex, color=purple] (apc) to [bend left=15] node[below] {\tiny P.\ref{proposition:discriminating}\ref{item:discriminating_A}} (api);

    \node (an) at (0,1) [draw,thick,minimum width=2cm,minimum height=2cm,  text width = 3.2cm, align=center] {$\alpha$\\ $K$-meet-nilpotent};
    \node (am) at (1,1) [draw,thick,minimum width=2cm,minimum height=2cm, text width = 3.2cm, align=center] {$\alpha$\\  pre-meet-complete};
    \node (ai) at (2,1) [draw,thick,minimum width=2cm,minimum height=2cm, text width = 3.2cm, align=center] {$\alpha$\\ pre-meet-metric};
    \node (ac) at (3,1) [draw,thick,minimum width=2cm,minimum height=2cm, text width = 3.2cm, align=center] {$\alpha$\\ pre-meet-covering};

    \draw[-latex, color=black] (an) -- (am);
    \draw[-latex, color=black] (am) -- (ai);
    \draw[-latex, color=black] (ai) -- (ac);

    \draw[-latex, color=purple] (ai) to [bend left=15] node[below] {\tiny P.\ref{proposition:discriminating}\ref{item:discriminating_A}} (am);
    \draw[-latex, color=purple] (ac) to [bend left=15] node[below] {\tiny P.\ref{proposition:discriminating}\ref{item:discriminating_A}} (ai);

    \node (pn) at (0,2) [draw,thick,minimum width=3cm,minimum height=2cm, text width = 3.2cm, align=center] {$\phi^K = e$};
    \node (pm) at (1,2) [draw,thick,minimum width=3cm,minimum height=2cm, text width = 3.2cm, align=center] {$\phi \StronglyConverges e$};
    \node (pi) at (2,2) [draw,thick,minimum width=3cm,minimum height=2cm, text width = 3.2cm, align=center] {$\Pyr(\phi) \subseteq \{ e \}$};
    \node (pf) at (3,2) [draw,thick,minimum width=3cm,minimum height=2cm, text width=3.2cm, align=center] {$|\Fix(\phi)| \le 1$};

    \draw[-latex, color=black] (pn) -- (pm);
    \draw[-latex, color=black] (pm) -- (pi);
    \draw[-latex, color=black] (pi) -- (pf);

    \draw[-latex, color=blue] (apn) to [bend right=30] node[right] {\tiny L.\ref{lemma:nilpotent}\ref{item:nilpotent_graph}} (an);
    \draw[-latex, color=blue] (apm) to [bend right=30] node[right] {\tiny P.\ref{proposition:metric_nilpotent}\ref{item:pre-meet-covering} } (am);
    \draw[-latex, color=blue] (api) to [bend right=30] node[right] {\tiny P.\ref{proposition:metric_nilpotent}\ref{item:pre-meet-covering} } (ai);
    \draw[-latex, color=blue] (apc) to [bend right=30] node[right] {\tiny P.\ref{proposition:metric_nilpotent}\ref{item:pre-meet-covering} } (ac);

    \draw[-latex, color=black] (an) to node[left] {\tiny L.\ref{lemma:nilpotent}\ref{item:nilpotent_minus}} (apn);
    \draw[-latex, color=purple] (am) to [bend right=30] (apm);
    \draw[-latex, color=black] (ai) to node[left] {\tiny P.\ref{proposition:metric_nilpotent}\ref{item:pre-meet-covering} } (api);
    \draw[-latex, color=black] (ac) to node[left] {\tiny P.\ref{proposition:metric_nilpotent}\ref{item:pre-meet-covering} } (apc);

    \draw[-latex, color=black] (an) to node[left] {\tiny Th.\ref{theorem:robert}\ref{item:robert_pre-meet}\ref{item:robert_pre-meet_nilpotent} } (pn);
    \draw[-latex, color=black] (am) to node[left] {\tiny Th.\ref{theorem:robert}\ref{item:robert_pre-meet}\ref{item:robert_pre-meet_metric} } (pm);
    \draw[-latex, color=black] (ai) to node[left] {\tiny Th.\ref{theorem:robert}\ref{item:robert_pre-meet}\ref{item:robert_pre-meet_injective} } (pi);
    
    \draw[-latex, color=red, dotted] (am) to node[left] {\tiny Ex.\ref{example:counterexample_nilpotent}} (pn);
    \draw[-latex, color=red, dotted] (ai) to node[left] {\tiny Ex.\ref{example:counterexample_convergence} } (pm);
    \draw[-latex, color=red, dotted] (ac) to node[left] {\tiny Ex.\ref{example:counterexample_limit} } (pf);

\end{tikzpicture}
}
\caption{Robert's theorem for pre-meet-dependency. 
The \textbf{\textcolor{red}{red dotted arrows}} represent counterexamples, the \textbf{\textcolor{purple}{purple arrows}} represent implications that hold when $\Lattice{L}$ is a frame, the \textbf{\textcolor{blue}{blue arrows}} represent implications that hold when $\alpha$ is a graph, and the \textbf{\textcolor{black}{black arrows}} represent unconditional implications.}
\label{figure:robert_alpha} 
\end{figure}
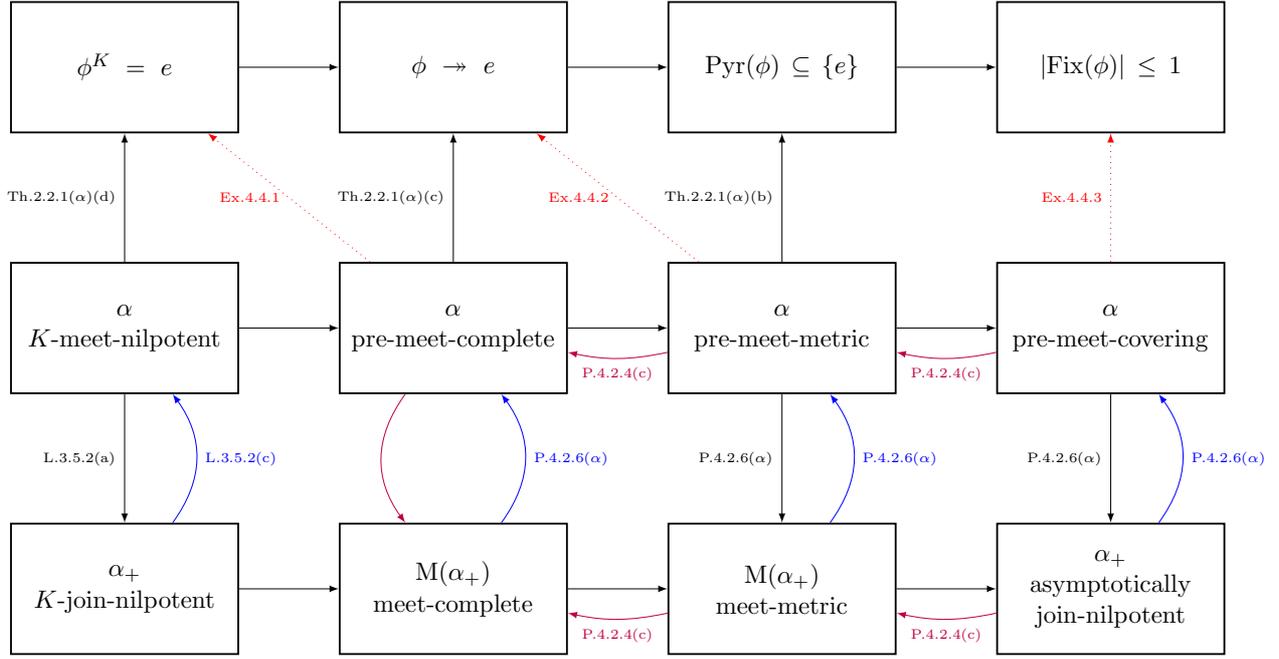

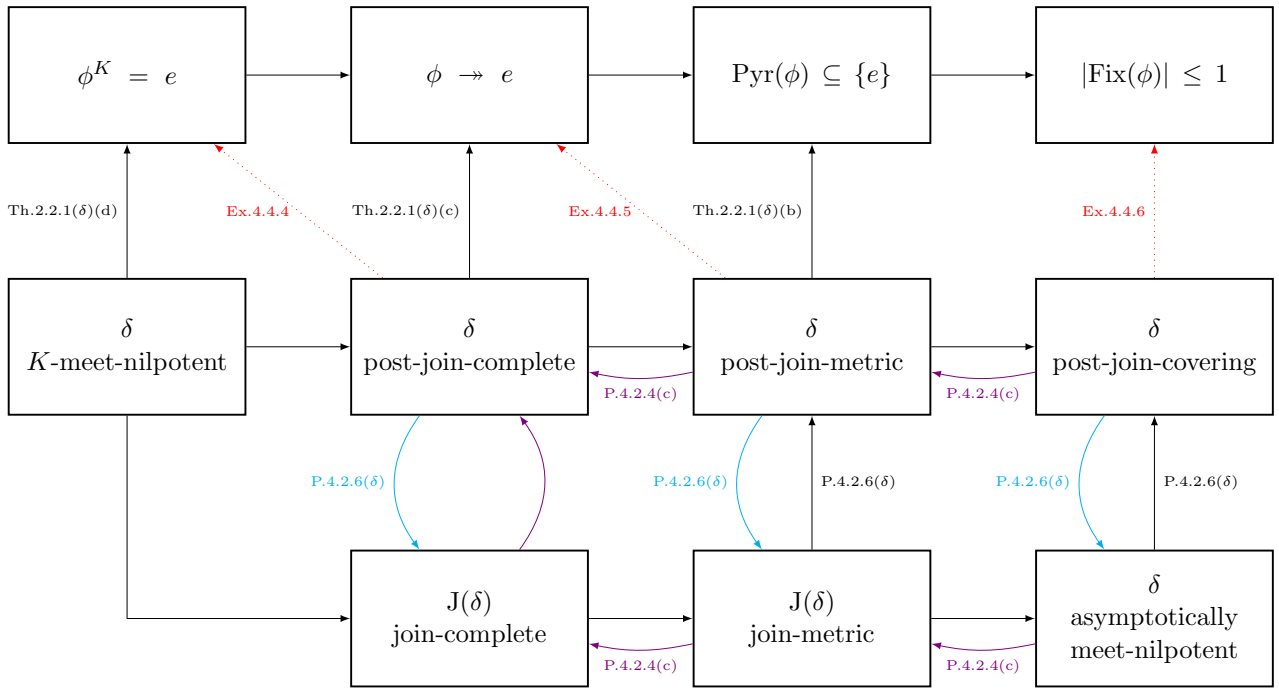
\begin{figure}
\resizebox{\textwidth}{!}{
\begin{tikzpicture}[xscale=5, yscale=4]
    \node (apm) at (1,0) [draw,thick,minimum width=2cm,minimum height=2cm, text width = 3.2cm, align=center] {$\JoinSequence( \delta )$\\ join-complete};
    \node (api) at (2,0) [draw,thick,minimum width=2cm,minimum height=2cm, text width = 3.2cm, align=center] {$\JoinSequence( \delta )$\\ join-metric};
    \node (apc) at (3,0) [draw,thick,minimum width=2cm,minimum height=2cm, text width = 3.2cm, align=center] {$\delta$\\ asymptotically meet-nilpotent};

    \draw[-latex, color=black] (apm) -- (api);
    \draw[-latex, color=black] (api) -- (apc);

    \draw[-latex, color=violet] (api) to [bend left=15] node[below] {\tiny P.\ref{proposition:discriminating}\ref{item:discriminating_A}} (apm);
    \draw[-latex, color=violet] (apc) to [bend left=15] node[below] {\tiny P.\ref{proposition:discriminating}\ref{item:discriminating_A}} (api);

    \node (an) at (0,1) [draw,thick,minimum width=2cm,minimum height=2cm,  text width = 3.2cm, align=center] {$\delta$\\ $K$-meet-nilpotent};
    \node (am) at (1,1) [draw,thick,minimum width=2cm,minimum height=2cm, text width = 3.2cm, align=center] {$\delta$\\  post-join-complete};
    \node (ai) at (2,1) [draw,thick,minimum width=2cm,minimum height=2cm, text width = 3.2cm, align=center] {$\delta$\\ post-join-metric};
    \node (ac) at (3,1) [draw,thick,minimum width=2cm,minimum height=2cm, text width = 3.2cm, align=center] {$\delta$\\ post-join-covering};

    \draw[-latex, color=black] (an) -- (am);
    \draw[-latex, color=black] (am) -- (ai);
    \draw[-latex, color=black] (ai) -- (ac);

    \draw[-latex, color=violet] (ai) to [bend left=15] node[below] {\tiny P.\ref{proposition:discriminating}\ref{item:discriminating_A}} (am);
    \draw[-latex, color=violet] (ac) to [bend left=15] node[below] {\tiny P.\ref{proposition:discriminating}\ref{item:discriminating_A}} (ai);

    \node (pn) at (0,2) [draw,thick,minimum width=3cm,minimum height=2cm, text width=3.2cm, align=center] {$\phi^K = e$};
    \node (pm) at (1,2) [draw,thick,minimum width=3cm,minimum height=2cm, text width=3.2cm, align=center] {$\phi \StronglyConverges e$};
    \node (pi) at (2,2) [draw,thick,minimum width=3cm,minimum height=2cm, text width=3.2cm, align=center] {$\Pyr(\phi) \subseteq \{ e \}$};
    \node (pf) at (3,2) [draw,thick,minimum width=3cm,minimum height=2cm, text width=3.2cm, align=center] {$|\Fix(\phi)| \le 1$};

    \draw[-latex, color=black] (pn) -- (pm);
    \draw[-latex, color=black] (pm) -- (pi);
    \draw[-latex, color=black] (pi) -- (pf);

    \draw[-latex, color=cyan] (am) to [bend right=30] node[left] {\tiny P.\ref{proposition:metric_nilpotent}\ref{item:post-join-covering} } (apm);
    \draw[-latex, color=cyan] (ai) to [bend right=30] node[left] {\tiny P.\ref{proposition:metric_nilpotent}\ref{item:post-join-covering} } (api);
    \draw[-latex, color=cyan] (ac) to [bend right=30] node[left] {\tiny P.\ref{proposition:metric_nilpotent}\ref{item:post-join-covering} } (apc);

    \draw[-latex, color=violet] (apm) to [bend right=30] (am);
    \draw[-latex, color=black] (an) |- (apm);
    \draw[-latex, color=black] (api) to node[right] {\tiny P.\ref{proposition:metric_nilpotent}\ref{item:post-join-covering} }  (ai);
    \draw[-latex, color=black] (apc) to node[right] {\tiny P.\ref{proposition:metric_nilpotent}\ref{item:post-join-covering} }  (ac);

    \draw[-latex, color=black] (an) to node[left] {\tiny Th.\ref{theorem:robert}\ref{item:robert_post-join}\ref{item:robert_post-join_nilpotent} } (pn);
    \draw[-latex, color=black] (am) to node[left] {\tiny Th.\ref{theorem:robert}\ref{item:robert_post-join}\ref{item:robert_post-join_metric} } (pm);
    \draw[-latex, color=black] (ai) to node[left] {\tiny Th.\ref{theorem:robert}\ref{item:robert_post-join}\ref{item:robert_post-join_injective} } (pi);
    
    \draw[-latex, color=red, dotted] (am) to node[left] {\tiny Ex.\ref{example:counterexample_complete_delta}} (pn);
    \draw[-latex, color=red, dotted] (ai) to node[left] {\tiny Ex.\ref{example:counterexample_metric_delta} } (pm);
    \draw[-latex, color=red, dotted] (ac) to node[left] {\tiny Ex.\ref{example:counterexample_covering_delta} } (pf);

\end{tikzpicture}
}
\caption{Robert's theorem for post-join-dependency. 
The \textbf{\textcolor{red}{red dotted arrows}} represent counterexamples, the \textbf{\textcolor{violet}{violet arrows}} represent implications that hold when $\Lattice{L}$ is a locale, the \textbf{\textcolor{cyan}{cyan arrows}} represent implications that hold when $\delta$ is monotone, and the \textbf{\textcolor{black}{black arrows}} represent unconditional implications.}
\label{figure:robert_delta} 
\end{figure}

\subsection{Tightness of Robert's theorem} \label{subsection:limitations_robert}

In this section, we show that the hypotheses in Robert's theorem cannot be relaxed much further without affecting the corresponding results. 

We first give three counterexamples for pre-meet-dependency (item \ref{item:robert_pre-meet} in Theorem \ref{theorem:generalised_robert}).

First, we give a counterexample to: ``If $\alpha$ is pre-meet-complete, then $\phi^K = e$ for some $K$.'' We give it for a power set lattice.

\begin{example} \label{example:counterexample_nilpotent}
Let $\phi$ be the backwards ray and $\alpha$ be the ray.  Then it is easily shown that $\phi$ pre-meet-depends on $\alpha$ (we shall make this connection more obvious in Section \ref{section:dependencies_boolean}), and that $\alpha$ is pre-meet-complete. However, we have $\phi^K( \N ) = \N_K$ for all $K$.
\end{example}

Second, we give a counterexample to: ``If $\alpha$ is pre-meet-metric, then $\phi \StronglyConverges e$.'' That statement holds for frames, and we give a counterexample for a non-distributive lattice.

\begin{example} \label{example:counterexample_convergence}
Let $\Lattice{L}$ be given as follows. Let $L = A \cup B \cup \{ \one \}$, where $A = (a_i : i \in \N)$ and $B = (b_i : i \in \N)$ are infinite ascending chains and $a_0 = b_0 = \zero$. We then have $a_i \meet b_j = b_k$ where $k = \min\{ i-1, j \}$. 

It is easily shown that $A$ is meet-metric, i.e. $x = \bigjoin_i x \meet a_i$ for all $x \in L$. Let us show that $A$ is not meet-complete. The sequence $B$ is strongly Cauchy for $A$, since $i \ge j$ implies $b_i \meet a_j = b_{ j-1 } = b_j \meet a_j$. However, it does not have a limit: if $x = b_j$, then $x \meet a_{j+1} = b_j$ while $b_i \meet a_{j+1} = b_{j+1}$ for all $i \ge j+1$; if $x \notin B$, then $x \meet a_i \notin B$ while $b_j \meet a_i \in B$.

Let $\alpha: L \to L$ be defined as $\alpha(a_{i+1}) = a_i$ for all $i \in \N$ and $\alpha( x ) = \one$ for all $x \notin A$. Since $a_0 = \zero$ and $\alpha(a_i) \le a_{i-1}$ for all $i \ge 1$, $\alpha$ is pre-meet-metric.

Let $\phi : L \to L$ be defined as 
\[
    \phi(x) = \begin{cases}
        b_{i+1} &\text{if } x = b_i \\
        \zero &\text{otherwise}.
    \end{cases}
\]
Then $\phi$ has no pyramidal points. All that is left to show is that $\phi$ pre-meet-depends on $\alpha$. Let $x, y, s \in L$ with $x \ne y$ and suppose $x \meet \alpha(s) = y \meet \alpha(s)$; we need to prove that $\phi(x) \meet s = \phi(y) \meet s$. This is clear for $s = \zero$ or $s \notin A$. Therefore, we can assume $s = a_{i+1}$ for some $i \ge 0$ (and hence $\alpha(s) = a_i$). We now perform a case analysis on $x$ and $y$.
\begin{enumerate}
    \item 
    $x, y \notin B$. Then $\phi(x) = \phi(y)$.

    \item
    $x \notin B, y \in B$. We have $i = 0$, for otherwise $x \meet a_i \in A$ while $y \meet a_i \in B$, hence $x \meet \alpha(s) \ne y \meet \alpha(s)$. Therefore, $\phi(x) \meet a_1 = \phi(y) \meet a_1 = \zero$.

    \item
    $x,y \in B$. Then $x = b_k$ and $y = b_l$ with $k,l \ge i-1$. We then have
    \[
        \phi(x) \meet a_{i+1} = b_{k+1} \meet a_{i+1} = b_i = b_{l+1} \meet a_{i+1} = \phi(y) \meet a_{i+1}.
    \]
\end{enumerate}
This example is illustrated in Figure \ref{figure:counterexample_convergence}.
\end{example}

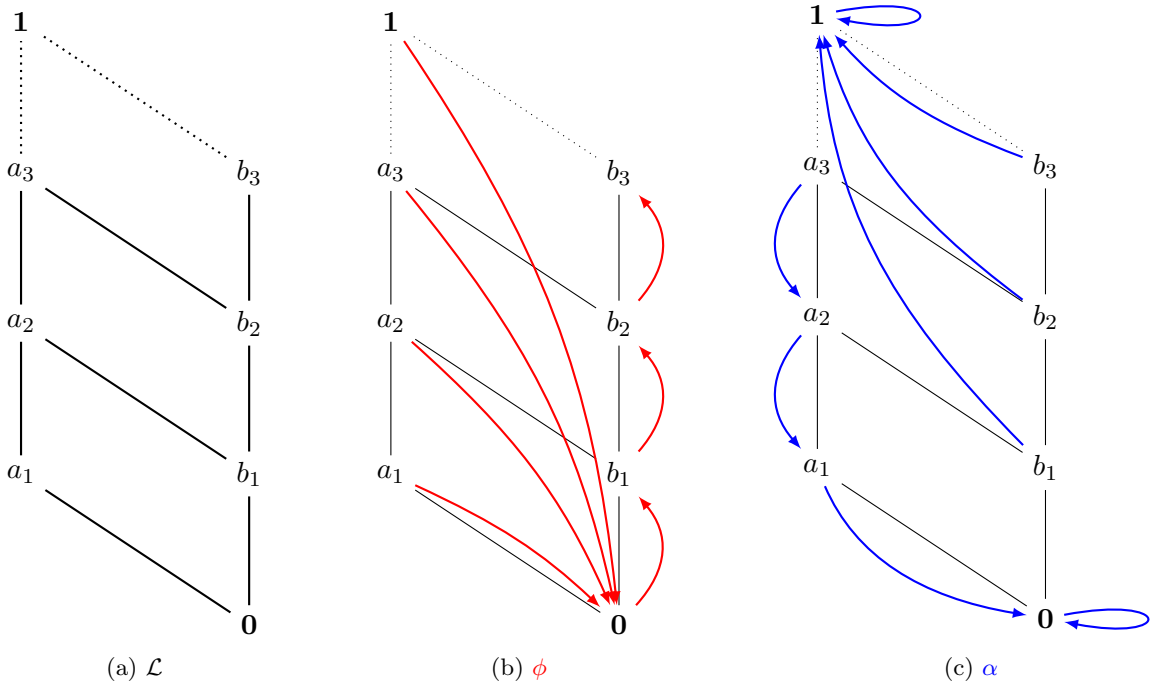
\begin{figure}[!htp]
\centering
\subfloat[$\Lattice{L}$]{

\begin{tikzpicture}[xscale=3, yscale=2]

    \node (one) at (0,4) {$\one$};
    \node (a3) at (0,3) {$a_3$};
    \node (a2) at (0,2) {$a_2$};
    \node (a1) at (0,1) {$a_1$};
    \node (b3) at (1,3) {$b_3$};
    \node (b2) at (1,2) {$b_2$};
    \node (b1) at (1,1) {$b_1$};
    \node (zero) at (1,0) {$\zero$};

    \draw[thick] (zero) -- (b1);
    \draw[thick] (zero) -- (a1);
    \draw[thick] (a1) -- (a2);
    \draw[thick] (a2) -- (a3);
    \draw[thick] (b1) -- (b2);
    \draw[thick] (b2) -- (b3);
    \draw[thick] (b1) -- (a2);
    \draw[thick] (b2) -- (a3);    
    \draw[thick, dotted] (a3) -- (one);    
    \draw[thick, dotted] (b3) -- (one);

\end{tikzpicture}
}
\hspace{1cm}\subfloat[\textcolor{red}{$\phi$}]{
\begin{tikzpicture}[xscale=3, yscale=2]

    \node (one) at (0,4) {$\one$};
    \node (a3) at (0,3) {$a_3$};
    \node (a2) at (0,2) {$a_2$};
    \node (a1) at (0,1) {$a_1$};
    \node (b3) at (1,3) {$b_3$};
    \node (b2) at (1,2) {$b_2$};
    \node (b1) at (1,1) {$b_1$};
    \node (zero) at (1,0) {$\zero$};


    \draw (zero) -- (b1);
    \draw (zero) -- (a1);
    \draw (a1) -- (a2);
    \draw (a2) -- (a3);
    \draw (b1) -- (b2);
    \draw (b2) -- (b3);
    \draw (b1) -- (a2);
    \draw (b2) -- (a3);    
    \draw[dotted] (a3) -- (one);    
    \draw[dotted] (b3) -- (one);

    \draw[-latex, red, thick] (zero) to [bend right=30] (b1);
    \draw[-latex, red, thick] (b1) to [bend right=30] (b2);
    \draw[-latex, red, thick] (b2) to [bend right=30] (b3);
    \draw[-latex, red, thick] (a1) to [bend left=10] (zero);
    \draw[-latex, red, thick] (a2) to [bend left=10] (zero);
    \draw[-latex, red, thick] (a3) to [bend left=10] (zero);
    \draw[-latex, red, thick] (one) to [bend left=10] (zero);


\end{tikzpicture}
}%
\hspace{1cm}\subfloat[\textcolor{blue}{$\alpha$}]{
\begin{tikzpicture}[xscale=3, yscale=2]

    \node (one) at (0,4) {$\one$};
    \node (a3) at (0,3) {$a_3$};
    \node (a2) at (0,2) {$a_2$};
    \node (a1) at (0,1) {$a_1$};
    \node (b3) at (1,3) {$b_3$};
    \node (b2) at (1,2) {$b_2$};
    \node (b1) at (1,1) {$b_1$};
    \node (zero) at (1,0) {$\zero$};

    \draw (zero) -- (b1);
    \draw (zero) -- (a1);
    \draw (a1) -- (a2);
    \draw (a2) -- (a3);
    \draw (b1) -- (b2);
    \draw (b2) -- (b3);
    \draw (b1) -- (a2);
    \draw (b2) -- (a3);    
    \draw[dotted] (a3) -- (one);    
    \draw[dotted] (b3) -- (one);

    \draw[-latex, blue, thick] (one) to [loop right] (one);
    \draw[-latex, blue, thick] (b3) to [bend left=15] (one);
    \draw[-latex, blue, thick] (a3) to [bend right=30] (a2);
    \draw[-latex, blue, thick] (b2) to [bend left=15] (one);
    \draw[-latex, blue, thick] (a2) to [bend right=30] (a1);
    \draw[-latex, blue, thick] (b1) to [bend left=15] (one);
    \draw[-latex, blue, thick] (a1) to [bend right=30] (zero);
    \draw[-latex, blue, thick] (zero) to [loop right] (zero);
    

\end{tikzpicture}
}
\caption{The counterexample given in Example \ref{example:counterexample_convergence}.}
\label{figure:counterexample_convergence}
\end{figure}

Third, we give a counterexample to: ``If $\alpha$ is pre-meet-covering, then $| \Fix( \phi ) | \le 1$.'' Again, this statement holds for frames and our counterexample is for a non-distributive lattice. The counterexample actually has two fixed points but replacing $\{ b \}$ in the example below by any arbitrary antichain $B$ yields as many fixed points as we want.

\begin{example} \label{example:counterexample_limit}
Let $\Lattice{L}$ be defined as follows. Let $L = A \cup \{ b \} \cup \{ \one \}$, where $A = (a_i : i \in \N)$ is an infinite ascending chain with $a_0 = \zero$ and $\bigjoin A = \one$, and $b \meet a_i = \zero$, $b \join a_i = \one$ for all $a_i \in A$.

Let $\alpha$ be the graph with $\alpha( \zero ) = \zero$, $\alpha( a_i ) = a_{i-1}$ for all $i \ge 1$, and $\alpha( b ) = \alpha( \one ) = \one$. Let $\phi$ with
\[
    \phi(x) = \begin{cases}
        b & \text{if } x = b\\
        \zero & \text{otherwise}.
    \end{cases}
\]
Then $\Fix(\phi) = \{ \zero, b \}$. All that is left to prove is that $\phi$ pre-meet-depends on $\alpha$. Let $x, y, s \in L$ with $x \ne y$ and suppose $x \meet \alpha(s) = y \meet \alpha(s)$; we need to prove that $\phi(x) \meet s = \phi(y) \meet s$. This is clear for $s \notin \{b, \one\}$, since in that case $\phi(u) \meet s = \zero$ for all $u \in L$. If $s \in \{b, \one\}$, then $\alpha(s) = \one$, thus $x = y$ and we are done.

This example is illustrated in Figure \ref{figure:counterexample_limit}.
\end{example}

\begin{figure}[!htp]
    \centering

\subfloat[$\Lattice{L}$]{
\begin{tikzpicture}[xscale=3, yscale=2]

    \node (one) at (0,4) {$\one$};
    \node (a3) at (0,3) {$a_3$};
    \node (a2) at (0,2) {$a_2$};
    \node (a1) at (0,1) {$a_1$};
    
    \node (b) at (1,2) {$b$};
    
    \node (zero) at (1,0) {$\zero$};

    \draw[thick] (zero) -- (b);
    \draw[thick] (zero) -- (a1);
    \draw[thick] (a1) -- (a2);
    \draw[thick] (a2) -- (a3);
    \draw[thick, dotted] (a3) -- (one);    
    \draw[thick] (b) -- (one);

\end{tikzpicture}
}%
\hspace{1cm}\subfloat[\textcolor{red}{$\phi$}]{
\begin{tikzpicture}[xscale=3, yscale=2]

    \node (one) at (0,4) {$\one$};
    \node (a3) at (0,3) {$a_3$};
    \node (a2) at (0,2) {$a_2$};
    \node (a1) at (0,1) {$a_1$};
    
    \node (b) at (1,2) {$b$};
    
    \node (zero) at (1,0) {$\zero$};

    \draw (zero) -- (b);
    \draw (zero) -- (a1);
    \draw (a1) -- (a2);
    \draw (a2) -- (a3);
    \draw[dotted] (a3) -- (one);    
    \draw (b) -- (one);

    \draw[-latex, red, thick] (zero) to [loop left] (zero);
    \draw[-latex, red, thick] (b) to [loop above] (b);
    \draw[-latex, red, thick] (a1) to [bend left=10] (zero);
    \draw[-latex, red, thick] (a2) to [bend left=10] (zero);
    \draw[-latex, red, thick] (a3) to [bend left=10] (zero);
    \draw[-latex, red, thick] (one) to [bend left=10] (zero);

\end{tikzpicture}
}%
\hspace{1cm}\subfloat[\textcolor{blue}{$\alpha$}]{
\begin{tikzpicture}[xscale=3, yscale=2]

    \node (one) at (0,4) {$\one$};
    \node (a3) at (0,3) {$a_3$};
    \node (a2) at (0,2) {$a_2$};
    \node (a1) at (0,1) {$a_1$};
    
    \node (b) at (1,2) {$b$};
    
    \node (zero) at (1,0) {$\zero$};

    \draw (zero) -- (b);
    \draw (zero) -- (a1);
    \draw (a1) -- (a2);
    \draw (a2) -- (a3);
    \draw[dotted] (a3) -- (one);    
    \draw (b) -- (one);

    \draw[-latex, blue, thick] (one) to [loop right] (one);
    \draw[-latex, blue, thick] (b) to [bend right=30] (one);
    \draw[-latex, blue, thick] (a3) to [bend right=30] (a2);
    \draw[-latex, blue, thick] (a2) to [bend right=30] (a1);
    \draw[-latex, blue, thick] (a1) to [bend right=30] (zero);
    \draw[-latex, blue, thick] (zero) to [loop right] (zero);

\end{tikzpicture}
}
    \caption{The counterexample given in Example \ref{example:counterexample_limit}.}
    \label{figure:counterexample_limit}
\end{figure}
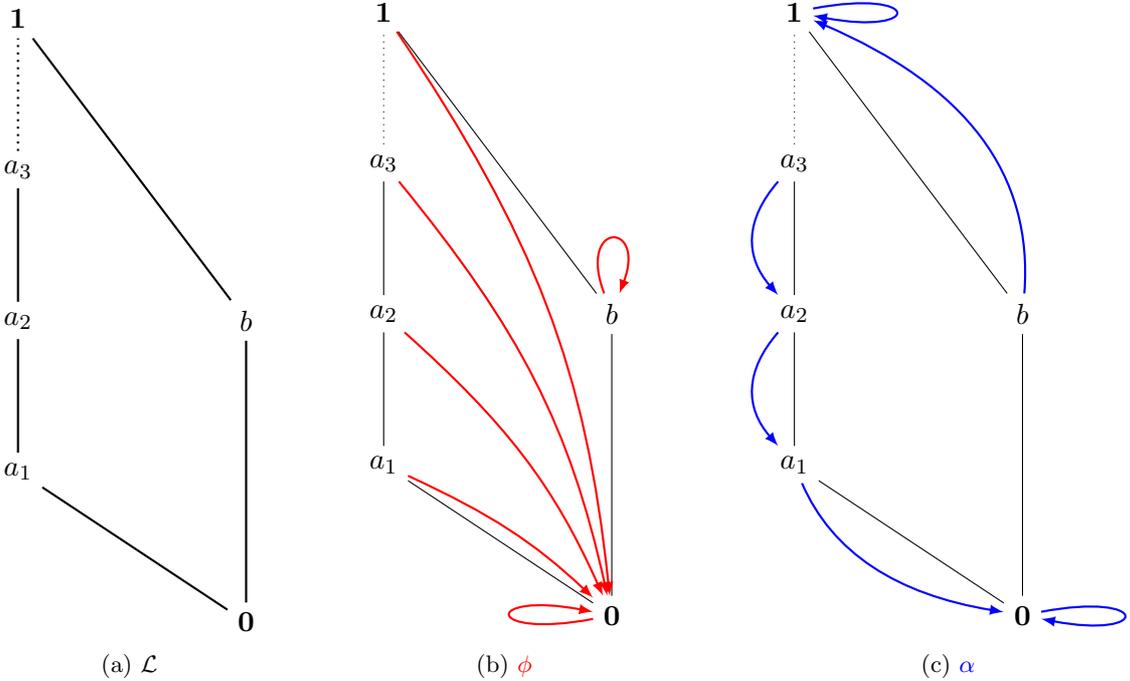

We now give three counterexamples for post-join-dependency (item \ref{item:robert_post-join} in Theorem \ref{theorem:generalised_robert}). All those counterexamples will be closely linked to their counterparts for pre-meet-dependency; as such, we shall omit the proofs.

First, we give a counterexample to: ``If $\delta$ is post-join-complete, then $\phi^K = e$.''

\begin{example} \label{example:counterexample_complete_delta}
Let $\phi = \delta$ be the backwards ray. Then $\delta$ is post-join-complete, $\phi$ satisfies $\phi^i( \N ) = \N_i$ for all $i \in \N$, and $\phi$ post-join-depends on $\delta$.
\end{example}

Second, we give a counterexample to: ``If $\delta$ is post-join-metric, then $\phi \StronglyConverges e$.''

\begin{example} \label{example:counterexample_metric_delta}
We consider the lattice opposite to the one in Example \ref{example:counterexample_convergence}. Let $L = A \cup B \cup \{ \zero \}$, where $A = (a_i : i \in \N)$ and $B = (b_i : i \in \N)$ are infinite descending chains and $a_0 = b_0 = \one$. We then have $a_i \join b_j = b_k$ where $k = \min\{ i-1, j \}$. Then $A$ is a join-metric but not join-complete sequence.

Let $\delta: L \to L$ be defined as $\delta(a_i) = a_{i+1}$ for all $i \in \N$ and $\delta( x ) = \one$ for all $x \notin A$. Then $\delta$ is $A$-pre-meet-metric.

Let $\phi : L \to L$ be defined as 
\[
    \phi(x) = \begin{cases}
        b_{i+1} &\text{if } x = b_i \\
        \one &\text{otherwise}.
    \end{cases}
\]
Then $\phi$ has no pyramidal points and $\phi$ post-join-depends on $\delta$.
\end{example}

Third, we give a counterexample to: ``If $\delta$ is post-join-covering, then $|\Fix( \phi )| \le 1$.''

\begin{example} \label{example:counterexample_covering_delta}
We consider the lattice opposite to the one in Example \ref{example:counterexample_limit}. Let $L = A \cup \{ b \} \cup \{ \zero \}$, where $A = (a_i : i \in \N)$ is an infinite descending chain with $a_0 = \one$ and $\bigmeet A = \zero$, and $b \join a_i = \one$, $b \meet a_i = \zero$ for all $a_i \in A$.

Let $\delta$ be defined as $\delta( a_i ) = a_{i-1}$ for all $i \in \N$, and $\delta( b ) = \delta( \zero ) = \one$. Let $\phi$ with
\[
    \phi(x) = \begin{cases}
        b & \text{if } x = b\\
        \one & \text{otherwise}.
    \end{cases}
\]
Then $\Fix(\phi) = \{ \one, b \}$ and $\phi$ post-join-depends on $\delta$.
\end{example}

\subsection{Feedback bounds} \label{Section:feedback_bounds}

While Robert's theorem was based on -complete mappings, the feedback bound is based on -metric mappings instead. Since there are four kinds of dependency, we obtain four kinds of feedback bounds. The four kinds of feedback vertex set (FVS) we consider here are defined as follows. Let $\Lattice{L}$ be a complete lattice and $\alpha, \beta, \gamma, \delta : L \to L$.
\begin{dependency}
    \item \label{item:pre-meet-FVS}
    Say $a \in L$ is a \Define{pre-meet-FVS} of $\alpha$ if there exists a meet-metric sequence $A = (a_i : i \in \N)$ with $a_0 = a$ and $\alpha(a_i) \le a_{i-1}$ for all $i \ge 1$.

    \item \label{item:pre-join-FVS}
    Say $b \in L$ is a \Define{pre-join-FVS} of $\beta$ if there exists a join-metric sequence $B = ( b_i : i \in \N )$ with $b_0 = b$ and $\beta( b_i ) \ge b_{i-1}$ for all $i \ge 1$.

    \item \label{item:post-meet-FVS}
    Say $c \in L$ is a \Define{post-meet-FVS} of $\gamma$ if there exists a meet-metric sequence $C = (c_i : i \in \N)$ with $c_0 = c$ and $\gamma(c_{i-1}) \ge c_i$ for all $i \ge 1$. 

    \item \label{item:post-join-FVS}
    Say $d \in L$ is a \Define{post-join-FVS} of $\delta$ if there exists a join-metric sequence $D = ( d_i : i \in \N )$ with $d_0 = d$ and $\delta( d_{i-1} ) \le d_i$ for all $i \ge 1$.
\end{dependency}



\begin{theorem}[Feedback bound for complete lattices] \label{theorem:feedback_bound}
Let $\phi, \alpha, \beta, \gamma, \delta : L \to L$. Then the following hold.
\begin{dependency}
    \item \label{item:feedback_bound_alpha}
    Let $\phi$ pre-meet-depend on $\alpha$ with pre-meet-FVS $a$. Then $|\Fix(\phi)| \le |a^\downarrow|$. 
    
    \item \label{item:feedback_bound_beta}
    Let $\phi$ pre-join-depend on $\beta$ with pre-join-FVS $b$. Then $|\Fix(\phi)| \le |b^\uparrow|$. 
    
    \item \label{item:feedback_bound_gamma}
    Let $\phi$ post-meet-depend on $\gamma$ with post-meet-FVS $c$. Then $|\Fix(\phi)| \le |c^\downarrow|$. 
    
    \item \label{item:feedback_bound_delta}
    Let $\phi$ post-join-depend on $\delta$ with post-join-FVS $d$. Then $|\Fix(\phi)| \le |d^\uparrow|$. 
\end{dependency}

\end{theorem}
 
\begin{proof}
We only give the proof of \ref{item:feedback_bound_alpha}; the other three proofs are similar and hence omitted.

We prove that if $s,t \in \Fix(\phi)$ satisfy $s \meet a = t \meet a$, then $s = t$. Similarly to the proof of Theorem \ref{theorem:generalised_robert}, by induction we have $s \meet a_i = t \meet a_i$ for all $i \in \N$, and hence $A_s = A_t$. Since $A$ is meet-metric, we obtain $s = t$.
\end{proof}

\section{Complete Boolean algebras} \label{section:CBA}

\subsection{Transpose and symmetric graphs} \label{section:transpose}

In this section, $\Lattice{L}$ is a complete Boolean algebra. For any $\psi: L \to L$, let the \Define{dual} of $\psi$ be the mapping $\Dual{ \psi }: L \to L$ defined by
\[
    \Dual{\psi}( x ) = \neg \psi( \neg x ).
\]
We define the \Define{transpose} of $\psi$ as
\[
    \Transpose{\psi} := \Dual{ \Residual{\psi} }.
\]
The \Define{co-transpose} of $\psi$ is naturally defined as $\CoTranspose{ \psi } = \Dual{ \Residuated{ \psi } }$.

Let us make some quick remarks about the dual and the transpose:
\begin{itemize}
    \item $f$ is a graph if and only if $\Dual{f}$ is a co-graph;

    \item $\Transpose{f}$ is a graph if and only if $f$ is a graph and similarly $\CoTranspose{ g }$ is a co-graph if and only if $g$ is a co-graph;
    
    \item for any $\psi: L \to L$, $\Transpose{ \psi } = \Dual{ \Residual{\psi} } = \Residuated{ \Dual{\psi} }$ and $\CoTranspose{ \psi } = \Dual{ \Residuated{\psi} } = \Residual{ \Dual{\psi} }$.
\end{itemize}
We obtain the commutative diagram in Figure \ref{figure:commutative_diagram}, where $f$ is a graph, and $g = \Residual{ f }$.

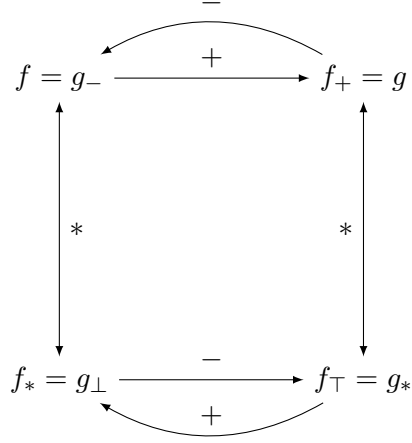
\begin{figure}
    \centering
    \begin{tikzpicture}[scale=4]
        \node (f)       at (0,1) {$f = \Residuated{ g }$};
        \node (fplus)   at (1,1) {$\Residual{ f } = g$};
        \node (fstar)   at (0,0) {$\Dual{ f } = \CoTranspose{ g }$};
        \node (ftop)    at (1,0) {$\Transpose{ f } = \Dual{ g }$};
    
        \draw[-latex] (f) to node[above] {$+$} (fplus);
        \draw[-latex] (fplus) [bend right=30] to node[above] {$-$} (f);
    
        \draw[-latex] (fstar) to node[above] {$-$} (ftop);
        \draw[-latex] (ftop) [bend left=30] to node[above] {$+$} (fstar);
    
        \draw[latex-latex] (f) to node[right] {$*$} (fstar);
    
        \draw[latex-latex] (fplus) to node[left] {$*$} (ftop);    
    \end{tikzpicture}
    \caption{Relations between a graph $f$, its residual $\Residual{ f }$, its dual $\Dual{ f }$ and its transpose $\Transpose{ f }$.}
    \label{figure:commutative_diagram}
\end{figure}


\begin{observation} \label{observation:dual_of_digraph}
The transpose of a digraph is obtained by reversing the direction of all the arcs.
\end{observation}

The main properties of the transpose graph are gathered in the theorem below.

\begin{theorem} \label{theorem:transpose}
Let $\Lattice{L}$ be a complete Boolean algebra. The following hold.
\begin{enumerate}
    \item \label{item:transpose_involution}
    For any graph $f$, $\Transpose{ \Transpose{f} } = f$.
    
    \item \label{item:transpose_graph}
    The transpose mapping $f \mapsto \Transpose{f}$ is a graph and a co-graph on the lattice of graphs.
    
    \item \label{item:transpose_composition}
    For any graphs $f,g$ we have $\Transpose{( f g )} = \Transpose{g} \Transpose{f}$.
    
    \item \label{item:transpose_transitive_reflexive}
    For any graph $f$, $\Transitive{ \Transpose{f} } = \Transpose{ \Transitive{f} }$ and $\Reflexive{ \Transpose{f} }  = \Transpose{ \Reflexive{f} }$.
\end{enumerate}
\end{theorem}

\begin{proof}
\begin{enumerate}
   \item \label{item:transpose_involution_proof}
    We have
    \[
        \Transpose{ \Transpose{f} } = \Residuated{ \Dual{ \Dual{ \Residual{f} } } } = \Residuated{ \Residual{ f } } = f.
    \]
    
    \item \label{item:transpose_graph_proof}
    If $F \subseteq \Graphs$ is a family of graphs, then $\bigjoin F$ is a graph satisfying
    \[
        \Residual{ \left( \bigjoin F \right) } = \bigmeet \Residual{ F }.
    \]
    Indeed, let $h = \bigjoin F$ and $i = \bigmeet \Residual{F}$. We have $h \ge f$ for any $f \in F$, hence $\Residual{h} \le \bigmeet \Residual{F} = i$; on the other hand, $i \le \Residual{f}$ for all $f \in F$, hence $\Residuated{i} \ge \bigjoin \Residuated{ \Residual{ F } } = \bigjoin F = h$, thus $i \le \Residual{h}$. 
    
    For any collection $\Psi$ of $L \to L$ mappings, we have $\Dual{(\bigmeet \Psi)} = \bigjoin \Dual{\Psi}$, thanks to the infinite De Morgan's laws. Therefore, 
    \[
        \Transpose{ \left( \bigjoin F \right) } = \Dual{ \Residual{ \left( \bigjoin F \right) } } = \Dual{ \left( \bigmeet \Residual{ F } \right) } = \bigjoin \Dual{ \Residual{ F } }  = \bigjoin \Transpose{F}.
    \]
    Thus $f \mapsto \Transpose{f}$ is a graph on the lattice of graphs $\Graphs$. 
    
    We have proved that the transposition mapping is a residuated involution. According to \cite[Exercise 1.14]{Bly05}, the transpose mapping is its own residual, and in particular, it is a co-graph.
    
    \item \label{item:transpose_composition_proof}
    Firstly, $\Residual{ ( fg ) } = \Residual{g} \Residual{f}$, since for all $x,y \in L$ we have
    \[
        f g( x ) \le y \iff g(x) \le \Residual{f}(y) \iff x \le \Residual{g} \Residual{f} (y) .
    \]
    Secondly, $\Dual{ ( fg ) } = \Dual{f} \Dual{g}$, since for all $x \in L$ we have
    \[
        \Dual{( fg )}(x) = \neg f( g( \neg x ) ) = \neg f ( \neg \neg g( \neg x) ) = \neg f( \neg \Dual{g}(x) ) = \Dual{f} \Dual{g} (x).
    \]
    Thus,
    \[
        \Transpose{ ( fg ) } = \Dual{ \Residual{ ( fg ) } } = \Dual{ ( \Residual{f} \Residual{g} ) } = \Dual{ \Residual{g} } \Dual{ \Residual{f} } = \Transpose{g} \Transpose{f}.
    \]
    
    \item \label{item:transpose_transitive_reflexive_proof}
    Firstly, we have
    \[
        \Transpose{ \Transitive{f} } = \Transpose{ \left( \bigjoin_{ i \in \N_1 } f^i \right) } = \bigjoin_{ i \in \N_1 } \Transpose{f^i} = \bigjoin_{ i \in \N_1 } {\Transpose{f}}^i = \Transitive{ \Transpose{f} }.
    \]
    Secondly, $\Residual{ \id  } = \id$ and $\Dual{\id} = \id$, therefore $\Transpose{ \id } = \id$. Thus 
    \[
        \Transpose{ \Reflexive{ f } } = \Transpose{ ( f \join \id ) } = \Transpose{ f } \join \id = \Reflexive{ \Transpose{ f } }.
    \]
\end{enumerate}
\end{proof}

\begin{corollary} \label{corollary:transitive_transpose}
A graph is transitive if and only if its transpose is transitive. Similarly, a graph is reflexive if and only if its transpose is reflexive.
\end{corollary}

From Lemma \ref{lemma:nilpotent}, we have that the graph $f$ is $K$-meet-nilpotent if and only if its transpose $\Transpose{ f }$ is $K$-meet-nilpotent. However, a similar equivalence does not hold for asymptotic meet-nilpotence, as the stairway to heaven is asymptotically meet-nilpotent, while its transpose the stairway to hell isn't.

Say $f$ is \Define{symmetric} if $f = \Transpose{f}$. The four examples of graphs given in Section \ref{section:graphs} are all symmetric; the proof is simple and hence omitted.

\begin{observation} \label{observation:symmetric_digraph}
A digraph $D$ is symmetric if $(u,v) \in E$ implies $(v,u) \in E$.
\end{observation}

If the out-neighbourhood $N = \Transpose{D}(v)$ is not empty, then $v$ belongs to the in-neighbourhood of $N$: $v \in D(N) = D \Transpose{D}(v)$. In more succinct terms, if $D( V ) = V$ then $D \Transpose{D}$ is reflexive. In general, the in-neighbourhood of $v$ satisfies $D(v) \subseteq D \Transpose{D} D(v)$ for any digraph $D$ and any vertex $v$. Those two facts are generalised to graphs on complete Boolean algebras as follows. For any graph $f$, let $\Symmetric{ f } = f \Transpose{ f }$; it is clear that $\Symmetric{ f }$ is a symmetric graph.

\begin{proposition} \label{proposition:symmetric_graphs}
The following hold for any graph $f$:
\begin{enumerate}
    \item
    $\Symmetric{ f }$ is reflexive if and only if $f( \one ) = \one$;
    
    \item 
    $\Symmetric{ f } f \ge f$.

\end{enumerate}
\end{proposition}

\begin{proof}
For all $y \in L$, we have
\[
    f(y) \join f(\neg y) \join \neg f( \one ) = f( y \join \neg y ) \join \neg f( \one ) = f( \one ) \join \neg f( \one ) = \one,
\]
hence
\[
    f(y) \ge \neg \left( f( \neg y ) \join \neg f( \one ) \right) = \Dual{f}( y ) \meet f( \one ).
\]
Setting $y = \Transpose{f}(x)$, we obtain
\begin{equation} \label{equation:symmetric}
    f( \Transpose{f}( x ) ) \ge \Dual{f}( \Transpose{f}( x ) ) \meet f( \one ) = \Residual{ \Transpose{f} }( \Transpose{f}( x ) ) \meet f( \one ) \ge x \meet f( \one ).
\end{equation}
Thus, if $f( \one ) = \one$, $\Symmetric{ f }(x) \ge x$ for all $x \in L$. Conversely, if $f( \one ) < \one$, then $\Symmetric{ f }( \one ) = f \Transpose{ f }( \one ) \le f( \one ) < \one$.

Moreover, setting $x = f(u)$ in Equation \eqref{equation:symmetric} yields $f \Transpose{f} f( u ) \ge f(u) \meet f( \one ) = f(u)$.
\end{proof}

If a digraph is both acyclic and symmetric, then it is empty. We generalise this fact to fixed point-free graphs on complete Boolean algebras.

\begin{corollary} \label{corollary:symmetric_strongly_acyclic}
Let $\Lattice{L}$ be a complete Boolean algebra. Then the unique fixed point-free symmetric graph on $\Lattice{L}$ is the empty graph.
\end{corollary}

\begin{proof}
If $f = \Transpose{f}$, then let $u = f( \one )$. We have $u \ge f( u ) \ge f( f(u) )$, and by Equation \eqref{equation:symmetric}, $f( f( u ) ) \ge u$, thus $u = f(u) = f(f(u))$. 
\end{proof}

\subsection{Dependencies and Boolean algebras} \label{section:dependencies_boolean}

We now show that all four dependencies are equivalent for mappings over a Boolean algebra that depend on graphs.

\begin{theorem} \label{theorem:CBA_dependencies}
Let $\Lattice{L}$ be a complete Boolean algebra and $\phi : L \to L$. Let $f$ be a graph on $\Lattice{L}$, then the following are equivalent:
\begin{equivalent}
    \item \label{item:Boolean_graph_alpha}
    $\phi$ pre-meet-depends on $f$;
    
    \item \label{item:Boolean_graph_beta}
    $\phi$ pre-join-depends on $\Dual{f}$;

    \item \label{item:Boolean_graph_gamma}
    $\phi$ post-meet-depends on $\Residual{f}$;
    
    \item \label{item:Boolean_graph_delta}
    $\phi$ post-join-depends on $\Transpose{f}$.
\end{equivalent}
\end{theorem}

The rest of this subsection is devoted to the proof of Theorem \ref{theorem:CBA_dependencies}.

\begin{lemma} \label{lemma:equality_CBA}
Let $\Lattice{L}$ be a complete Boolean algebra, and $x,y, a \in L$. Then $x \meet a = y \meet a$ if and only if $\neg x \meet a = \neg y \meet a$; similarly, $x \join a = y \join a$ if and only if $\neg x \join a = \neg y \join a$.
\end{lemma}

\begin{proof}
Suppose $x \meet a = y \meet a$. Let $z = x \meet a \meet \neg y$. We have
\[
    y \meet a = x \meet a = x \meet a \meet (y \join \neg y) = (x \meet a \meet y) \join (x \meet a \meet \neg y) = (y \meet a) \join z.
\]
Hence $z \le y$, and $z \le \neg y$, thus $z \le y \meet \neg y = \zero$. Thus $x \meet a \meet \neg y = \zero$ and by symmetry $\neg x \meet a \meet y = \zero$. We obtain
\[
    \neg x \meet a = \neg x \meet a \meet (y \join \neg y) = (\neg x \meet a \meet y) \join (\neg x \meet a \meet \neg y) = \neg x \meet a \meet \neg y,
\]
and by symmetry $\neg y \meet a = x \meet a \meet \neg y$. Thus $\neg x \meet a = \neg y \meet a$.

The proof of the second statement is similar and hence omitted.
\end{proof}

\begin{lemma} \label{lemma:dependencies_dual}
Let $\Lattice{L}$ be a complete Boolean algebra and $\phi, \alpha, \beta, \gamma, \delta : L \to L$. Then the following hold.
\begin{dependency}
    \item \label{item:pre-meet-dependency_dual}
    $\phi$ pre-meet-depends on $\alpha$ if and only if $\Dual{ \phi }$ pre-meet-depends on $\alpha$.

    \item \label{item:pre-join-dependency_dual}
    $\phi$ pre-join-depends on $\beta$ if and only if $\Dual{ \phi }$ pre-join-depends on $\beta$.

    \item \label{item:post-meet-dependency_dual}
    $\phi$ post-meet-depends on $\gamma$ if and only if $\Dual{ \phi }$ post-meet-depends on $\gamma$.

    \item \label{item:post-join-dependency_dual}
    $\phi$ post-join-depends on $\delta$ if and only if $\Dual{ \phi }$ post-join-depends on $\delta$.
\end{dependency}
\end{lemma}

\begin{proof}
We only prove \ref{item:pre-meet-dependency_dual}; the other proofs are similar and hence omitted. We only have to prove one implication; the other one follows from duality. If $\phi$ pre-meet-depends on $\alpha$, then for all $x, y, s \in L$,
\begin{align*}
    x \meet \alpha( s ) = y \meet \alpha( s ) &\implies \neg x \meet \alpha( s ) = \neg y \meet \alpha( s ) \\
    &\implies \phi( \neg x ) \meet s = \phi( \neg y ) \meet s \\
    &\implies \Dual{ \phi }( x ) \meet s = \Dual{ \phi }( y ) \meet s.
\end{align*}
Thus $\Dual{ \phi }$ pre-meet-depends on $\alpha$.
\end{proof}

\begin{lemma} \label{lemma:dependencies_CBA}
Let $\Lattice{L}$ be a complete Boolean algebra and $\phi, \alpha, \beta, \gamma, \delta : L \to L$. Then the following hold.
\begin{dependency}
    \item \label{item:pre-meet-dependency_CBA}
    If $\phi$ pre-meet-depends on $\alpha$, then $\phi$ pre-join-depends on $\Dual{ \alpha }$.

    \item \label{item:pre-join-dependency_CBA}
    If $\phi$ pre-join-depends on $\beta$, then $\phi$ pre-meet-depends on $\Dual{ \beta }$.

    \item \label{item:post-meet-dependency_CBA}
    If $\phi$ post-meet-depends on $\gamma$, then $\phi$ post-join-depends on $\Dual{ \gamma }$.

    \item \label{item:post-join-dependency_CBA}
    If $\phi$ post-join-depends on $\delta$, then $\phi$ post-meet-depends on $\Dual{ \delta }$.
\end{dependency}
\end{lemma}

\begin{proof}
Once again, we only prove \ref{item:pre-meet-dependency_CBA}; the other proofs are similar and hence omitted. If $\phi$ pre-meet-depends on $\alpha$, then for all $x,y, s \in L$ we have
\begin{align*}
    x \join \Dual{ \alpha }( s ) = y \join \Dual{ \alpha }( s ) 
    &\implies \neg x \meet \neg \Dual{ \alpha }( s ) = \neg y \meet \neg \Dual{ \alpha }( s ) \\
    &\implies \neg x \meet \alpha( \neg s ) = \neg y \meet \alpha( \neg s ) \\
    &\implies \phi( \neg x ) \meet \neg s = \phi( \neg y ) \meet \neg s \\
    &\implies \neg \phi( \neg x ) \join s = \neg \phi( \neg y ) \join s \\
    &\implies \Dual{ \phi }( x ) \join s = \Dual{ \phi }( y ) \join s.
\end{align*}
Thus $\Dual{ \phi }$ pre-join-depends on $\Dual{ \alpha }$. By Lemma \ref{lemma:dependencies_dual}, $\phi$ pre-join-depends on $\Dual{ \alpha }$ as well.
\end{proof}

Lemma \ref{lemma:dependencies_CBA} can be summarised as follows: $\phi$ (pre/post)-meet-depends on $\psi$ if and only if $\phi$ (pre/post)-join-depends on $\Dual{ \psi }$.

\begin{lemma} \label{lemma:dependencies_transpose}
Let $\Lattice{L}$ be a complete Boolean algebra and $\phi, \alpha, \beta, \gamma, \delta : L \to L$. Then the following hold.
\begin{dependency}
    \item \label{item:pre-meet-dependency_transpose}
    If $\phi$ post-join-depends on $\Transpose{ \alpha }$, then $\phi$ pre-meet-depends on $\alpha$.

    \item \label{item:pre-join-dependency_transpose}
    If $\phi$ post-meet-depends on $\CoTranspose{ \beta }$, then $\phi$ pre-join-depends on $\beta$.

    \item \label{item:post-meet-dependency_transpose}
    If $\phi$ pre-join-depends on $\CoTranspose{ \gamma }$, then $\phi$ post-meet-depends on $\gamma$.

    \item \label{item:post-join-dependency_transpose}
    If $\phi$ pre-meet-depends on $\Transpose{ \delta }$, then $\phi$ post-join-depends on $\delta$.
\end{dependency}
\end{lemma}

\begin{proof}
Again, we only prove \ref{item:pre-meet-dependency_transpose}; the other proofs are similar and hence omitted. Let $\beta = \Dual{ \alpha }$ so that $\Transpose{ \alpha } = \Residuated{ \beta }$. If $\phi$ post-join-depends on $\Residuated{ \beta }$, then by Lemma \ref{lemma:pre-post-dependence}, $\phi$ pre-join-depends on $\beta$. Then, by Lemma \ref{lemma:dependencies_CBA}, $\phi$ pre-meet-depends on $\Dual{ \beta } = \alpha$.
\end{proof}

\begin{proof}[Proof of Theorem \ref{theorem:CBA_dependencies}]
The proof follows from Lemma \ref{lemma:dependencies_transpose}, where $\alpha = f$, $\beta = \Dual{ f }$, $\gamma = \Residual{ f }$, and $\delta = \Transpose{ f }$.
\end{proof}

Since $f$ pre-join-depends on $\Residual{f}$, Lemma \ref{lemma:dependencies_CBA} shows that $f$ pre-meet-depends on $\Transpose{f}$. In fact, the dependency relations amongst $f$, $\Residual{f}$, $\Dual{f}$, and $\Transpose{f}$ are represented in Figure \ref{figure:dependencies2}, where the head of each arc depends on the tail of the arc. For instance, the fact that $\Transpose{f}$ post-meet-depends on $\Residual{f}$ is the red arc on the right.

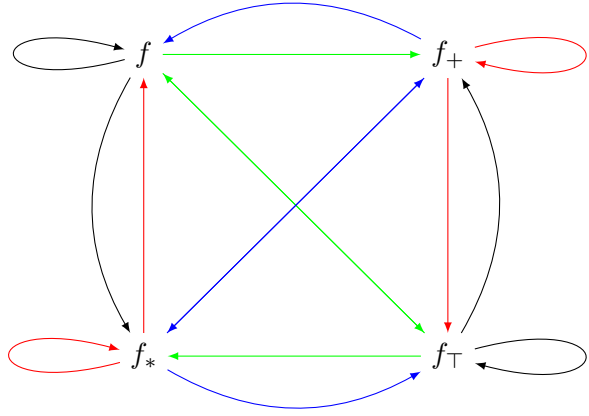
\begin{figure}
    \centering
    
    \begin{tikzpicture}[scale=4]
        \node (f) at (0,1) {$f$};
        \node (fm) at (1,1) {$\Residual{f}$};
        \node (ft) at (1,0) {$\Transpose{f}$};
        \node (fs) at (0,0) {$\Dual{f}$};
        
        \draw[color=green,-latex] (f) -- (fm);
        \draw[color=green,-latex] (f) -- (ft);
        \draw[color=green,-latex] (ft) -- (f);
        \draw[color=green,-latex] (ft) -- (fs);
        
        \draw[color=red,-latex] (fm) to [loop right] (fm);
        \draw[color=red,-latex] (fm) -- (ft);
        \draw[color=red,-latex] (fs) -- (f);
        \draw[color=red,-latex] (fs) to [loop left] (fs);
    
        \draw[color=black,-latex] (ft) to [bend right] (fm);
        \draw[color=black,-latex] (ft) to [loop right] (ft);
        \draw[color=black,-latex] (f) to [loop left] (f);
        \draw[color=black,-latex] (f) to [bend right] (fs);
    
        \draw[color=blue,-latex] (fs) -- (fm);
        \draw[color=blue,-latex] (fs) to [bend right] (ft);
        \draw[color=blue,-latex] (fm) to [bend right] (f);
        \draw[color=blue,-latex] (fm) -- (fs);
    \end{tikzpicture}
    
    \caption{Dependencies amongst $f$, $\Residual{f}$, $\Dual{f}$, and $\Transpose{f}$. Each colour indicates a different kind of dependency: \textbf{\textcolor{green}{green for pre-meet}}, \textbf{\textcolor{blue}{blue for pre-join}}, \textbf{\textcolor{red}{red for post-meet}}, and \textbf{\textcolor{black}{black for post-join}}.}
    \label{figure:dependencies2}
\end{figure}

\subsection{Robert's theorem for dependency on graphs over complete Boolean algebras} \label{subsection:robert_CBA}

The statement of Robert's theorem for pre-meet-dependency can be simplified further, depending on whether $\Lattice{L}$ is a frame  -- in which case $\alpha$ being pre-meet-covering is sufficient for convergence, or $\alpha$ being a graph -- in which case one can choose $A = \MeetSequence( \Residual{\alpha} )$. We thus derive Robert's theorem for dependency on graphs over complete Boolean algebras in this subsection. Similarly to Section \ref{section:robert}, we shall give Robert's theorem for complete Boolean algebras, study its tightness, and then refine the feedback bound. The overall picture for Robert's theorem is depicted in Figure \ref{figure:robert_CBA}.

We begin this subsection by showing that a graph being pre-meet-complete is actually equivalent to its transpose being asymptotically meet-nilpotent.

\begin{proposition} \label{proposition:CBA_complete_graphs}
Let $\Lattice{L}$ be a complete Boolean algebra and $f$ be a graph, then the following are equivalent:
\begin{equivalent}
    \item \label{item:CBA_complete_graphs_pre-meet-complete}
    $f$ is pre-meet-complete;

    \item \label{item:CBA_complete_graphs_pre-join-complete}
    $\Dual{f}$ is pre-join-complete;
        
    \item \label{item:CBA_complete_graphs_post-meet-complete}
    $\Residual{f}$ is post-meet-complete, and equivalently asymptotically join-nilpotent;

    \item \label{item:CBA_complete_graphs_post-join-complete}
    $\Transpose{f}$ is post-join-complete, and equivalently asymptotically meet-nilpotent.
\end{equivalent}
\end{proposition} 

\begin{proof}
Since a complete Boolean algebra is both a frame and a locale, by Proposition \ref{proposition:discriminating}, a sequence is meet/join-complete if and only if it is meet/join-covering. As such, Proposition \ref{proposition:metric_nilpotent} holds if we replace -covering by -complete everywhere. In particular, by item \ref{item:pre-meet-covering}, $f$ is pre-meet-complete if and only if $\Residual{f}$ is asymptotically join-nilpotent, and by item \ref{item:post-meet-covering}, the latter is equivalent to $\Residual{f}$ being post-meet-complete. By duality, $f$ is pre-meet-complete if and only if $\Dual{f}$ is pre-join-complete, and the rest of the proof is similar as above.
\end{proof}

We can now give Robert's theorem for dependency on graphs over complete Boolean algebras. As seen from Theorem \ref{theorem:CBA_dependencies}, all four kinds of dependencies are equivalent in this case, hence we only have one statement.

\begin{theorem}[Robert's theorem for dependency on graphs over complete Boolean algebras] \label{theorem:robert_CBA}
Let $\Lattice{L}$ be a complete Boolean algebra, $f$ be a graph on $\Lattice{L}$ with $\Transpose{f}$ asymptotically meet-nilpotent, and $\phi: L \to L$ pre-meet-depend on $f$. The following hold.
\begin{enumerate}
    \item \label{item:robert_CBA_e}
    There exists $e \in L$ such that $e = \bigmeet_i \phi^i( x ) \join \Transpose{f}^i( \one )$ for all $x \in L$ and $\phi \StronglyConverges e$.


    \item \label{item:robert_CBA_nilpotent}
    If $f$ is $K$-meet-nilpotent, then $e = \phi^K( x )$ for all $x \in L$.
\end{enumerate}
\end{theorem}

\begin{proof}
\begin{enumerate}
    \item 
    By Theorem \ref{theorem:CBA_dependencies}, $\phi$ post-join depends on $\Transpose{f}$, which is post-join-complete by Proposition \ref{proposition:CBA_complete_graphs}. The result then follows from  Theorem \ref{theorem:robert}\ref{item:robert_post-join}\ref{item:robert_post-join_metric}.

    \item 
    In this case, $\Transpose{f}$ is also $K$-meet-nilpotent, and Theorem \ref{theorem:robert}\ref{item:robert_post-join}\ref{item:robert_post-join_nilpotent} concludes.
\end{enumerate}

\end{proof}

\begin{figure} 
\resizebox{\textwidth}{!}{
\begin{tikzpicture}[xscale=5, yscale=4]
    
    \node (an) at (0,0) [draw,thick,minimum width=2cm,minimum height=2cm, text width = 3.2cm, align=center] {$\Transpose{f}$\\ $K$-meet-nilpotent};
    \node (am) at (1,0) [draw,thick,minimum width=2cm,minimum height=2cm, text width = 3.2cm, align=center] {$\Transpose{f}$\\ asymptotically meet-nilpotent};
    \node (al) at (2,0) [draw,thick,minimum width=2cm,minimum height=2cm, text width = 3.2cm, align=center] {$\Pyr( \Transpose{f} ) = \{ \zero \}$};
    \node (ai) at (3,0) [draw,thick,minimum width=2cm,minimum height=2cm, text width = 3.2cm, align=center] {$\Transpose{f}$\\ fixed point-free};

    \draw[-latex, color=black] (an) -- (am);
    \draw[-latex, color=black] (am) -- (al);
    \draw[-latex, color=green] (al) to [bend left=15] (am);
    \draw[-latex, color=black] (al) -- (ai);

    \node (pn) at (0,1) [draw,thick,minimum width=3cm,minimum height=2cm, text width = 3.2cm, align=center] {$\phi^K = e$};
    \node (pm) at (1,1) [draw,thick,minimum width=3cm,minimum height=2cm, text width = 3.2cm, align=center] {$\phi \StronglyConverges e$};
    \node (pi) at (2,1) [draw,thick,minimum width=3cm,minimum height=2cm, text width = 3.2cm, align=center] {$\Pyr(\phi) \subseteq \{ e \}$};
    \node (pf) at (3,1) [draw,thick,minimum width=3cm,minimum height=2cm, text width = 3.2cm, align=center] {$| \Par( \phi ) | \le 1$};

    \draw[-latex, color=black] (pn) -- (pm);
    \draw[-latex, color=black] (pm) -- (pi);
    \draw[-latex, color=black] (pi) -- (pf);

    \draw[-latex, color=black] (an) to node[left] {\tiny Th.\ref{theorem:robert_CBA}\ref{item:robert_CBA_nilpotent} } (pn);
    \draw[-latex, color=black] (am) to node[left] {\tiny Th.\ref{theorem:robert_CBA}\ref{item:robert_CBA_e} } (pm);
    \draw[-latex, color=black] (ai) to node[left] {\tiny Th.\ref{theorem:precarious} } (pf);
    
    \draw[-latex, color=red, dotted] (am) to node[left] {\tiny Ex.\ref{example:counterexample_nilpotent_CBA} } (pn);
    \draw[-latex, color=red, dotted] (ai) to node[left] {\tiny Ex.\ref{example:counterexample_limit_CBA} } (pi);

\end{tikzpicture}
}
\caption{Robert's theorem for dependency on graphs over complete Boolean algebras. 
The \textbf{\textcolor{red}{red dotted arrows}} represent counterexamples, the \textbf{\textcolor{green}{green arrow}} represents an implication that holds when $\Lattice{L}$ is a power set lattice, while the \textbf{\textcolor{black}{black arrows}} represent unconditional implications.}
\label{figure:robert_CBA} 
\end{figure}
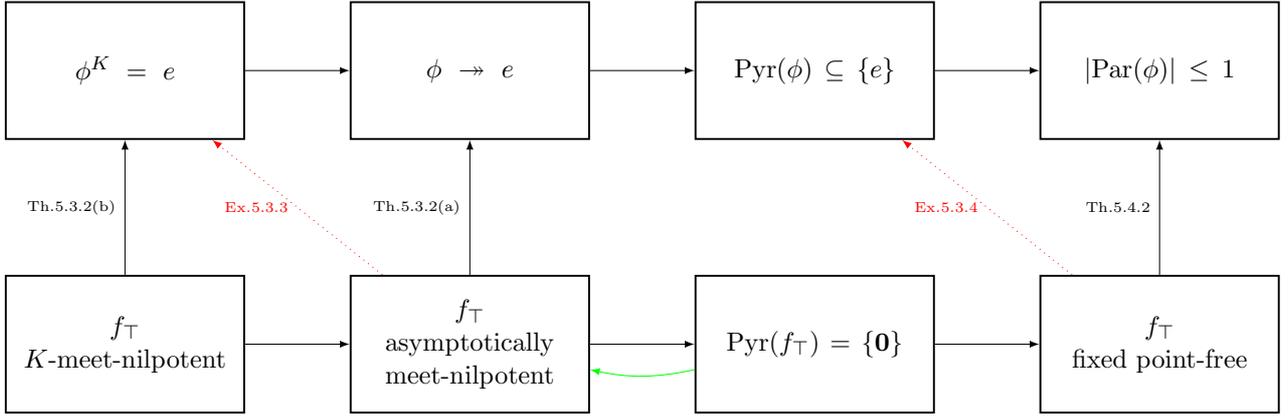

We now comment on the tightness of Robert's theorem for complete Boolean algebras. We give two counterexamples.

Firstly, we give a counterexample to: ``If $\Transpose{f}$ is asymptotically meet-nilpotent, then $\phi^K = e$.'' 

\begin{example} \label{example:counterexample_nilpotent_CBA}
Let $f$ be the backwards ray and $\phi = \Transpose{f}$ be the ray. Then $\phi$ pre-meet-depends on $f$, and $\Transpose{f}$ is not nilpotent.
\end{example}

Secondly, we give a counterexample to: ``If $\Transpose{f}$ is fixed point-free, then $\Pyr( \phi ) \subseteq \{ e \}$.''

\begin{example} \label{example:counterexample_limit_CBA}
Let $\phi$ be the stairway to hell. Then $\phi$ pre-meet-depends on the stairway to heaven $f$, but $\Transpose{f} = \phi$ is not asymptotically meet-nilpotent: $| \Pyr( \phi ) | = 2$.
\end{example}



    





\subsection{Fixed-point free graphs and the feedback bound} \label{subsection:fixed_point-free_CBA}







In \cite{RRM13}, a digraph $D$ is called precarious if there exists a Boolean network $\phi$ with multiple fixed points (i.e. a hypodox) which post-join depends on $D$. They prove that a digraph is precarious if and only if it has an infinite walk. In our terminology, that means $D$ is not fixed point-free. We now generalise this characterisation to all complete Boolean algebras; we also show refine it by giving the corresponding result for parabolic points.

We begin with a lemma that shall be used repeatedly in this subsection.

\begin{lemma} \label{lemma:difference_CBA}
Let $\Lattice{L}$ be a complete Boolean algebra and $x, y \in L$. Let $A = \{ s \in L : x \join s = y \join s \}$ and $a = \bigmeet A$. Then $A = a^\uparrow$.
\end{lemma}

\begin{proof}
First, $a \in A$ since
\[
    x \join a = x \join \bigmeet A = \bigmeet x \join A = \bigmeet y \join A = y \join a.
\]
Second, if $b \ge a$, then 
\[
    x \join b = x \join a \join b =  y \join a \join b = y \join b.
\]
\end{proof}

\begin{theorem} \label{theorem:precarious}
Let $\Lattice{L}$ be a complete Boolean algebra and $f$ be a graph on $\Lattice{L}$. Then the following are equivalent:
\begin{equivalent}
    \item \label{item:precarious_fixed_point-free}
    $\Transpose{f}$ is fixed point-free;

    \item \label{item:precarious_fixed_points}
    $| \Fix( \phi ) | \le 1$ whenever $\phi : L \to L$ pre-meet-depends on $f$;

    \item \label{item:precarious_parabolic_points}
    $| \Par( \phi ) | \le 1$ whenever $\phi : L \to L$ pre-meet-depends on $f$.
\end{equivalent}
\end{theorem}

\begin{proof}
$\ref{item:precarious_fixed_point-free} \implies \ref{item:precarious_fixed_points}$. Suppose $x, y \in \Fix( \phi )$. Let $A = \{ s \in L : x \join s = y \join s \}$ and $a = \bigmeet A$. By Lemma \ref{lemma:difference_CBA}, $a \in A$, hence we have the following chain of implications:
\begin{alignat*}{3}
    x \join a = y \join a &\implies \phi( x ) \join \Transpose{f}( a ) = \phi( y ) \join \Transpose{f}( a ) &\quad& \text{because $\phi$ post-join depends on $\Transpose{f}$} \\
    &\implies \Transpose{f}( a ) \ge a &\quad& \text{because $a = \bigmeet A$} \\
    &\implies a = \zero &\quad& \text{because $\Transpose{f}$ is fixed-point free} \\
    &\implies x = y &\quad& \text{because $x \join a = y \join a$}.
\end{alignat*}

$\ref{item:precarious_fixed_point-free} \implies \ref{item:precarious_parabolic_points}$. In view of the previous implication, we only need to show that any parabolic point of $\phi$ is a fixed point. Suppose $x \in \Par( \phi )$, with $(x_i : i \in \N)$ such that $\phi( x_i ) = x_{i-1}$ for all $i \in \N_1$. Extend this sequence to all $\Z$ with $x_{ -i } = \phi^i( x )$ for all $i \in \N$. For any $i \in \Z$, let $D_i = \{ s \in L: x_i \join s = x_{i-1} \join s \}$ and $d_i = \bigmeet D_i$. By Lemma \ref{lemma:difference_CBA}, $d_i \in D_i$, hence we have the following chain of implications:
\begin{alignat*}{3}
    x_i \join d_i = x_{i-1} \join d_i &\implies x_{i-1} \join \Transpose{f}( d_i ) = \phi( x_{i-2} ) \join \Transpose{f}( d_i ) &\quad& \text{because $\phi$ post-join depends on $\Transpose{f}$} \\
    &\implies \Transpose{f}( d_i ) \ge d_{i-1} &\quad& \text{because $d_{ i-1 } = \bigmeet D_{i-1}$}.
\end{alignat*}
Let $D = (d_i : i \in \Z)$, then $\Transpose{f} \left( \bigjoin D \right) \ge \bigjoin D$, thus $\bigjoin D = \zero$. In particular, $d_0 = \zero$, which implies $x \in \Fix( \phi )$.

$\ref{item:precarious_parabolic_points} \implies \ref{item:precarious_fixed_points}$. Trivial.

$\ref{item:precarious_fixed_points} \implies \ref{item:precarious_fixed_point-free}$. If $\Transpose{f}$ is not fixed point-free, then $\phi = \Transpose{f}$ pre-meet-depends on $f$, and $| \Fix( \phi ) | \ge 2$.
\end{proof}

We now refine the feedback bound for complete Boolean algebras. Recall the definition of induced subgraph from item \ref{item:conjunction} in Section \ref{subsection:graphs}. Say $t$ is a \Define{fixed point-free set} of $f$ if $\Transpose{f}[t]$, the subgraph of $\Transpose{f}$ induced by $t$, is fixed point-free. 

\begin{lemma} \label{lemma:fpf_set}
Let $\Lattice{L}$ be a complete Boolean algebra and let $f$ be a graph on $\Lattice{L}$. The following hold.
\begin{enumerate}
    \item \label{item:fpf_set_definition}
    If $t$ is a fixed point-free set of $f$, then for any $u \le t$, $\Transpose{f}( u ) \ge u \implies u = \zero$. 

    \item \label{item:fpf_set_FVS}
    If $a$ is a pre-meet-FVS of $f$, then $\neg a$ is a fixed point-free set of $f$.
\end{enumerate}
\end{lemma}

\begin{proof}
\begin{enumerate}
    \item 
    We have the following set of equivalences:
    \begin{align*}
        \Transpose{f}[t] \text{ fixed point-free} &\iff \left[ \Transpose{f}(u) \meet t \ge u \implies u = \zero \ \forall u \le t \right] \\
        &\iff \left[ \Transpose{f}(u) \ge u \implies u = \zero \ \forall u \le t \right].
    \end{align*}

    \item 
    Let $b = \neg a$, then the sequence $B = \neg A$ satisfies $b_0 = b$,
    \[
        f( a_i ) \le a_{i-1} \iff \Dual{f}( b_i ) \ge b_{i-1} \iff \Residuated{ \Transpose{f} }( b_i ) \ge b_{i-1} \iff \Transpose{f}( b_{i-1} ) \le b_i,
    \]
    and $\bigmeet B = \zero$. Therefore,
    \[
        \bigmeet_i \Transpose{f}^i( b ) \le \bigmeet_i b_i = \zero.
    \]
    If $c \le b$ and $\Transpose{f}( c ) \ge c$, then
    \[
        \bigmeet_i \Transpose{f}^i( b ) \ge \bigmeet_i \Transpose{f}^i( c ) \ge c,
    \]
    which implies $c = \zero$, and hence $b$ is a fixed point-free set of $f$.
\end{enumerate}

\end{proof}

In view of Lemma \ref{lemma:fpf_set} \ref{item:fpf_set_FVS}, the feedback bound for complete Boolean algebras in Theorem \ref{theorem:feedback_bound_CBA} below is tighter than the one for complete lattices in Theorem \ref{theorem:feedback_bound}.

\begin{theorem}[Feedback bound for dependency on graphs over complete Boolean algebras] \label{theorem:feedback_bound_CBA}
Let $\Lattice{L}$ be a complete Boolean algebra, $f$ be a graph on $\Lattice{L}$ and $t$ be an acyclic set of $f$. If $\phi : L \to L$ pre-meet-depends on $f$, then $| \Fix( \phi ) | \le | t^\uparrow |$.
\end{theorem}

\begin{proof}
Suppose $x, y \in \Fix( \phi )$ such that $x \join t = y \join t$. Let $A = \{ s \in L : x \join s = y \join s \}$ and $a = \bigmeet A$. First, $a \le t$ since $t \in A$. By Lemma \ref{lemma:difference_CBA}, $a \in A$, hence we have the following chain of implications:
\begin{alignat*}{3}
    x \join a = y \join a &\implies x \join \Transpose{f}( a ) = y \join \Transpose{f}( a ) &\quad& \text{because $\phi$ post-join depends on $\Transpose{f}$} \\
    &\implies \Transpose{f}( a ) \ge a &\quad& \text{because $a = \bigmeet A$} \\
    &\implies a = \zero &\quad& \text{because $t$ is a fixed point-free set of $f$} \\
    &\implies x = y &\quad& \text{because $x \join a = y \join a$}.
\end{alignat*}
Therefore, if $x \join t = y \join t$, then $x = y$. Thus, the mapping $\Fix( \phi ) \to t^\uparrow$ that maps $x$ to $x \join t$ is injective and $| \Fix( \phi ) | \le | t^\uparrow |$.
\end{proof}

\section{Applications to automata networks} \label{subsection:applications_robert}

\subsection{Reviewing the original Robert theorem} \label{subsection:original_robert}

The \Define{direct product} of complete lattices is defined as follows \cite[Chapter I, Section 3]{Gra13}. Let $V$ be a set, and for any $v \in V$, let $\Lattice{L}_v = ( L_v, \join_v, \meet_v, \zero_v, \one_v )$ be a complete lattice. Then their direct product is $\Lattice{L} = \prod_{v \in V} \Lattice{L}_v$ with $\Lattice{L} = ( L, \join, \meet, \zero, \one )$ is defined as follows:
\begin{align*}
    L &= \prod_{v \in V} L_v = \{ x = (x_v \in L_v : v \in V) \}, \\
    x \join y &= ( x_v \join_v y_v : v \in V ), \\
    x \meet y &= ( x_v \meet_v y_v : v \in V ), \\
    \zero &= ( \zero_v : v \in V ), \\
    \one &= ( \one_v : v \in V ).
\end{align*}
We note that the notation is consistent: $\zero = (\zero_v : v \in V)$ and $\one = ( \one_v : v \in V)$. We shall introduce shorthand notations for any subset $I$ of $V$, e.g. $L_I = \prod_{i \in I} L_i$, $x_I = (x_i : i \in I)$, etc.

Let $\Lattice{L} = \prod_{v \in V} \Lattice{L}_v$. Then an \Define{automata network} is a mapping $\phi: L \to L$; a \Define{finite automata network} is an automata network where $V$ is finite; a \Define{Boolean network} is an automata network where $|L_v| = 2$ for all $v \in V$. 

Denoting $T = \prod_{v \in V} \{ \zero_v, \one_v \}$, we have that $\Lattice{T} = (T, \join, \meet, \zero, \one)$ is isomorphic to the power set lattice $\Lattice{\Powerset}(V)$. $T$ is a bi-topology, hence the mapping $h$ from the proof of the bi-topology theorem defined by $h(x) = \bigmeet T \cap x^\downarrow$ is a graph. More explicitly, we have
\[
    h(s)_v = \begin{cases}
        \zero_v &\text{if }s_v = \zero_v,\\
        \one_v &\text{otherwise}.
    \end{cases}
\]
Then an \Define{interaction graph} is any graph on $\Lattice{L}$ of the form $f = gh$ where $g$ is a graph on $\Lattice{T}$. Therefore, there is a one-to-one correspondence between interaction graphs on $\Lattice{L}$ and digraphs on $V$.

The following proposition justifies our terminology and shows that Theorem \ref{theorem:generalised_robert} is indeed a generalisation of Robert's original theorem.

\begin{proposition} \label{lemma:interaction_graph}
Let $\Lattice{L} = \prod_{v \in V} \Lattice{L}_v$ as above. If $f : L \to L$ is an interaction graph on $\Lattice{L}$, then there exists a digraph $D$ on $V$ such that $\phi : L \to L$ pre-meet-depends on $f$ if and only if $\phi$ depends on $D$.
\end{proposition}

\begin{proof}
Let $T = \prod_{v \in V} \{ \zero_v, \one_v \}$, then there is a bijection $\tau : \Powerset(V) \to T$ where $\tau(a) = ( \one_a, \zero_{\neg a} )$. We remark that for all $x \in L$, $x \meet \tau( a ) = ( x_a, \zero_{\neg a} )$. Let $f = gh$ be an interaction graph on $\Lattice{L}$ and let $D = \tau^{-1} g \tau$ be the corresponding digraph on $V$.

We then have the following chain of equivalences:
\begin{align}
    \nonumber
    \phi \text{ depends on } D 
    &\iff  \left[ x_{ D( a ) } = y_{ D( a ) } \implies \phi(x)_a = \phi(y)_a \quad \forall a \subseteq V \right] \\
    \nonumber
    &\iff \left[ x \meet g \tau( a ) = y \meet g \tau( a ) \implies \phi(x) \meet \tau( a ) = \phi(y) \meet \tau( a ) \quad \forall a \subseteq V \right] \\
    \nonumber
    &\iff \left[ x \meet g( t ) = y \meet g( t ) \implies \phi(x) \meet t = \phi(y) \meet t \quad \forall t \in T \right] \\
    \label{equation:condition1}
    &\iff \left[ x \meet g h(s) = y \meet g h(s) \implies \phi(x) \meet h(s) = \phi(y) \meet h(s) \quad \forall s \in L \right].
\end{align}
Since $h(s) \ge s$, the property in the RHS of \eqref{equation:condition1} clearly implies
\begin{equation}    \label{equation:condition2}
    x \meet gh(s) = y \meet gh(s) \implies \phi(x) \meet s = \phi(y) \meet s \quad \forall s \in L,
\end{equation}
which in turn is equivalent to $\phi$ pre-meet-depending on $f$. Conversely, since $h$ is idempotent, restricting \eqref{equation:condition2} to $h(s)$ implies the property in the RHS of \eqref{equation:condition1}.
\end{proof}

\subsection{Beyond the original Robert theorem} \label{subsection:beyond_original_Robert}


The so-called converse of the complete Banach contraction principle shows that if a mapping converges to a unique fixed point, then it is contractive for some metric \cite{Bes59}. The bounded Banach contraction principle has a similar converse \cite{Jac00}. Instead of searching for a converse of Robert's theorem in general, we show it for the finite case below. So far, we placed on $X = Q^V$ the structure of a direct product of lattices, but by using a different lattice structure (namely, a chain) we prove that any nilpotent mapping $\phi : X \to X$ pre-meet-depends on a meet-nilpotent graph.

\begin{theorem} \label{theorem:all_nilpotent_alpha}
Let $X$ be a finite set and $\phi : X \to X$ be $K$-nilpotent. Then there exists a lattice $\Lattice{L}$ on $X$ and a $K$-meet-nilpotent graph $f : X \to X$ on $\Lattice{L}$ such that $\phi$ pre-meet-depends on $f$.
\end{theorem}




\begin{proof}
Let $e$ be the unique fixed point of $\phi$. For all $x \in X$, let $d(x) = \min\{ k : \phi^k( x ) = e \}$; for all $0 \le i \le K$, let $S_i = \{ x \in X: d(x) = k \}$. We note that $S_0 = \{ e \}$ and that if $x \in S_i$ for some $i \ge 1$, then $\phi( x ) \in S_{i-1}$. For any $1 \le i \le K$, let $U_i = \phi( S_{i+1} ) \subseteq S_i$ and $T_i = S_i \setminus U_i$, so that $T_K = S_K$. 
The linear order is given by:
\begin{itemize}
    \item 
    $S_0 > T_1 > U_1 > \dots > T_K > U_K$ (where $A < B$ indicates $a < b$ for all $a \in A$ and $b \in B$ for any $A,B \subseteq X$);

    \item 
    $T_i$ is ordered with ``consecutive pre-images'', i.e. if $a < b \in T_i$ satisfy $\phi(a) = \phi(b)$, then $\phi(c) = \phi(a)$ for all $a < c < b$;

    \item 
    $U_i$ inherits the order of $S_{i+1}$, i.e. if $a < b \in S_{i+1}$, then $\phi(a) \le \phi(b)$.
    
\end{itemize}

Let $\Lattice{L}$ be the lattice on $X$ induced by that linear order. We now prove that $\phi$ is monotone. Suppose $a < b$. If $d(a) > d(b)$, then either $d( \phi(a) ) > d( \phi(b) )$ and hence $\phi(a) < \phi(b)$, or $b = e$ and hence $\phi(a) = \phi(b) = e$. If $d(a) = d(b)$, then $a,b \in S_{i+1}$ for some $i \ge 0$, and hence by definition $\phi(a) \le \phi(b)$.

Since $\phi$ is monotone and fixes $e = \one$, and since $\Lattice{L}$ is a finite chain, $\phi$ is a co-graph by Proposition \ref{proposition:classification_L}. Let $f = \Residuated{ \phi }$. Since $\phi$ is a co-graph, $f$ is a graph by Theorem \ref{theorem:residuated}\ref{item:g_co-graph}, $f$ is $K$-meet-nilpotent by Lemma \ref{lemma:nilpotent}\ref{item:nilpotent_plus}, and $\phi$ pre-meet-depends on $f$ by Figure \ref{figure:dependencies1}.
\end{proof}



Below is an example of a nilpotent mapping $\phi$ that does not satisfy the consequence of the original Robert's theorem (and hence does not satisfy the hypothesis either). Indeed, $\phi$ is $(q-1)$-nilpotent, even though $|V| = 2$. However, $\phi$ pre-meet-depends on a graph $f$ such that $f^{q-1}( \one ) = \zero$, and hence $\phi^{q-1}( x ) = e$ for all $x$. Obviously, this example can be vastly generalised to yield nilpotence with arbitrarily long trajectories.

\begin{example} \label{example:nilpotent_reflexive_interaction_graph}
We consider $\Lattice{L} = \prod_{v \in V} \Lattice{L}_v$ where $V = \{1, 2\}$ and $\Lattice{L}_1 = \Lattice{L}_2$ is the chain $0 < 1 < \dots < q-1$. For any $a \in \{0, \dots, q-1\}$, let $a^+ = \min\{ q-1, a+1 \}$ and $a^- = \max\{ 0, a-1 \}$. Let $x = (x_1, x_2) \in L$ and define the automata network $\phi$ by
\[
    \phi( x_1, x_2 ) = ( x_1^+, x_2^+ ).
\]
Then $\phi$ is $(q-1)$-join-nilpotent but not $K$-join-nilpotent for any $K < q-1$. Note that the interaction graph of $\phi$ is reflexive. However, since $\phi$ is a co-graph, it pre-meet-depends on the graph $f = \Residuated{ \phi }$, given explicitly by
\[
    f( x_1, x_2 ) = ( x_1^-, x_2^- ).
\]

This example is illustrated for $q = 4$ in Figure \ref{figure:nilpotent_reflexive_interaction_graph}.
\end{example}

\begin{figure}[!htp]
\centering
\subfloat[$\Lattice{L}$]{
\begin{tikzpicture}[xscale=1.3]
    \node (00) at (0,0) {$00$};
    
    \node (10) at (-0.5,1) {$10$};
    \node (01) at (0.5,1) {$01$};
    
    \node (20) at (-1,2) {$20$};
    \node (11) at (0,2) {$11$};
    \node (02) at (1,2) {$02$};
    
    \node (30) at (-1.5,3) {$30$};
    \node (21) at (-0.5,3) {$21$};
    \node (12) at (0.5,3) {$12$};
    \node (03) at (1.5,3) {$03$};

    \node (31) at (-1,4) {$31$};
    \node (22) at (0,4) {$22$};
    \node (13) at (1,4) {$13$};

    \node (32) at (-0.5,5) {$32$};
    \node (23) at (0.5,5) {$23$};

    \node (33) at (0,6) {$33$};

    \draw[thick] (00) -- (01);
    \draw[thick] (00) -- (10);
    
    \draw[thick] (10) -- (20);
    \draw[thick] (10) -- (11);
    \draw[thick] (01) -- (11);
    \draw[thick] (01) -- (02);
    
    \draw[thick] (20) -- (30);
    \draw[thick] (20) -- (21);
    \draw[thick] (11) -- (21);
    \draw[thick] (11) -- (12);
    \draw[thick] (02) -- (12);
    \draw[thick] (02) -- (03);

    \draw[thick] (30) -- (31);
    \draw[thick] (21) -- (31);
    \draw[thick] (21) -- (22);
    \draw[thick] (12) -- (22);
    \draw[thick] (12) -- (13);
    \draw[thick] (03) -- (13);

    \draw[thick] (31) -- (32);
    \draw[thick] (22) -- (32);
    \draw[thick] (22) -- (23);
    \draw[thick] (13) -- (23);

    \draw[thick] (32) -- (33);
    \draw[thick] (23) -- (33);
\end{tikzpicture}
}
\hspace{1cm}\subfloat[\textcolor{red}{$\phi$}]{
\begin{tikzpicture}[xscale=1.3]
    \node (00) at (0,0) {$00$};
    
    \node (10) at (-0.5,1) {$10$};
    \node (01) at (0.5,1) {$01$};
    
    \node (20) at (-1,2) {$20$};
    \node (11) at (0,2) {$11$};
    \node (02) at (1,2) {$02$};
    
    \node (30) at (-1.5,3) {$30$};
    \node (21) at (-0.5,3) {$21$};
    \node (12) at (0.5,3) {$12$};
    \node (03) at (1.5,3) {$03$};

    \node (31) at (-1,4) {$31$};
    \node (22) at (0,4) {$22$};
    \node (13) at (1,4) {$13$};

    \node (32) at (-0.5,5) {$32$};
    \node (23) at (0.5,5) {$23$};

    \node (33) at (0,6) {$33$};

    \draw (00) -- (01);
    \draw (00) -- (10);
    
    \draw (10) -- (20);
    \draw (10) -- (11);
    \draw (01) -- (11);
    \draw (01) -- (02);
    
    \draw (20) -- (30);
    \draw (20) -- (21);
    \draw (11) -- (21);
    \draw (11) -- (12);
    \draw (02) -- (12);
    \draw (02) -- (03);

    \draw (30) -- (31);
    \draw (21) -- (31);
    \draw (21) -- (22);
    \draw (12) -- (22);
    \draw (12) -- (13);
    \draw (03) -- (13);

    \draw (31) -- (32);
    \draw (22) -- (32);
    \draw (22) -- (23);
    \draw (13) -- (23);

    \draw (32) -- (33);
    \draw (23) -- (33);

    \draw[-latex, red, thick] (00) -- (11);
    \draw[-latex, red, thick] (10) -- (21);
    \draw[-latex, red, thick] (01) -- (12);
    \draw[-latex, red, thick] (20) -- (31);
    \draw[-latex, red, thick] (11) -- (22);
    \draw[-latex, red, thick] (02) -- (13);
    \draw[-latex, red, thick] (30) to [bend left=15] (31);
    \draw[-latex, red, thick] (21) -- (32);
    \draw[-latex, red, thick] (12) -- (23);
    \draw[-latex, red, thick] (03) to [bend right=15] (13);
    \draw[-latex, red, thick] (31) to [bend left=15] (32);
    \draw[-latex, red, thick] (22) -- (33);
    \draw[-latex, red, thick] (13) to [bend right=15] (23);
    \draw[-latex, red, thick] (32) to [bend left=15] (33);
    \draw[-latex, red, thick] (23) to [bend right=15] (33);
    \draw[-latex, red, thick] (33) to [loop left] (33);
\end{tikzpicture}
}
\hspace{1cm}\subfloat[\textcolor{blue}{$f$}]{
\begin{tikzpicture}[xscale=1.3]
    \node (00) at (0,0) {$00$};
    
    \node (10) at (-0.5,1) {$10$};
    \node (01) at (0.5,1) {$01$};
    
    \node (20) at (-1,2) {$20$};
    \node (11) at (0,2) {$11$};
    \node (02) at (1,2) {$02$};
    
    \node (30) at (-1.5,3) {$30$};
    \node (21) at (-0.5,3) {$21$};
    \node (12) at (0.5,3) {$12$};
    \node (03) at (1.5,3) {$03$};

    \node (31) at (-1,4) {$31$};
    \node (22) at (0,4) {$22$};
    \node (13) at (1,4) {$13$};

    \node (32) at (-0.5,5) {$32$};
    \node (23) at (0.5,5) {$23$};

    \node (33) at (0,6) {$33$};

    \draw (00) -- (01);
    \draw (00) -- (10);
    
    \draw (10) -- (20);
    \draw (10) -- (11);
    \draw (01) -- (11);
    \draw (01) -- (02);
    
    \draw (20) -- (30);
    \draw (20) -- (21);
    \draw (11) -- (21);
    \draw (11) -- (12);
    \draw (02) -- (12);
    \draw (02) -- (03);

    \draw (30) -- (31);
    \draw (21) -- (31);
    \draw (21) -- (22);
    \draw (12) -- (22);
    \draw (12) -- (13);
    \draw (03) -- (13);

    \draw (31) -- (32);
    \draw (22) -- (32);
    \draw (22) -- (23);
    \draw (13) -- (23);

    \draw (32) -- (33);
    \draw (23) -- (33);

    \draw[-latex, blue, thick] (00) to [loop left]  (00);
    \draw[-latex, blue, thick] (10) to [bend right=15] (00);
    \draw[-latex, blue, thick] (01) to [bend left=15] (00);
    \draw[-latex, blue, thick] (20) to [bend right=15] (10);
    \draw[-latex, blue, thick] (11) -- (00);
    \draw[-latex, blue, thick] (02) to [bend left=15] (01);
    \draw[-latex, blue, thick] (30) to [bend right=15] (20);
    \draw[-latex, blue, thick] (21) -- (10);
    \draw[-latex, blue, thick] (12) -- (01);
    \draw[-latex, blue, thick] (03) to [bend left=15] (02);
    \draw[-latex, blue, thick] (31) -- (20);
    \draw[-latex, blue, thick] (22) -- (11);
    \draw[-latex, blue, thick] (13) -- (02);
    \draw[-latex, blue, thick] (32) -- (21);
    \draw[-latex, blue, thick] (23) -- (12);
    \draw[-latex, blue, thick] (33) -- (22);
\end{tikzpicture}
}
\caption{A nilpotent automata network whose convergence cannot be proved via Robert's original theorem.}
\label{figure:nilpotent_reflexive_interaction_graph}
\end{figure}
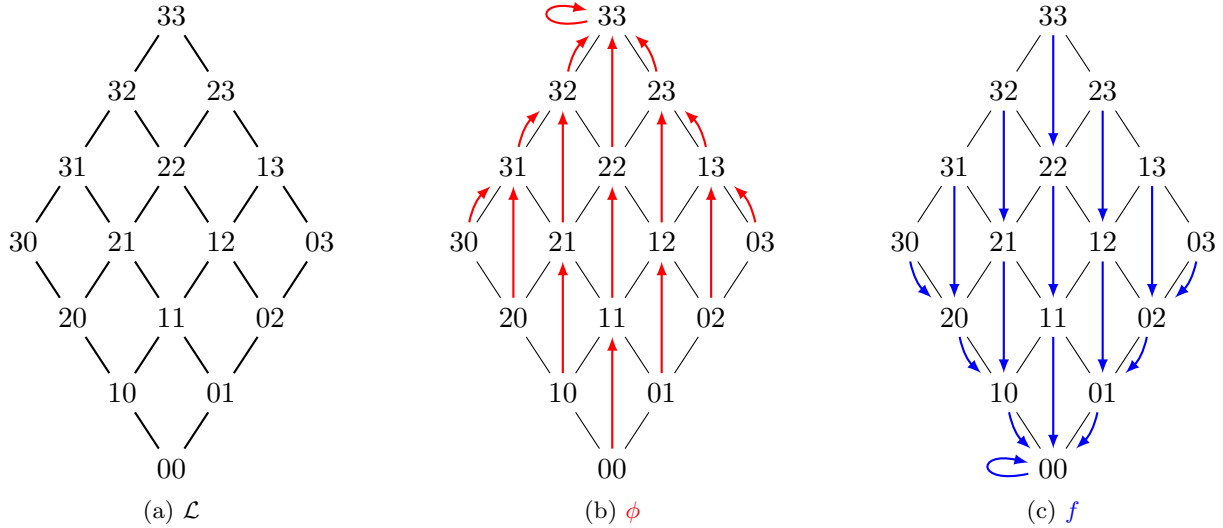



Finally, the original Robert's theorem for finite automata networks can be recast as follows: If there exists a sequence $S_0 = \emptyset, S_1, \dots, S_K = V$ of subsets of $V$ such that
\[
    x_{S_i} = y_{S_i} \implies \phi( x )_{S_{i+1}} = \phi( y )_{S_{i+1}}
\]
for all $i \le K-1$, then $\phi^K = e$, its unique fixed point. (In Robert's theorem, $S_i = { \Residual{G^\phi} }^i( \emptyset )$.) This idea can be easily generalised to any complete lattice and any finite sequence beginning at $\zero$ and finishing at $\one$.

\begin{proposition} \label{proposition:sequence_nilpotent}
Let $\Lattice{L}$ be a complete lattice. If there exists a sequence $A = (a_0 = \zero, a_1, \dots, a_K = \one)$ of elements of $L$ such that 
\[
    x \meet a_i = y \meet a_i \implies \phi( x ) \meet a_{i+1} = \phi( y ) \meet a_{i+1}
\]
for all $0 \le i \le K-1$, then $\phi^K = e$ with $\Fix( \phi ) = \{ e \}$.
\end{proposition}

\begin{proof}
The proof is by straightforward induction. Alternatively, define $\alpha: L \to L$ by
\[
    \alpha(x) = \begin{cases}
    a_i &\text{if } x = a_{i+1}, i < K\\
    \one &\text{otherwise}.
    \end{cases}
\]
Then $\phi$ pre-meet-depends on $\alpha$, and $\alpha$ is $K$-meet-nilpotent, hence $\phi^K = e$.
\end{proof}

\section{Conclusion} \label{section:conclusion}

This paper made two main contributions: the generalisation of a ``graph theory'' for mappings on complete lattices, and the generalisation of Robert's theorem in the same context. As such, this work can be expanded in two major ways.

Firstly, one could further develop this sort of graph theory for complete lattices. For instance, one could define further concepts from digraphs, such as dominating set ($x \in L$ such that $\Reflexive{f}(x) = \one$), total dominating set ($f(x) = \one$), independent set ($f(x) \meet x = \zero$), etc. Moreover, we saw that induced subgraphs could be defined for frames, while transpose graphs could be defined for complete Boolean algebras. In general, different classes of complete lattices inherit different concepts from digraphs, and hence yield distinct graph theories with various degrees of richness.

Secondly, we have given bounds on fixed points (and even pyramidal points) of automata networks. However, we know that the interaction graph has a strong influence on other dynamical properties of finite automata networks, such as the rank (see \cite{Gad20}). It would be very interesting to generalise the results on finite automata networks based on the interaction graph to our setting. Moreover, an automata network also has a signed interaction graph, which indicates whether an influence is positive (excitation) or negative (inhibition) or both (this is typically the case in linear networks). The signed interaction graph has a strong influence on the dynamics of finite automata networks, see for example \cite{ADG04, Ara08, Ric09, R10, ARS17, Ric18}. Our framework based on lattices does not seem appropriate to handle signed interaction graphs, especially the case where one entity can simultaneously excite a second entity and inhibit a third one. As such, developing a different or more general framework to encompass the signed interaction graph is an important avenue for future work in this area.


\begin{thebibliography}{10}

\bibitem{Ara08}
Julio Aracena.
\newblock Maximum number of fixed points in regulatory {B}oolean networks.
\newblock {\em Bulletin of Mathematical Biology}, 70:1398--1409, 2008.

\bibitem{ADG04}
Julio Aracena, Jacques Demongeot, and Eric Goles.
\newblock Positive and negative circuits in discrete neural networks.
\newblock {\em IEEE Transactions on Neural Networks}, 15(1):77--83, January
  2004.

\bibitem{ARS17}
Julio Aracena, Adrien Richard, and Lilian Salinas.
\newblock Number of fixed points and disjoint cycles in monotone boolean
  networks.
\newblock {\em SIAM Journal on Discrete mathematics}, 31:1702--1725, 2017.

\bibitem{Bes59}
C.~Bessaga.
\newblock On the converse of the banach fixed-point principle.
\newblock {\em Colloq. Math.}, 7:41--43, 1959.

\bibitem{Bly05}
T.~S. Blyth.
\newblock {\em Lattices and Ordered Algebraic Structures}.
\newblock Springer, 2005.

\bibitem{BJ72}
T.~S. Blyth and M.~F. Janowitz.
\newblock {\em Residuation Theory}.
\newblock Pergamon Press Ltd., 1972.

\bibitem{Cam99}
Peter~J. Cameron.
\newblock {\em Permutation Groups}, volume~45 of {\em London Mathematical
  Society Student Texts}.
\newblock Cambridge University Press, 1999.

\bibitem{Gad13}
Maximilien Gadouleau.
\newblock Closure solvability for network coding and secret sharing.
\newblock {\em IEEE Transactions on Information Theory}, 59(12):7858--7869,
  December 2013.

\bibitem{Gad18a}
Maximilien Gadouleau.
\newblock Finite dynamical systems, hat games, and coding theory.
\newblock {\em SIAM Journal on Discrete Mathematics}, 32(3):1922--1945, August
  2018.

\bibitem{Gad20}
Maximilien Gadouleau.
\newblock On the influence of the interaction graph on a finite dynamical
  system.
\newblock {\em Natural Computing}, 19:15--28, 2020.

\bibitem{Gad21}
Maximilien Gadouleau.
\newblock Dynamical properties of disjunctive {B}oolean networks.
\newblock In {\em Proc. International Workshop on Cellular Automata and
  Discrete Complex Systems}, pages 1--15, July 2021.

\bibitem{GG15}
Maximilien Gadouleau and Nicholas Georgiou.
\newblock New constructions and bounds for {W}inkler's hat game.
\newblock {\em SIAM Journal of Discrete Mathematics}, 29:823--834, 2015.

\bibitem{GRF16}
Maximilien Gadouleau, Adrien Richard, and Eric Fanchon.
\newblock Reduction and fixed points of {B}oolean networks and linear network
  coding solvability.
\newblock {\em IEEE Transactions on Information Theory}, 62(5):2504--2519,
  2016.

\bibitem{GRR15}
Maximilien Gadouleau, Adrien Richard, and S\o{}ren Riis.
\newblock Fixed points of {B}oolean networks, guessing graphs, and coding
  theory.
\newblock {\em SIAM Journal on Discrete Mathematics}, 29(4):2312--2335, 2015.

\bibitem{GR11}
Maximilien Gadouleau and S\o{}ren Riis.
\newblock Graph-theoretical constructions for graph entropy and network coding
  based communications.
\newblock {\em IEEE Transactions on Information Theory}, 57(10):6703--6717,
  October 2011.

\bibitem{GH09}
Steven Givant and Paul Halmos.
\newblock {\em Introduction to Boolean Algebras}.
\newblock Springer, 2009.

\bibitem{GD03}
Andrzej Granas and James Dugundji.
\newblock {\em Fixed Point Theory}.
\newblock Springer, 2003.

\bibitem{Gra13}
George Gr\"atzer.
\newblock {\em General Lattice Theory: Second Edition}.
\newblock Birkh\"auser, 2013.

\bibitem{Jac00}
Jacek Jachymski.
\newblock A short proof of the converse to the contraction principle and some
  related results.
\newblock {\em Topological Methods in Nonlinear Analysis}, 15:179--186, 2000.

\bibitem{RRM13}
Landon Rabern, Brian Rabern, and Matthew Macauley.
\newblock Dangerous reference graphs and semantic paradoxes.
\newblock {\em Journal of Philosophical Logic}, 42(5):727--765, 2013.

\bibitem{Ran17}
Francesco Ranzato.
\newblock A new characterization of complete {H}eyting and co-{H}eyting
  algebras.
\newblock {\em Logical Methods in Computer Science}, 13(3:25):1--11, 2017.

\bibitem{RR13}
A.~Richard and P.~Ruet.
\newblock From kernels in directed graphs to fixed points and negative cycles
  in boolean networks.
\newblock {\em Discrete Applied Mathematics}, 161(7):1106--1117, 2013.

\bibitem{Ric09}
Adrien Richard.
\newblock Positive circuits and maximal number of fixed points in discrete
  dynamical systems.
\newblock {\em Discrete Applied Mathematics}, 157:3281--3288, 2009.

\bibitem{R10}
Adrien Richard.
\newblock Negative circuits and sustained oscillations in asynchronous automata
  networks.
\newblock {\em Advances in Applied Mathematics}, 44(4):378 -- 392, 2010.

\bibitem{Ric18}
Adrien Richard.
\newblock Fixed points and connections between positive and negative cycles in
  boolean networks.
\newblock {\em Discrete Applied Mathematics}, to appear.

\bibitem{Rii07a}
S\o{}ren Riis.
\newblock Graph entropy, network coding and guessing games.
\newblock {\em ArXiv}, November 2007.

\bibitem{Rii07}
S\o{}ren Riis.
\newblock Information flows, graphs and their guessing numbers.
\newblock {\em The Electronic Journal of Combinatorics}, 14:1--17, 2007.

\bibitem{Rob80}
Fran\c{c}ois Robert.
\newblock Iterations sur des ensembles finis et automates cellulaires
  contractants.
\newblock {\em Linear Algebra and its Applications}, 29:393--412, 1980.

\bibitem{Rob95}
Fran\c{c}ois Robert.
\newblock {\em Les Syst\`emes Dynamiques Discrets}.
\newblock Springer, 1995.

\bibitem{Yab93}
Steven Yablo.
\newblock Paradox without self-reference.
\newblock {\em Analysis}, 53(4):251--252, 1993.

\end{thebibliography}

\end{document}